\pdfoutput=1
\documentclass[
affil-it, % affiliations in italic
auth-lg, % authors in large
british, % dates in British format
compression, % compressed title page
contents, % show table of contents
custom-packages={}, % Load listed custom packages
references={bigone}, % Name of bibliography
]{lib/preprint}

\newcommand{\boundary}{\flat}
\newcommand{\grp}[1]{\underline{#1}}
\renewcommand{\Re}{\mathrm{Re}}
\renewcommand{\Im}{\mathrm{Im}}
\newcommand{\coker}{\mathrm{coker}}
\renewcommand{\pr}{\mathrm{pr}}

\begin{document}
    
    % Hypersetup
    \hypersetup{
        pdftitle = {Explicit Non-Abelian Gerbes with Connections},
        pdfauthor = {Dominik Rist,Christian Saemann,Martin Wolf},
        pdfkeywords = {},
    }
    
    % Date of preprint
    \date{\today}
    
    % All emails in the order of appearance
    \email{dr.dominik.rist@gmail.com,c.saemann@hw.ac.uk,m.wolf@surrey.ac.uk}
    
    % All preprint numbers in the order of appearance
    \preprint{EMPG--22--02,DMUS--MP--22/01}
    
    % Title
    \title{Explicit Non-Abelian Gerbes with Connections} 
    
    % All authors
    \author[a]{Dominik~Rist\,\orcidlink{0000-0002-1817-3458}\,}
    \author[a]{Christian~Saemann\,\orcidlink{0000-0002-5273-3359}\,}
    \author[b]{Martin~Wolf\,\orcidlink{0009-0002-8192-3124}\,}
    
    % All affiliations
    \affil[a]{Maxwell Institute for Mathematical Sciences, Department of Mathematics,\\ Heriot--Watt University, Edinburgh EH14 4AS, United Kingdom}
    \affil[b]{School of Mathematics and Physics,\\ University of Surrey, Guildford GU2 7XH, United Kingdom}
    
    % Abstract
    \abstract{We define the notion of adjustment for strict Lie 2-groups and provide the complete cocycle description for non-Abelian gerbes with connections whose structure 2-group is an adjusted 2-group. Most importantly, we depart from the common fake-flat connections and employ adjusted connections. This is an important generalisation that is needed for physical applications especially in the context of supergravity. We give a number of explicit examples; in particular, we lift the spin structure on $S^4$, corresponding to an instanton--anti-instanton pair, to a string structure, a 2-group bundle with connection. We also outline how categorified forms of Bogomolny monopoles known as self-dual strings can be obtained via a Penrose--Ward transform of string bundles over twistor space.}
    
    % Acknowledgements
    \acknowledgements{We thank Hyungrok Kim and Paul Skerritt for fruitful discussions. We are particularly grateful to David Roberts for his contribution to the `Workshop on Higher Gauge Theory and Higher Quantization' in 2014 and for a number of helpful  discussions related to signs in the literature. We are also grateful to the anonymous referees, whose comments have led to many improvements.}
    
    % Declarations
    \declarations{
        \textbf{Funding.}
        The work of CS was partially supported by the Leverhulme Research Project Grant RPG-2018-329 `The Mathematics of M5-Branes'.\\[5pt]
        \textbf{Conflict of interest.}
        The authors have no relevant financial or non-financial interests to disclose.\\[5pt]
        \textbf{Data statement.}
        No additional research data beyond the data presented and cited in this work are needed to validate the research findings in this work.\\[5pt]
        \textbf{Licence statement.}
        For the purpose of open access, the authors have applied a Creative Commons Attribution (CC-BY) license to any author-accepted manuscript version arising.
    }
    
    % Body
    \begin{body}
        
        \section{Introduction and results}   
        
        Higher forms arise as local connection forms or gauge potentials in a number of contexts within physics. The prime example is the Kalb--Ramond $B$-field of string theory, which is then found again in the low-energy supergravity limits. A whole family of higher form gauge potentials arises in the tensor hierarchies of gauged supergravity theories and, consequentially, double and exceptional field theory. Mathematically, these gauge potentials are connections on higher or categorified principal bundles also known as gerbes. This perspective is evidently crucial when supergravity theories are to be considered on topologically non-trivial spaces. 
        
        In the Abelian case, the theory of gerbes is well-established and used in many contexts. There are a number of equivalent descriptions, such as the geometrically appealing bundle gerbes~\cite{Murray:9407015,Murray:2007ps} or the computationally useful Hitchin--Chatterjee gerbes~\cite{Hitchin:1999fh,Chatterjee:1998}. Although the Kalb--Ramond $B$-field is locally an ordinary differential 2-form, Abelian gerbes are not fully sufficient for its description. In the presence of additional gauge bundles, the connections on these mix non-trivially with the connection on the Abelian gerbe. This is the case, e.g., in gauged supergravity, and it was in this context that what are now called connections on string bundles had been observed for the first time~\cite{Bergshoeff:1981um,Chapline:1982ww} in their local, infinitesimal form. Similarly, in the context of T-duality, one often wants to describe string theory backgrounds which correspond to Abelian gerbes on the total space of principal torus bundles. These can be conveniently captured by non-Abelian generalisations of gerbes at the topological level~\cite{Nikolaus:2018qop}. A differential refinement, leading to the full Buscher rules is also possible~\cite{Kim:2022opr}. Moreover, there is reason to believe that aspects of the dynamics of multiple M5-branes can be captured by non-Abelian gerbes~\cite{Saemann:2017zpd,Saemann:2019dsl,Rist:2020uaa}; see also~\cite{Fiorenza:2020hiq} for recent evidence that the string group $\sfString(3)$ naturally emerges on a single M5-brane.\footnote{See also~\cite{Jurco:2019woz} for a general review on higher structures in M-theory.} Finally, several kinds of non-Abelian gerbes have been used in a number of string-theory inspired twistor constructions~\cite{Saemann:2012uq,Saemann:2013pca,Jurco:2014mva,Jurco:2016qwv,Samann:2017dah}.
        
        Non-Abelian gerbes have been introduced in various forms, and all familiar definitions of principal bundles have found higher generalisations, cf.~\cite{Nikolaus:2011ag} and references therein. One of the earliest forms is perhaps the differential cocycle description given in~\cite{Breen:math0106083} and~\cite{Aschieri:2003mw}, and because of its computational usefulness, we shall focus mostly on this picture in the following. This form of non-Abelian gerbes generalises the notion of a principal fibre bundle by lifting the cocycle relations for transition functions and local connection forms to hold up to homotopies, which are then encoded in higher components of the cocycle. 
        
        Whilst the topological description of non-Abelian gerbes is relatively straightforward, its differential refinement by a connection is subtle. An explicit cocycle description for such a differential refinement was obtained in~\cite{Breen:math0106083,Aschieri:2003mw}, but these cocycles contain an additional local 2-form datum\footnote{In~\cite{Breen:math0106083,Aschieri:2003mw}, this 2-form is denoted by $\delta_{ij}$.} compared to the expected $L_\infty$-algebra-valued connections~\cite{Sati:2008eg}. A detailed study of the evident notion of higher parallel transport~\cite{Baez:0511710,Schreiber:0705.0452,Schreiber:2008aa} then suggested that the additional datum is, in a sense, spurious, and that one should work instead with reduced cocycles~\cite{Baez:0511710}, which have subsequently been used throughout much of the literature. This, however, led to a second problem. Consistency of the corresponding parallel transport requires the 2-form part $F$ of the total curvature of a non-Abelian gerbe with connection, which is also known as its fake curvature, to vanish~\cite{Schreiber:0705.0452,Schreiber:2008aa}. The fake-flatness condition, however, is problematic for two reasons: firstly, we know that interesting string backgrounds do not satisfy this condition, and secondly, it can be shown that locally, a fake-flat non-Abelian gerbe with connection is always isomorphic to an Abelian gerbe with connection~\cite{Gastel:2018joi,Saemann:2019dsl}, rendering it essentially useless beyond ordinary Abelian gerbes for all but topological applications in physics. Another undesirable consequence is that after adding fake flatness to the cocycle conditions of non-Abelian gerbes with connection, ordinary principal bundles with non-flat connections are no longer trivially higher bundles. This is clearly undesirable from a category theoretical viewpoint.
        
        The solution to these problems is provided by a third possible type of cocycles which are more general than the reduced ones of~\cite{Baez:0511710}, but do not contain the additional 2-form datum of~\cite{Breen:math0106083,Aschieri:2003mw}. In these \uline{adjusted cocycles} the 3-form part of the curvature as well as the gluing relations and gauge transformations for the 2-form part of the connection are deformed by a term proportional to the fake curvature. 
        
        At a very explicit, local and infinitesimal level and in one particular case, this had been observed already in~\cite{Bergshoeff:1981um,Chapline:1982ww}. A more general mathematical perspective was then given in~\cite{Sati:2008eg,Sati:2009ic,Fiorenza:2010mh}. In this picture, connections are locally morphisms from the Weil algebra of the (higher) gauge algebra to the de~Rham complex on the relevant coordinate patch. The deformation to adjusted connections is then a coordinate change on the relevant Weil algebra. The motivation for the deformation here was a match with the Bianchi identities describing the Green--Schwarz anomaly cancellation derived previously from a physical perspective.
        
        It was then found that adjusted connections have other desirable properties, too. In particular, the BRST complex, the Chevalley--Eilenberg complex dual to the higher gauge Lie algebroid, closes, i.e.~the differential is indeed nilquadratic without further conditions on the curvature forms, contrary to the reduced cocycles of~\cite{Baez:0511710}. This property was then turned into the definition of adjusted connections in~\cite{Saemann:2019dsl}, where also the adjustments for a number of more general higher Lie algebras were computed. Shortly after, it was then shown that a consistent adjusted higher parallel transport that does not rely on fake flatness is indeed possible~\cite{Kim:2019owc}.\footnote{The discussion in~\cite{Kim:2019owc} makes it clear why the adjustment is only visible at the level of connections or parallel transport. Topologically, higher principal bundles are functors from the \v Cech groupoid to the delooping of the higher structure group, but the extension of the \v Cech groupoid to a higher groupoid is trivial in the sense that all higher morphisms are identities. Parallel transport, however, is locally defined as a functor from a higher path groupoid to the delooping of the higher structure group. In the former, the higher morphisms are non-trivial.} Moreover, the origin of the adjustment was traced back to a generalisation of the notion of higher gauge algebra in~\cite{Borsten:2021ljb}, where also adjustments for large classes of higher bundles were derived that are relevant in the context of gauged supergravity.
        
        The explicit finite form of the cocycles for non-Abelian gerbes with adjusted connection has never been provided. This may be due to the fact that the relevant adjustment data is only known for few higher gauge algebras and our understanding of how adjustments should be seen geometrically is still very incomplete. It is the primary goal of this paper to fill this gap, i.e.~to give explicit cocycle and coboundary relations for adjusted connections with a focus on string bundles. A secondary goal is to construct an explicit example.
        
        An evident higher gauge group to use in such an example is a 2-group model of the string group. Recall that a spin structure over an $n$-dimensional oriented Riemannian manifold $X$ is a lift of the frame bundle, a principal $\sfSO(n)$-bundle, to a principal $\sfSpin(n)$-bundle. The obstruction to the existence of such a lift is given by the second Stiefel--Whitney class of $X$. One can now consider spin bundles over loop spaces, as suggested by string theory, and it turns out that the loop space $LX$ of a manifold $X$ carries a spin bundle if and only if\footnote{The theorem assumes $X$ is simply connected and $n\geq5$. The `only if' part holds further requiring that $X$ be 2-connected.} spin bundles over $X$ can be lifted to principal $\sfString(n)$-bundles~\cite{mclaughlin1992orientation}. Here, $\sfString(n)$ is a topological group, defined only up to a large class of equivalence, which is a 3-connected cover of $\sfSpin(n)$, but otherwise shares the same homotopy groups as $\sfSpin(n)$. 
        
        It turns out that the group $\sfString(n)$ is conveniently described in terms of a 2-group model, cf.~\cite{Nikolaus:2011zg}, and from this perspective, string structures are indeed non-Abelian gerbes. This is also the form in which string structures most naturally arise in supergravity and string theory. There are now essentially two extreme string 2-group models. Firstly, there is an infinite-dimensional one based on Kac--Moody central extensions of loop groups, which was developed in~\cite{Baez:2005sn}. On the other hand, there is the finite-dimensional, but much more complicated and less explicit model~\cite{Schommer-Pries:0911.2483}.\footnote{See~\cite{Demessie:2016ieh} for a description of (unadjusted) non-Abelian gerbes using this group.} For convenience, we shall be focusing on the infinite-dimensional model because it admits fully explicit formulas for all cocycles.
        
        Explicitly, we have the following results to report:
        \begin{enumerate}[(i)]\itemsep-2pt
            \item In \cref{ssec:adjusted_BRST_2_groupoid}, we define an adjusted BRST 2-groupoid which integrates in particular the adjusted infinitesimal local gauge transformations known from the literature. 
            \item The definition of the BRST 2-groupoid relied on an additional datum on the gauge Lie 2-group, the adjustment. In \cref{ssec:adjustedCrossedModules}, we define the category of adjusted crossed modules of Lie groups. 
            \item An equivalence of adjusted crossed modules of Lie groups does not directly translate to an equivalence of adjusted BRST 2-groupoids, and for special adjustments, we identify an additional sufficient condition for this, called $P$-flatness, in \ref{ssec:equivalences}.
            \item We then compute the fully differential cocycle description of higher principal bundles\footnote{For us, `higher principal bundle' is synonymous with `principal 2-bundle'.} with adjusted structure 2-group as well as adjusted (and hence, generalised notion of) connection in \cref{sec:adj_cocycles} by stackification of the BRST 2-groupoid.
            \item An instructive example is obtained by replacing a Lie group $\sfG$ by the equivalent crossed module $\caL\sfG$ of path and loop groups of $\sfG$, which induces an adjustment datum as we explain in \cref{ssec:adjustedCrossedModules}. We then show that any principal bundle with connection can equivalently be described as a higher principal bundle with adjusted connection in \cref{sec:principalAdjustedHigherPrincipal}.
            \item Concretely, we present the explicit cocycles for the spin structure on $S^4$, which can be regarded as the principal $\sfSpin(4)$-bundle $\sfSpin(5)\rightarrow S^4\cong\sfSpin(5)/\sfSpin(4)$, as an $\caL\sfSpin(4)$-bundle with adjusted connection in \cref{sec:Spin4Bundle}.
            \item We show that the quotients of certain string 2-groups reduce to quotients of ordinary groups and hence lead to ordinary manifolds regarded as Lie groupoids in \cref{sec:instantonAntiInstantonStringLift}. As a simple example, we consider the space $\sfString(3)/\sfString(2)\cong S^3$, which implies the categorified Hopf fibration of $\sfString(3)$ over $S^3$.
            \item Our second major result is the concrete example of a string 2-bundle constructed in \cref{sec:instantonAntiInstantonStringLift}. Here, we lift the principal $\caL\sfSpin(4)$-bundle equivalent to the spin structure on $S^4$ to a $\sfString(4)$-bundle.\footnote{The case $\sfString(4)$ is particularly interesting from an M-theory perspective: it contains the Lie group $\sfSpin(4)\cong\sfSU(2)\times\sfSU(2)$, which is important in the context of M2-brane models~\cite{Bagger:2007jr,Gustavsson:2007vu}.} 
            \item Our formulas give the explicit finite gauge transformations for the non-Abelian self-dual strings introduced in~\cite{Saemann:2017rjm}, closing a significant gap. This is explained in \cref{ssec:nonAbelianSelfDualString}.
            \item Finally, in \cref{sec:commentsTwistorApproach}, we comment on an iterated Penrose--Ward transform that allows for a construction of solutions to the non-Abelian self-dual string equations of~\cite{Saemann:2017rjm}.
        \end{enumerate}
        An additional minor result is our clarification of the relation between two group cocycles used in the literature to describe the Kac--Moody central extension of the loop group of a group at the beginning of \cref{app:proofs}.
        
        We have tried to be fairly self-contained and detailed in our presentation, mostly for two reasons. Firstly, we had to fix a number of sign inaccuracies in the literature we used, and we therefore wanted to present our computations in a verifiable way. Secondly, we ourselves would have found the detailed review part in this paper invaluable when we started this research project.
        
        \begin{remark}
            Throughout this paper, all manifolds $X,Y,\ldots$ are assumed to be smooth so we drop the adjective `smooth'. Furthermore, $\Omega^k(X)$ are the smooth differential $k$-forms on $X$ with $\Omega^0(X)=\scC^\infty(X)$. We denote the complex of differential forms on $X$ taking values in a vector space $\frg$ by $\Omega^\bullet(X,\frg)=\Omega^\bullet(X)\otimes \frg$.
            
            We denote the unit element of a (Lie) group $\sfG,\sfH,\ldots$ by $\unit$. For simplicity, we will mostly use matrix Lie group/algebra notation. For example if $g$ is an element of a Lie group and $X$ an element of the corresponding Lie algebra, then $g^{-1}Xg$ is another element of the Lie algebra.
        \end{remark}
                
        \section{Principal 2-bundles with adjusted connections}\label{sec:higherPrincipalBundles}
        
        In this section, we start by briefly reviewing crossed modules of Lie groups and Lie algebras. We then integrate the infinitesimal (higher) gauge transformations of adjusted connections of~\cite{Saemann:2019dsl} to a BRST 2-groupoid. The stackification of this BRST 2-groupoid then yields the differential cocycles and coboundaries for principal 2-bundles with adjusted connections. The definition of the BRST 2-groupoid involves an adjustment, a purely algebraic datum on the gauge 2-group, leading to the notion of adjusted crossed module of Lie groups. We also examine how an equivalence between such crossed modules translates to an equivalence of BRST 2-groupoid. Finally, we comment on the differential cohomology for principal 2-bundles with unadjusted connections.
        
        \subsection{Crossed modules and their morphisms}
        
        \paragraph{Crossed modules of Lie groups.}
        Recall that a \uline{crossed module of Lie groups} $\caG\coloneqq(\sfH\overset{\sft}{\longrightarrow}\sfG,\acton)$ consists of two Lie groups $\sfG$ and $\sfH$, an automorphism action $\acton$ of $\sfG$ on $\sfH$, and a morphism of Lie groups $\sft:\sfH\rightarrow\sfG$ such that the conditions
        \begin{equation}\label{eq:crossedModuleConditions}
            \sft(g\acton h_1)\ =\ g\sft(h_1)g^{-1}
            \eand
            \sft(h_1)\acton h_2\ =\ h_1h_2h_1^{-1}
        \end{equation}
        hold for all $g\in\sfG$ and for all $h_{1,2}\in\sfH$; the second condition is called the \uline{Peiffer identity}. Crossed modules of Lie groups can be used to describe strict Lie 2-groups, cf.~\cite{Baez:0307200}. 
        
        Given two crossed modules of Lie groups, $\caG_{1,2}\coloneqq(\sfH_{1,2}\overset{\sft_{1,2}}{\longrightarrow}\sfG_{1,2},\acton_{1,2})$, a \uline{strict morphism} $\phi:\caG_1\rightarrow \caG_2$ is a pair of group homomorphisms $\phi_\sfG:\sfG_1\rightarrow \sfG_2$ and $\phi_\sfH:\sfH_1\rightarrow \sfH_2$ such that 
        \begin{equation}
            \begin{aligned}
                (\sft_2\circ \phi_\sfH)(h)=(\phi_\sfG\circ \sft_1)(h)
                \eand 
                \phi_\sfH(g\acton_1 h)=\phi_\sfG(g)\acton_2\phi_\sfH(h)
            \end{aligned}
        \end{equation}
        for all $g\in \sfG_1$ and $h\in \sfH_1$. A \uline{strict quasi-isomorphism} is a strict morphism that preserves the subgroup $\ker(\sft_1)\subseteq\sfH_1$ as well as the quotient group $\coker(\sft_1)=\sfG_1/\im(\sft_1)$.
        
        A weaker form of morphisms can be defined in terms of smooth flippable butterflies~\cite{Aldrovandi:0808.3627}. Specifically, given two crossed modules of Lie groups, $\caG_{1,2}=(\sfH_{1,2}\overset{\sft_{1,2}}{\longrightarrow}\sfG_{1,2},\acton_{1,2})$, a \uline{butterfly} is a commutative diagram of Lie groups of the form
        \begin{subequations}\label{eq:butterfly}
            \begin{equation}
                \begin{tikzcd}
                    \sfH_1 \arrow[dd,"\sft_1",swap] \arrow[dr,"\lambda_1"]& & \sfH_2 \arrow[dl,"\lambda_2",swap] \arrow[dd,"\sft_2"]
                    \\
                    & \sfE \arrow[dl,"\gamma_1",swap] \arrow[dr,"\gamma_2"]& 
                    \\
                    \sfG_1 & & \sfG_2
                \end{tikzcd}
            \end{equation}
            where $\sfE$ is a Lie group, $\lambda_{1,2}$ and $\gamma_{1,2}$ are morphisms of Lie groups, the NE--SW diagonal is a short exact sequence (i.e.~a Lie group extension), and the NW--SE diagonal is a complex. In addition, the actions need to compatible in the sense of
            \begin{equation}
                \lambda_{1,2}(\gamma_{1,2}(e)\acton_{1,2}h_{1,2})\ =\ e\lambda_{1,2}(h_{1,2})e^{-1} 
            \end{equation}
        \end{subequations}
        for all $e\in\sfE$ and for all $h_{1,2}\in\sfH_{1,2}$. The butterfly is called \uline{flippable} whenever both diagonals are short exact sequences. Given such a flippable butterfly, we call $\caG_1$ and $\caG_2$ \uline{equivalent}. 
        
        We note that any butterfly between two crossed modules $\caG_{1,2}\coloneqq(\sfH_{1,2}\overset{\sft}{\longrightarrow}\sfG_{1,2},\acton_{1,2})$ can be turned into a span of strict Lie 2-group morphisms~\cite[Remark 8.5]{Noohi:0506313} of the form
        \begin{equation}
            \begin{tikzcd}
                & \caK \arrow[dl,"\phi"'] \arrow[dr,"\psi"]
                \\
                \caG_1 & & \caG_2
            \end{tikzcd}
        \end{equation}
        where $\caK$ is a crossed module of the form $\sfH_1\times \sfH_2\rightarrow \sfE$, and $\phi$ and $\psi$ are strict morphisms of 2-groups with $\phi$ a quasi-isomorphism. We call such spans \uline{weak morphisms}. If both $\phi$ and $\psi$ are quasi-isomorphisms, we call $\caG_1$ and $\caG_2$ equivalent.
        
        \paragraph{Crossed modules of Lie algebras.}        
        Applying the Lie functor to a crossed module of Lie groups $\caG\coloneqq(\sfH\overset{\sft}{\longrightarrow}\sfG,\acton)$, we obtain a \uline{crossed module of Lie algebras} $\sfLie(\caG)\coloneqq(\frh\overset{\sft}{\longrightarrow}\frg,\acton)$ where $\frg$ and $\frh$ are the Lie algebras of $\sfG$ and $\sfH$, respectively.\footnote{We follow the literature and slightly abuse notation by denoting the linearisations of $\acton$ and $\sft$ again by the same symbols.} The linearisations of the conditions in~\eqref{eq:crossedModuleConditions} read as 
        \begin{equation}\label{eq:crossedModuleLieAlgebras}
            \sft(V\acton W_1)\ =\ [V,\sft(W_1)]
            \eand
            \sft(W_1)\acton W_2\ =\ [W_1,W_2]
        \end{equation}
        for all $V\in\frg$ and for all $W_{1,2}\in\frh$. 
        
        Both strict morphisms and butterflies between crossed modules of Lie groups differentiate straightforwardly to corresponding notions of morphisms of crossed modules of Lie algebras. We note that for a butterfly between crossed modules of Lie algebras $\sfLie(\caG_{1,2})\coloneqq(\frh_{1,2}\xrightarrow{~t_{1,2}~}\frg_{1,2})$,
        \begin{equation}
            \begin{tikzcd}
                \frh_1 \arrow[dd,"\sft_1",swap] \arrow[dr,"\lambda_1"]& & \frh_2 \arrow[dl,"\lambda_2",swap] \arrow[dd,"\sft_2"]
                \\
                & \fre \arrow[dl,"\gamma_1",swap] \arrow[dr,"\gamma_2"]& 
                \\
                \frg_1 & & \frg_2
            \end{tikzcd}
        \end{equation}
        with a section $\sigma:\frg_1\rightarrow \fre$ of $\gamma_1$ as morphism of vector spaces, we can construct a corresponding weak morphisms of $L_\infty$-algebras between $\sfLie(\caG_{1})$ and $\sfLie(\caG_{2})$, with both trivially regarded as strict 2-term $L_\infty$-algebras~\cite[Prop.~3.4]{Noohi:0910.1818}. Explicitly, the splitting $\frg_1$ provides us with a second map $\tau:\fre\rightarrow \frh_2$, and we have maps 
        \begin{equation}
            \begin{aligned}
                \phi_\frh&:\frh_1\rightarrow \frh_2~,~~~&\phi_\frh(W)&=-\tau(\lambda_1(W))~,
                \\
                \phi_\frg&:\frg_1\rightarrow \frg_2~,~~~&\phi_\frg(V)&=\gamma_2(\sigma(V))~,
                \\
                \phi_2&:\frg_1\wedge \frg_1\rightarrow \frh_2~,~~~&\phi_2(V_1,V_2)&=-\tau([\sigma(V_1),\sigma(V_2)]_\fre)
            \end{aligned}
        \end{equation}
        for all $V,V_{1,2}\in\frg_1$ and for all $W\in \frh_1$ defining a weak morphism of $L_\infty$-algebras between $\sfLie(\caG_{1})$ and $\sfLie(\caG_{2})$. That is, 
        \begin{equation}\label{eq:L_infty_alg}
            \begin{aligned}
                \sft_2(\phi_\frh(W))\ &=\ \phi_\frg (\sft_1(W))~,
                \\
                \phi_\frg([V_1,V_2]_{\frg_1})\ &=\ [\phi_\frg(V_1),\phi_\frg(V_2)]_{\frg_2}+\sft_2(\phi_2(V_1,V_2))~,
                \\
                \phi_\frh(V_1\acton_1 W)\ &=\ \phi_\frg(V_1)\acton_2 \phi_\frh(W)+\phi_2(V_1,\sft_1(W))~,
                \\
                0\ &=\ \phi_2(V_1,[V_2,V_3])+\phi_\frg(V_1)\acton_2\phi_2(V_2,V_3)+\mbox{cycl.}
            \end{aligned}
        \end{equation}
        for all $V_{1,2}\in \frg$ and $W\in \frh$; see e.g.~\cite{Jurco:2018sby,Jurco:2019bvp} for more details. Conversely, any weak morphism of $L_\infty$-algebras can be turned into a butterfly. 
        
        \subsection{Adjusted BRST 2-groupoid for local connections}\label{ssec:adjusted_BRST_2_groupoid}
        
        \paragraph{Generalities.} 
        A surprising feature of the straightforward definition of connections on higher principal bundles is that (higher) gauge transformations generally do not glue together consistently. This phenomenon is well-known in physics, particularly in the local and infinitesimal setting, see e.g.~the discussion in~\cite[Section 4.1]{Jurco:2018sby} or~\cite{Saemann:2019dsl}.         
        
        \paragraph{Unadjusted BRST complex.} Local connection forms on principal bundles, the Lie algebra of local and infinitesimal gauge transformations, and the actions of the latter on the former combine into an action algebroid. The Chevalley--Eilenberg algebra of this action algebroid is a semi-free differential graded commutative algebra with an underlying cochain complex, which physicists usually refer to as the \uline{BRST complex}. Our first aim is to construct the analogous structure for local connection forms on principal 2-bundles.
        
        Given a crossed module of Lie groups $\caG$, we can construct the local and infinitesimal gauge transformations of local connection forms on a contractible manifold $U$ taking values in the crossed module of Lie algebras $\sfLie(\caG)\coloneqq(\frh\overset{\sft}{\longrightarrow}\frg,\acton)$. Concretely, this can be done by constructing the Weil algebra $\sfW(\sfLie(\caG))$ of $\sfLie(\caG)$ following~\cite{Sati:2008eg}, identifying connections with morphisms of differential graded algebras $\sfW(\sfLie(\caG))\rightarrow \Omega^\bullet(U)$. Gauge transformations are then partially flat homotopies between these morphisms, and higher gauge transformations correspond to higher homotopies.
        
        Instead of considering partially flat homotopies, we can also construct the BRST complex directly, by considering the inner homomorphisms from $T[1]U$ to\footnote{Recall that $\sfLie(\caG)$ can be identified with a strict $L_\infty$-algebra, which naturally has the structure of a differential graded (dg-) manifold. Moreover, the shifted tangent bundle of any dg-manifold is again a dg-manifold.}
        $T[1]\sfLie(\caG)$ in the category of N$Q$-manifolds. The BRST complex is obtained by restricting these inner homomorphisms in a straightforward manner, and the details are found in~\cite{Saemann:2019dsl}.
        
        The result of the above construction is the well-established form of gauge transformations for $\sfLie(\caG)$-valued local connection forms. In particular, the connection is given by
        \begin{equation}
            A\in \Omega^1(U, \frg)~,~~~B\in \Omega^2(U,\frh)~,
        \end{equation}
        and infinitesimal gauge transformations are parameterised by $\alpha=(\alpha_0,\alpha_1)$ with
        \begin{subequations}
            \begin{equation}
                \alpha_0\in \Omega^0(U,\frg)~,~~\alpha_1\in \Omega^1(U,\frh)
            \end{equation} 
            and act according to
            \begin{equation}
                \begin{aligned}
                    \delta_\alpha A\ =\ \rmd \alpha_0+[A,\alpha_0]-\sft(\alpha_1)~,~~~
                    \delta_\alpha B\ =\ \rmd \alpha_1+A\acton \alpha_1-\alpha_0\acton B~.
                \end{aligned}
            \end{equation}
        \end{subequations}
        One easily verifies that, as expected, the commutator of two gauge transformations is again a gauge transformation.
        
        Higher gauge transformations, however, are problematic. Using the above machinery, we find that they are parameterised by
        \begin{subequations}
            \begin{equation}
                \beta \in \Omega^0(U,\frh)
            \end{equation}
            and act according to
            \begin{equation}
                \begin{aligned}
                    \delta_\beta \alpha_0\ =\ \sft(\beta)~,~~~
                    \delta_\beta \alpha_1\ =\ \rmd\beta+A\acton \beta~.
                \end{aligned}
            \end{equation}
        \end{subequations}
        The problem is that the gauge transformations parameterised by $\alpha=(\alpha_0,\alpha_1)$ and $\alpha'=(\alpha_0+\delta_\beta \alpha_0,\alpha_1+\delta_\beta \alpha_1)$ yield different results, in particular
        \begin{equation}
            \delta_{\alpha'} B-\delta_\alpha B=F\acton \beta~,
        \end{equation} 
        where
        \begin{equation}
            F\coloneqq \rmd A +\tfrac12[A,A]+\sft(B)
        \end{equation} 
        is the so-called \uline{fake curvature} of the connection. We thus see that higher gauge transformations act appropriately if $F=0$, and this is the familiar fake curvature condition arising often in the context of higher connections.
        
        Physicists would say that the BRST complex is open, and the BRST differential, the Chevalley--Eilenberg differential of the corresponding higher Lie algebroid squares to terms proportionally to an assumed equation of motion $F=0$. This is physically undesirable for a number of reasons as explained in the introduction. From a mathematical perspective, we only have the structure of an action Lie 2-algebroid if we restrict to local connection forms with $F=0$.
        
        \paragraph{Adjusted BRST complex and adjusted BRST Lie 2-algebroid.} As shown in~\cite{Saemann:2019dsl}, this problem can be fixed by deforming the Weil algebra, as also previously suggested in~\cite{Sati:2008eg} for different reasons. This deformation then induces a deformation of the BRST complex derived from the deformed Weil algebra.
        
        The deformation can also be derived directly. We wish to a) preserve the infinitesimal gauge transformations of flat connections, b) preserve trivial infinitesimal gauge transformations, and c) have only analytic expressions in the connection forms and their derivatives appear in (higher) gauge transformations.
        
        \begin{proposition}
            The unique such deformation are the adjusted (higher) gauge transformations given by
            \begin{equation}\label{eq:infinitesimal_deformation}
                \begin{aligned}
                    \delta_\alpha A\ &=\ \rmd \alpha_0+[A,\alpha_0]-\sft(\alpha_1)~,~~~&
                    \delta_\alpha B\ &=\ \rmd \alpha_1+A\acton \alpha_1-\alpha_0\acton B+\kappa(\alpha_0,F)~,
                    \\
                    \delta_\beta \alpha_0\ &=\ \sft(\beta)~,~~~
                    &\delta_\beta \alpha_1\ &=\ \rmd\beta+A\acton \beta~,
                \end{aligned}
            \end{equation}
            where $\kappa$ is a bilinear function $\kappa:\frg\times \frg\rightarrow \frh$ (extended in the evident way to forms taking values in $\frg$ and $\frh$) that satisfies 
            \begin{equation}\label{eq:linear_kappa_conditions}
                \begin{aligned}            
                    \kappa(\sft(W),V)\ &=\ -V\acton W~,
                    \\
                    \kappa([V_1,V_2],V_3)\ &=\ V_1\acton \kappa(V_2,V_3)-V_2\acton \kappa(V_1,V_3)+\kappa(V_1,[V_2,V_3])-\kappa(V_2,[V_1,V_3])
                    \\
                    &\hspace{1cm}-\kappa(V_1,\sft(\kappa(V_2,V_3)))+\kappa(V_2,\sft(\kappa(V_1,V_3)))
                \end{aligned}
            \end{equation}
            for all $V_{1,2}\in\frg$ and $W\in\frh$.         
        \end{proposition}
        \begin{proof}
            Since $F$ is a 2-form, we evidently can only deform the action of gauge transformation for $B$. It is then easy to see that the most general such deformation meeting our requirements is indeed $\kappa(\alpha_0,F)$. The rest is a straightforward computation. The first condition in~\eqref{eq:linear_kappa_conditions} ensures that infinitesimal adjusted higher gauge transformations link infinitesimal adjusted gauge transformations with the same image. The second condition ensures that the commutator of two infinitesimal adjusted gauge transformations is again an infinitesimal adjusted gauge transformations.
        \end{proof}
        
        The datum $\kappa$ was called an \uline{adjustment} in a much wider context of $L_\infty$-algebra valued connection forms in~\cite[Definition 4.2]{Saemann:2019dsl}. 
        
        The fact that all structures behave as required, i.e.~the (higher) gauge transformations combine into a Lie 2-algebra act appropriately, immediately imply the following, see also~\cite{Saemann:2019dsl}:
        \begin{corollary}
            There is an action Lie 2-algebroid, the \uline{adjusted BRST Lie 2-algebroid}, combining local connection forms with their local and infinitesimal gauge and higher transformations.
        \end{corollary}
        \noindent For our discussion, it suffices to know the (faithful) action of (higher) gauge transformations. For the interested reader, we note that the adjusted BRST Lie 2-algebroid has a base space given by the local connection forms, while the fibers are $\IZ_2$-graded vector spaces describing (local and infinitesimal) gauge transformations (degree 0 elements) and higher gauge transformations (degree 1 elements). There is an anchor map, which describes the action of the gauge transformations on the connection forms, and a graded Lie bracket on sections, which describes the action of higher gauge transformations, which is part of the Lie 2-algebra of (higher) gauge transformations.
        
        \paragraph{Adjusted BRST 2-groupoid.} We now wish to integrate the above BRST Lie 2-algebroid to obtain the corresponding gauge or BRST 2-groupoid $\scB$ for the gauge 2-group given by the crossed module of Lie groups $\caG\coloneqq(\sfH\overset{\sft}{\longrightarrow}\sfG,\acton)$. The groupoid $\scB$ then allows us to construct the differential cocycle data of principal 2-bundles with connection by stackification.
        
        The finite form of the unadjusted gauge transformations are parameterised by 
        \begin{subequations}
            \begin{equation}
                a\in \Omega^0(U,\sfG)~,~~~\lambda \in\Omega^1(U,\frh)~,
            \end{equation}
            and these act according to\footnote{Here and in the following, we encounter full and partial linearisations of various maps in the crossed module of Lie groups $\caG$, which are well-defined. For example, the expression $a^{-1}\acton B$ involves a map $\acton:\sfG\times \frh\rightarrow \frh$ obtained by linearising $\acton:\sfG\times \sfH\rightarrow \sfH$ in the second argument. A bit more subtle is the expression $m^{-1}(A\acton m)$ appearing in~\eqref{eq:unadjustedHigherGaugeTransformationsB} in the term $m^{-1}\nabla m$, which combines a half-linearisation of the action with a pullback to produce an element of the Lie algebra $\frh$.}
            \begin{equation}
                \begin{aligned}
                    \tilde A\ &=\ a^{-1}Aa+a^{-1}\rmd a-\sft(\lambda)~,
                    \\
                    \tilde B\ &=\ a^{-1}\acton B+\rmd\lambda+\tilde{A}\acton\lambda+\tfrac12[\lambda,\lambda]~.
                \end{aligned}
            \end{equation}
        \end{subequations}
        This is the form of gauge transformations of connections on higher principal bundles found often in the literature.
        
        The finite form of the higher gauge transformations are parameterised by 
        \begin{subequations}
            \begin{equation}
                m\ \in\ \Omega^0(U,\sfH)~,
            \end{equation}
            and these act according to
            \begin{equation}\label{eq:unadjustedHigherGaugeTransformationsB}
                \begin{aligned}
                    \tilde a\ &=\ \sft(m)a~,
                    \\
                    \tilde \lambda\ &=\ \lambda+a^{-1}\acton(m^{-1}\nabla m)~,
                \end{aligned}
            \end{equation}
            where $\nabla\coloneqq\rmd+A\acton$. 
        \end{subequations}
        
        As expected, higher gauge transformations link gauge transformations with different images, which is problematic. In the finite setting, we have the following deformation, integrating~\eqref{eq:infinitesimal_deformation}.
        \begin{lemma}\label{lemma:composition_and_globularity}
            The gauge transformations
            \begin{subequations}
                \begin{equation}\label{eq:finite_gauge_transformations}
                    \begin{aligned}
                        \tilde A\ &=\ a^{-1}Aa+a^{-1}\rmd a-\sft(\lambda)~,
                        \\
                        \tilde B\ &=\ a^{-1}\acton B+\rmd\lambda+\tilde{A}\acton\lambda+\tfrac12[\lambda,\lambda]-\kappa\big(a^{-1},F\big)
                        \\
                        \tilde a\ &=\ \sft(m)a~,
                        \\
                        \tilde \lambda\ &=\ \lambda+a^{-1}\acton(m^{-1}\nabla m)
                    \end{aligned}
                \end{equation}
                with $\kappa:\sfG\times \frg\rightarrow \frh$ satisfying
                \begin{equation}\label{eq:finite_adjustment_condition}
                    \begin{aligned}
                        \kappa(\sft(h),V)\ &=\ h(V\acton h^{-1})~,
                        \\
                        \kappa(g_2g_1,V)\ &=\ g_2\acton\kappa(g_1,V)+\kappa\big(g_2,g_1Vg^{-1}_1-\sft(\kappa(g_1,V))\big)
                    \end{aligned}
                \end{equation}
            \end{subequations}
            compose according to
            \begin{equation}\label{eq:gauge_trafor_composition}
                (a_3,\lambda_3)\ \coloneqq \ (a_1,\lambda_1)\circ(a_2,\lambda_2)\ =\ (a_1a_2,\lambda_2+a_2^{-1}\acton \lambda_1)~,
            \end{equation} 
            are globular in the sense that gauge transformations parameterised by $(a,\lambda)$ and $(\tilde a,\tilde \lambda)$ have the same image, and they integrate the infinitesimal gauge transformations~\eqref{eq:infinitesimal_deformation}.
        \end{lemma}
        \begin{proof}
            This is again a direct computation. Globularity of the gauge transformations is equivalent to the first equation in~\eqref{eq:finite_adjustment_condition}, and the fact that gauge transformations compose is equivalent to the second condition. Useful formulas for the explicit computations are found in \ref{app:kappaProperties}. Finally, it is easy to see that the linearisation of~\eqref{eq:finite_gauge_transformations} is~\eqref{eq:infinitesimal_deformation}.
        \end{proof}
        \noindent Denoting the two different maps $\kappa:\frg\times \frg\rightarrow \frh$ and $\kappa:\sfG\times \frg\rightarrow \frh$ by the same symbol is justified since the former is the linearisation of the latter, i.e.
        \begin{equation}\label{eq:linearisedKappa}
            \kappa(V_1,V_2)\ \coloneqq\ \lim_{\eps\to0}\tfrac1\eps\kappa(\exp(\eps V_1),V_2)
        \end{equation}
        for all $V_{1,2}\in\frg$.        
        
        \Cref{lemma:composition_and_globularity} allows us to state the following theorem, which also defines the adjusted BRST 2-groupoid.
        \begin{theorem}\label{thm:adjusted_BRST_2_groupoid}
            Consider a crossed module of Lie groups $\caG\coloneqq(\sfH\overset{\sft}{\longrightarrow}\sfG,\acton)$ together with a map $\kappa:\sfG\times \frg\rightarrow \frh$ satisfying~\eqref{eq:finite_adjustment_condition}. For each manifold $U$, there is the \uline{adjusted BRST 2-groupoid} $\scB$ which is the (strict) action 2-groupoid with objects the local connection forms,
            \begin{equation}
                \scB_0\ \coloneqq\ \Omega^1(U,\frg)\oplus \Omega^2(U,\frh)~,
            \end{equation}
            the morphisms the gauge transformations between these,
            \begin{equation}
                \scB_1\ \coloneqq\ \Omega^1(U,\frg)\oplus \Omega^2(U,\frh)\oplus \Omega^0(U,\sfG)\oplus \Omega^1(U,\frh)~,
            \end{equation} 
            and the 2-morphisms the higher gauge transformation between these,
            \begin{equation}
                \scB_2\ \coloneqq\ \Omega^1(U,\frg)\oplus \Omega^2(U,\frh)\oplus \Omega^0(U,\sfG)\oplus \Omega^1(U,\frh)\oplus \Omega^0(U,\sfH)~.
            \end{equation}
            The source and targets of the 1- and 2-morphisms are given by the diagram
            \begin{equation}
                \begin{tikzcd}[column sep=3.5cm]
                    (A,B) 
                    \arrow[r,bend left=30,"{(A,B,a,\lambda)}"{name=U}]
                    \arrow[r,bend right=30,"{(A, B,\tilde a,\tilde\lambda)}"{name=D},swap]
                    \arrow[Rightarrow,"{(A,B,a,\lambda,m)}", from=U, to=D,start anchor={[yshift=-1ex]},end anchor={[yshift=1ex]}]
                    & (\tilde A,\tilde B)~,
                \end{tikzcd}
            \end{equation}
            where $\tilde A,\tilde B,\tilde a,\tilde\lambda$ are found in~\eqref{eq:finite_gauge_transformations}. The units read as 
            \begin{equation}
                \sfid_{(A,B)}\ \coloneqq\ (A,B,\unit,0)
                \eand
                \sfid_{(A,B,a,\lambda)}\ \coloneqq\ (A,B,a,\lambda,\unit)~,
            \end{equation} 
            horizontal composition is given by 
            \begin{equation}
                (A_1,B_1,a_1,\lambda_1,m_1)\otimes (A_2,B_2,a_2,\lambda_2,m_2)\ \coloneqq \ (A_1,B_1,a_1a_2,\lambda_2+a_2^{-1}\acton \lambda_1,m_1 a_1\acton m_2)~,
            \end{equation} 
            and vertical composition is given by
            \begin{equation}
                (\tilde A,\tilde B,\tilde a,\tilde \lambda,m_2)\circ (A,B,a,\lambda,m_1)\ \coloneqq\ (A,B,a,\lambda,m_2m_1)~.
            \end{equation} 
            The inverses of horizontal and vertical composition are given by
            \begin{equation}
                \begin{aligned}
                    (A,B,a,\lambda)^{-1}\ &\coloneqq\ (\tilde A,\tilde B,a^{-1},-a\acton \lambda)~,
                    \\
                    (A,B,a,\lambda,m)^{-1}\ &\coloneqq\ (\tilde A,\tilde B,\tilde a,\tilde \lambda,m^{-1})~,
                \end{aligned}
            \end{equation}
            respectively.
        \end{theorem}
        \begin{proof}
            Besides the properties established in \ref{lemma:composition_and_globularity}, this is a simple verification of the axioms of a strict 2-groupoid.
        \end{proof}
        
        \subsection{Adjusted crossed modules and their morphisms}\label{ssec:adjustedCrossedModules}
        
        \paragraph{Adjusted crossed modules of Lie groups.} We note that the map $\kappa$ appearing in the adjusted BRST 2-groupoid is a purely algebraic datum on the crossed module of Lie groups $\caG$. Before proceeding with the construction of principal 2-bundles with adjusted connections, let us explore this structure a bit further.
        
        The full geometric meaning of adjustments is still under investigation. In the infinitesimal case of certain $L_\infty$-algebra valued connections, it was shown that the adjustment corresponds to an alternator in a notion of categorified Lie algebra, which describes the failure of the Lie bracket~\cite{Borsten:2021ljb} to be antisymmetric. Moreover, the adjustment is necessary for the consistent definition of invariant polynomials in an $L_\infty$-algebra, as explained in~\cite{Sati:2008eg,Saemann:2019dsl}. 
        
        \begin{definition}\label{def:generalAdjustment}
            An \uline{adjusted crossed module of Lie groups} $\caG_\kappa\coloneqq(\caG,\kappa)$ is a crossed module of Lie groups $\caG\coloneqq(\sfH\overset{\sft}{\longrightarrow}\sfG,\acton)$ together with a map 
            \begin{subequations}\label{eq:alternativeAdjustmentCondition}
                \begin{equation}\label{eq:general_adjustment}
                    \kappa\,:\,\sfG\times\frg\ \rightarrow\ \frh~,
                \end{equation}
                called the \uline{adjustment}, which is linear in $\frg$ and satisfies
                \begin{align}
                    \kappa(\sft(h),V)\ &=\ h(V\acton h^{-1})~,\label{eq:alternativeAdjustmentCondition_a}
                    \\
                    \kappa(g_2g_1,V)\ &=\ g_2\acton\kappa(g_1,V)+\kappa\big(g_2,g_1Vg^{-1}_1-\sft(\kappa(g_1,V))\big)\label{eq:alternativeAdjustmentCondition_b}
                \end{align}
                for all $g_{1,2}\in\sfG$, $h\in\sfH$, and $V\in\frg$.
            \end{subequations}
        \end{definition}
        
        We will present concrete examples below, but already note the following.
        \begin{remark}\label{rem:not_all_2_groups_have_adjustments}
            Not all crossed modules of Lie groups admit an adjustment. For a counterexample, consider the crossed module $\sfT\sfB_n^{F2}$ discussed in~\cite[Section 4.2]{Kim:2022opr}.
        \end{remark}
        
        \paragraph{Special adjustments.} In this paper, we will focus on a particular class of adjustments which we call special adjustments. Among other things, these allow us to introduce the important notion of $P$-flatness later, and all our examples will be special. 
        
        Recall that the data of a crossed module of Lie groups $(\sfH\overset{\sft}{\longrightarrow}\sfG,\acton)$ that are preserved under quasi-isomorphisms are the kernel and cokernel of $\sft$. Whilst the adjustment datum taking values in the complement of the kernel of $\sft$ is not canonical, it can often be determined by a simple map, as we now show.
        
        \begin{proposition}\label{prop:special_adjustment}
            Consider a crossed module of Lie groups $\caG\coloneqq(\sfH\overset{\sft}{\longrightarrow}\sfG,\acton)$ together with its crossed module of Lie algebras $\sfLie(\caG)\coloneqq(\frh\overset{\sft}{\longrightarrow}\frg,\acton)$. Let
            \begin{subequations}
                \begin{equation}
                    P\,:\,\frg\ \rightarrow\ \frh
                \end{equation}
                be a morphism of vector spaces such that
                \begin{equation}\label{eq:P_rels_1}
                    \sft\circ P\circ\sft\ =\ \sft~.
                \end{equation}
            \end{subequations}
            Then, $\kappa_0$ with\footnote{Here, we use the shorthand notation $g[V,g^{-1}]\coloneqq gVg^{-1}-V$.} 
            \begin{equation}\label{eq:kappa_0}
                \kappa_0(g,V)\ \coloneqq\ P(g[V,g^{-1}])
            \end{equation}
            satisfies the conditions~\eqref{eq:alternativeAdjustmentCondition} up to terms that are in the kernel of $\sft$. We shall call an adjustment $\kappa$ such that $\sft\circ\kappa=\sft\circ\kappa_0$ for some map $P$ as above a \uline{special adjustment}.
        \end{proposition}
        \begin{proof}
            We have
            \begin{equation}
                \begin{aligned}
                    \sft(\kappa_0(\sft(h),V))\ &=\ \sft(P(\sft(h)[V,\sft(h)^{-1}]))
                    \\ 
                    &=\ \sft(P(\sft(h(V\acton h^{-1}))))
                    \\
                    &=\ \sft(h(V\acton h^{-1}))
                \end{aligned}
            \end{equation}
            for all $h\in\sfH$ and for all $V\in\frg$, as well as 	
            \begin{equation}
                \begin{aligned}
                    \sft(\kappa_0(g_2g_1,V))\ &=\ \sft(P(g_2g_1[V,g_1^{-1}g_2^{-1}]))
                    \\
                    &=\ g_2\sft(P(g_1[V,g_1^{-1}]))g_2^{-1}+\sft\Big(P(g_2[g_1Vg^{-1}_1,g^{-1}_2])\Big)
                    \\
                    &\hspace{1cm}-\sft\Big(P(g_2[\sft(P(g_1[V,g_1^{-1}])),g^{-1}_2])\Big)
                    \\
                    &=\ \sft\Big(g_2\acton P(g_1[V,g_1^{-1}])+P(g_2[g_1Vg^{-1}_1-\sft(P(g_1[V,g_1^{-1}])),g^{-1}_2])\Big)
                    \\
                    &=\ \sft\Big(g_2\acton\kappa_0(g_1,V)+\kappa\big(g_2,g_1Vg^{-1}_1-\sft(\kappa_0(g_1,V))\big)\Big)~,
                \end{aligned}
            \end{equation}
            where we have used that
            \begin{equation}\label{eq:P_rels_2}
                \sft(P(g_2\sft(P(g_1[V,g_1^{-1}]))g_2^{-1}))\ =\ g_2\sft(P(g_1[V,g_1^{-1}]))g_2^{-1}
            \end{equation}
            for all $g_{1,2}\in\sfG$ and for all $V\in\frg$.
        \end{proof}
        
        \paragraph{Decomposition induced by special adjustment.} Let us attempt to give a clearer picture of the role of the map $P$. Given a special adjustment of a crossed module of Lie groups $\caG\coloneqq(\sfH\overset{\sft}{\longrightarrow}\sfG,\acton)$ with its corresponding crossed module of Lie algebras $\sfLie(\caG)\coloneqq(\frh\overset{\sft}{\longrightarrow}\frg,\acton)$ and special adjustment datum $\kappa$ and $P$ as in \ref{prop:special_adjustment}, we note that we have a natural decomposition of $\sfLie(\caG)$ as
        \begin{equation}\label{eq:decompositions}
            \frh\ \cong\ \ker(\sft)\oplus\im(P\circ\sft)
            \eand
            \frg\ \cong\ \coker(\sft)\oplus\im(\sft\circ P)~,
        \end{equation}
        where we identify $\coker(\sft)=\im(\mathrm{id}-\sft\circ P)$, and we also have $\ker(\sft)=\im(\mathrm{id}-P\circ\sft)$.
        
        The map $P$ thus provides the decomposition~\eqref{eq:decompositions}, and $\kappa_0$ as in~\eqref{eq:kappa_0} is the image of an adjustment in $\im(P\circ \sft)$, so that only the image in $\ker(\sft)$ remains to be specified. For injective $\sft$ as in the example underlying \ref{sec:principal_G_bundles_as_adjusted_higher_bundles}, it suffices to specify $P$ to define an adjustment.
        
        \begin{remark}
            A close analogue of the map $P$ appeared  in~\cite{Ho:2012nt}, and there also with the goal of lifting the fake-flatness constraint. In this paper, the authors considered the data in a crossed modules of Lie groups $(\sfH\overset{\sft}{\longrightarrow}\sfG,\acton)$ together with a map $\sfs:\sfG\rightarrow\sfH$. They also identified the projector $\sft\circ\sfs$, corresponding to our projector $\sft\circ P$, as relevant, and presented a deformation of the infinitesimal $B$-field gauge transformations. Whilst there is thus some overlap between our construction, there are some crucial difference in the conditions on $\sfs$ and $P$ as well as the complete form of the deformation.
        \end{remark}
        
        \paragraph{Example.} We now finally come to an example of an adjusted crossed module of Lie groups.\footnote{We shall extend this example further in \ref{sec:Lie2GroupModel} to a 2-group model of the string group, the key example of an adjusted crossed module for our purposes.} Recall that the based (parametrised) \uline{path} and \uline{loop groups} of a Lie group are given by 
        \begin{equation}\label{eq:basedPathLoopGroups}
            P_0\sfG\ \coloneqq\ \{p\in\scC^\infty([0,1],\sfG)\,|\,p(0)=\unit\}
            \eand
            L_0\sfG\ \coloneqq\ \{p \in P_0\sfG\,|\,p(1)=\unit\}~,
        \end{equation} 
        respectively.\footnote{In our conventions for based path and loop groups, we follow~\cite{Baez:2005sn}, and we refer to this paper for further details. We shall ignore technicalities such as introducing sitting instants in our loop and path parametrisations, as these can be trivially incorporated into our constructions. For a detailed discussion of these issues, see e.g.~\cite{Schreiber:0705.0452}. The important point to keep in mind is that the concatenation of two elements in our path space with the same endpoint always forms an element in our loop space.} Both $P_0\sfG$ and $L_0\sfG$ are Fr{\'e}chet--Lie groups with the evident pointwise products. Furthermore, the endpoint evaluation map
        \begin{equation}\label{eq:endpointEvaluation}
            \begin{aligned}
                \boundary\,:\,P_0\sfG\ &\rightarrow\ \sfG~,
                \\
                p\ &\mapsto\ p(1)
            \end{aligned}
        \end{equation}
        is a morphism of Fr{\'e}chet--Lie groups. Similarly, we define the \uline{based, parametrised path} and \uline{loop Lie algebras} of the Lie algebra $\frg$
        \begin{equation}
            P_0\frg\ \coloneqq\ \{\gamma\in\scC^\infty([0,1],\frg)\,|\,\gamma(0)=0\}\eand
            L_0\frg\ \coloneqq\ \{\gamma\in P_0\frg\,|\,\gamma(1)=0\}
        \end{equation}
        with pointwise Lie brackets. 
        
        We can now combine $P_0\sfG$ and $L_0\sfG$ into the crossed module of Lie groups
        \begin{equation}\label{eq:definitionLoopLieCrossed}
            \caL\sfG\ \coloneqq\ (L_0\sfG\hookrightarrow P_0\sfG,\Ad)~,
        \end{equation}
        where $\Ad$ is the pointwise adjoint action of $L_0\sfG$ on $P_0\sfG$. The corresponding crossed module of Lie algebras $\sfLie(\caL\sfG)$ is then 
        \begin{equation}\label{eq:crossedModuleLg}
            \caL\frg\ \coloneqq\ (L_0\frg\hookrightarrow P_0\frg,\ad)~.
        \end{equation}
        
        The crossed module $\caL\sfG$ is equivalent to $\sfG$, trivially regarded as the crossed module of Lie groups\footnote{`cm' refers to crossed module}
        \begin{equation}\label{eq:definitionLieGroupLieCrossed}
            \sfG_{\rm cm}\ \coloneqq\ (\unit\hookrightarrow\sfG,\id)~,
        \end{equation}
        and the evident flippable butterfly between $\caL \sfG$ and $\sfG_{\rm cm}$ is\footnote{Note that $L_0\sfG$ is, in fact, a normal subgroup of $P_0\sfG$, and we have
            \begin{equation}\label{eq:identification}
                \begin{tikzcd}[ampersand replacement=\&]
                    P_0\sfG \arrow[r,->>]\arrow[rr,bend right=30,"\boundary"] \& P_0\sfG/L_0\sfG\arrow[r,"\cong"] \& \sfG~.
                \end{tikzcd}
            \end{equation}
            The equivalence between the crossed modules $\caL\sfG$ and $\sfG_{\rm cm}$ is a lift of the isomorphism between the quotient group $P_0\sfG/L_0\sfG$ and $\sfG$.}
        \begin{equation}\label{eq:loopGroupButterfly}
            \begin{tikzcd}
                \unit\arrow[dd]\arrow[dr]& & L_0\sfG\arrow[dl]\arrow[dd]
                \\
                & P_0\sfG\arrow[dl,"\boundary",swap]\arrow[dr,"\id"]& 
                \\
                \sfG & & P_0\sfG
            \end{tikzcd}
        \end{equation}
        
        \begin{proposition}\label{prop:adjustmentPathLoopGroups}
            Consider some fixed map $\wp\in\scC^\infty([0,1],\IR)$ with $\wp(0)=0$ and $\wp(1)=1$. A special adjustment for $\caL\sfG$ is then given by the map
            \begin{subequations}\label{eq:adjustmentPathLoopGroups}
                \begin{equation}
                    \begin{aligned}
                        \kappa\,:\,P_0\sfG\times P_0\frg\ &\rightarrow\ L_0\frg~,
                        \\
                        (g,V)\ &\mapsto\ P(g[V,g^{-1}])
                    \end{aligned}
                \end{equation}
                for all $g\in P_0\sfG$ and for all $V\in P_0\frg$ with
                \begin{equation}\label{eq:Pmap}
                    P\ \coloneqq\ \id-\wp\cdot\boundary~.
                \end{equation}
            \end{subequations}
        \end{proposition}
        \begin{proof}
            We note that $\sft$ is injective, so it suffices to check the condition $\sft\circ P\circ\sft=\sft$ for a special adjustment from \cref{prop:special_adjustment}. This follows immediately since $\boundary\circ\sft=0$. 
        \end{proof}
        
        Because $\frg$ is the cokernel of $\sft$ in $\sfLie(\sfG_{\rm cm})$, the map $P$ yields a choice of embedding\footnote{There are certainly much more general choices.} of $\frg$ into $P_0\frg$, which leads to the decomposition
        \begin{equation}
            \begin{aligned}
                P_0\frg\ &\rightarrow\ \frg\oplus L_0\frg~,
                \\
                V\ &\mapsto\ (P-\id)(V)+P(V)\ =\ \wp\cdot\boundary(V)+(\id-\wp\cdot\boundary)(V)
            \end{aligned}
        \end{equation}
        for all $V\in P_0\frg$.
        
        Note that the crossed module $\caL\frg$ is quasi-isomorphic to the Lie algebra $\frg$, regarded as the trivial crossed module $\frg_{\rm cm}\coloneqq(0\hookrightarrow\frg,\id)$, with the quasi-isomorphism the projection by endpoint evaluation. The embedding defined by $P$ then yields an inverse quasi-isomorphism,
        \begin{equation}\label{eq:quasi-iso_Lie_algebras}
            \frg_{\rm cm}\ \xrightarrow{~\cong~}\ \caL\frg~,
        \end{equation}
        see e.g.~\cite{Saemann:2019dsl} for details.
        
        \paragraph{Morphisms.} 
        The definition of strict morphisms of adjusted crossed modules of Lie groups is straightforward. 
        
        \begin{definition}
            A \uline{strict morphism} $\phi$ of crossed modules $\caG_{1,2}$ of Lie groups with adjustments $\kappa_{1,2}$ is a strict morphism of crossed modules $\phi:\caG_1\rightarrow\caG_2$ of Lie groups such that 
            \begin{equation}
                \phi_\frh(\kappa_1(g,V))\ =\ \kappa_2(\phi_\sfG(g),\phi_\frg(V))
            \end{equation}
            for all $g\in\sfG$ and for all $V\in\frg$, where $\phi_\frg$ and $\phi_\frg$ denote the components of the strict morphism $\phi:\sfLie(\caG_1)\rightarrow\sfLie(\caG_2)$ induced by $\phi$. \uline{Strict quasi-isomorphisms} of adjusted crossed modules of Lie groups are strict morphism of adjusted crossed modules of Lie groups that become strict quasi-isomorphisms after forgetting the adjustment.
        \end{definition}
        
        To define weak morphisms of adjusted crossed modules, we extend the spans of strict Lie 2-groups.
        
        \begin{definition}
            A \uline{(weak) morphism} $\phi$ of crossed modules $\caG_{1,2}$ of Lie groups with adjustments $\kappa_{1,2}$ is a span of the form 
            \begin{equation}
                \begin{tikzcd}
                    & \caK\arrow[dl,"\phi"'] \arrow[dr,"\psi"]
                    \\
                    \caG_1 & & \caG_2
                \end{tikzcd}
            \end{equation}
            where $\caK$ is a third adjusted crossed module of Lie groups and $\phi$ and $\psi$ are strict morphisms of adjusted crossed modules of Lie groups with $\phi$ a quasi-isomorphism.
        \end{definition}
        
        \begin{example} 
            As a simple example, consider the butterfly~\eqref{eq:loopGroupButterfly} witnessing the equivalence between the crossed modules $\sfG_{\rm cm}$ from~\eqref{eq:definitionLieGroupLieCrossed}, with the unique and trivial adjustment, and $\caL\sfG$ from~\eqref{eq:crossedModuleLg} with adjustment~\eqref{eq:adjustmentPathLoopGroups}. We have the following span of crossed modules
            \begin{equation}
                \begin{tikzcd}
                    & \caL\sfG \arrow[dl,"\phi"'] \arrow[dr,"\psi"]
                    \\
                    \sfG_{\rm cm} & & \caL\sfG
                \end{tikzcd}
            \end{equation}
            where $\phi$ is the endpoint evaluation map and $\psi$ is the identity. Recall that $\phi$ is a quasi-isomorphism, and $\phi$ is trivially a strict morphism of adjusted crossed modules. We thus see that the adjustment on $\caL\sfG$ (up to choice of the map $P$) is induced by the trivial adjustment on $\sfG_{\rm cm}$.
        \end{example}
        
        \paragraph{Adjustment for crossed modules of Lie algebras.} 
        It is useful to develop also the infinitesimal notion of adjustments. The Lie functor linearises an adjusted crossed module of Lie groups to an adjusted crossed module of Lie algebras in a similar fashion as in the unadjusted setting. 
        
        We compute the following.
        \begin{proposition}
            Given an adjusted crossed module of Lie groups $\caG$, the adjustment datum $\kappa:\frg\times \frg\rightarrow \frh$ in the corresponding crossed module of Lie algebras $\sfLie(\caG)\coloneqq(\frh\overset{\sft}{\longrightarrow}\frg,\acton)$ satisfies 
            \begin{equation}
                \begin{aligned}            
                    \kappa(\sft(W),V)\ &=\ -V\acton W~,
                    \\
                    \kappa([V_1,V_2],V_3)\ &=\ V_1\acton \kappa(V_2,V_3)-V_2\acton \kappa(V_1,V_3)+\kappa(V_1,[V_2,V_3])-\kappa(V_2,[V_1,V_3])
                    \\
                    &\hspace{1cm}-\kappa(V_1,\sft(\kappa(V_2,V_3)))+\kappa(V_2,\sft(\kappa(V_1,V_3)))
                \end{aligned}
            \end{equation}
            for all $V_{1,2}\in\frg$ and for all $W\in\frh$. Given a crossed module of Lie algebras with such a map $\kappa$, we call it an \uline{adjusted crossed module of Lie algebras}.
        \end{proposition}
        
        \begin{proof}
            This is a short, direct computation.
        \end{proof}
        
        \begin{example}
            The special adjustment~\eqref{eq:adjustmentPathLoopGroups} induces an adjustment datum on the corresponding crossed module of Lie algebras, which reads as
            \begin{equation}\label{eq:adjustmentPathLoopAlgebras}
                \begin{aligned}
                    \kappa\,:\,P_0\frg\times P_0\frg\ &\rightarrow\ L_0\frg~,
                    \\
                    (V_1,V_2)\ &\mapsto\ P([V_1,V_2])
                \end{aligned}
            \end{equation}
            for all $V_{1,2}\in P_0\frg$ as follows by a direct computation.
        \end{example}
        
        Let us also briefly consider the form of morphisms of adjusted crossed modules of Lie algebras. Consider adjusted crossed modules $\sfLie(\caG_{1,2})\coloneqq(\frh_{1,2}\overset{\sft_{1,2}}{\longrightarrow}\frg_{1,2},\acton_{1,2},\kappa_{1,2})$ together with a butterfly $\phi:\sfLie(\caG_{1})\rightarrow\sfLie(\caG_2)$ in the sense of 2-term $L_\infty$-algebra morphisms and satisfying~\eqref{eq:L_infty_alg}. As explained in~\cite[Proposition 3.4]{Noohi:0910.1818}, the morphism $\phi$ can be turned into a butterfly
        \begin{equation}
            \begin{tikzcd}
                \frh_1 \arrow[dd,"\sft_1",swap] \arrow[dr,"\lambda_1"]& & \frh_2 \arrow[dl,"\lambda_2",swap] \arrow[dd,"\sft_2"]
                \\
                & \frh_2\oplus \frg_1 \arrow[dl,"\gamma_1",swap] \arrow[dr,"\gamma_2"]& 
                \\
                \frg_1 & & \frg_2
            \end{tikzcd}
        \end{equation}
        with all maps evident except for
        \begin{equation}
            \gamma_2(W_2,V_1)\ =\ \sft_2(W_2)+\phi_\frg(V_1)
        \end{equation}
        for all $W_2\in\frh_2$ and for all $V_1\in\frg_1$. According to~\cite[Remark 8.5]{Noohi:0506313}, this butterfly can be turned into the span of Lie crossed modules 
        \begin{equation}
            \begin{tikzcd}
                & \frh_1\oplus\frh_2\rightarrow \frh_2\oplus \frg_1 \arrow[dl,"{\pr_1\rightarrow \pr_2}"'] \arrow[dr,"{\pr_2\rightarrow \gamma_2}"]
                \\
                \frh_1\rightarrow \frg_1 & & \frh_2\rightarrow \frg_2
            \end{tikzcd}
        \end{equation}
        Lifting this to a span of adjusted crossed modules of Lie algebras with adjustments $\kappa_{1,2}$ and $\hat\kappa$ for $\frh_{1,2}\rightarrow\frg_{1,2}$ and $\frh_1\oplus\frh_2\rightarrow\frh_2\oplus\frg_1$, respectively, amounts to imposing
        \begin{equation}
            \begin{aligned}
                \pr_1(\hat\kappa(W_2,V_1;W_2',V_1'))\ &=\ \kappa_1(V_1,V_1')~,
                \\
                \pr_2(\hat\kappa(W_2,V_1;W_2',V_1'))\ &=\ \kappa_2(\sft(W_2)+\phi_\frg(V_1),\sft(W'_2)+\phi_\frg(V'_1))
            \end{aligned}
        \end{equation}
        for all $W_1,W'_1\in\frh_1$ and $W_2,W'_2\in\frh_2$ and for all $V_1,V'_1\in\frg_1$ and $V_2,V'_2\in\frg_2$. We hence see that the adjustment $\hat\kappa$ is fully fixed by the adjustments $\kappa_{1,2}$, and we arrive at the following statement.
        
        \begin{proposition}
            A pair of weakly isomorphic crossed modules of Lie algebras that can both be equipped with an adjustment are also weakly isomorphic as adjusted crossed modules of Lie algebras.         
        \end{proposition}
        \noindent
        We note that this proposition states that adjustments on crossed modules are unique up to weak isomorphisms. This, however, is not sufficient to ensure that the corresponding adjusted BRST 2-groupoids are equivalent, which is the actually desired statement. We will study this problem in more detail the following section.
        
        \subsection{Equivalences of specially adjusted crossed modules}\label{ssec:equivalences}
        
        It is natural to expect that equivalent adjusted crossed modules yield equivalent adjusted BRST 2-groupoids over a manifold $U$. This, however, is not generally true. In the following, we develop a sufficient additional condition for special adjustments.

        \paragraph{Equivalent specially adjusted crossed modules.} Consider two specially adjusted crossed modules $\caG=(\caG,\kappa,P)$ and $\tilde \caG=(\tilde\caG,\tilde\kappa,\tilde P)$ which are equivalent by a span of such crossed modules,
        \begin{equation}\label{eq:span_specially_adjusted}
            \begin{tikzcd}
                & \hat \caG\arrow[dl,"\phi"'] \arrow[dr,"\psi"]
                \\
                \caG & & \tilde \caG
            \end{tikzcd}
        \end{equation}
        with $\hat \caG=(\hat \caG,\hat \kappa,\hat P)$ another crossed module with special adjustment involving $\hat P$. 
        \begin{definition}
            We call $\caG$ and $\tilde \caG$ equivalent, if $\phi$ and $\psi$ are quasi-isomorphisms with
            \begin{equation}
                P\circ \phi\ =\ \phi\circ \hat P
                \eand
                \tilde P\circ \psi\ =\ \psi\circ \hat P~.
            \end{equation} 
        \end{definition}
        \noindent Clearly, $\phi$ and $\psi$ respect the decompositions~\eqref{eq:decompositions}. 
        
        It is now sufficient to consider equivalence for a  strict quasi-isomorphisms, e.g.~$\phi$; this will then extend to the other side of the span~\eqref{eq:span_specially_adjusted}. We thus want to examine whether the BRST 2-groupoids $\scB$ for $\caG$ and $\hat \scB$ for $\hat \caG$ are equivalent. For this to be true, gauge-equivalence classes of connections have to be mapped injectively to gauge-equivalence classes of connections under $\phi$.
        
        As a first step, we apply the decomposition~\eqref{eq:decompositions} to both $\caG$ and $\hat \caG$, obtaining
        \begin{equation}\label{eq:decomposition_equivalence}
			\begin{aligned}
            \hat\frh\ &\cong\ \ker(\hat\sft)\oplus\im(\hat P\circ\hat\sft)
            &&\eand
            &\hat\frg\ &\cong\ \coker(\hat\sft)\oplus\im(\hat\sft\circ \hat P)~,
            \\
            \frh\ &\cong\ \ker(\sft)\oplus\im(P\circ\sft)
            &&\eand
            &\frg\ &\cong\ \coker(\sft)\oplus\im(\sft\circ P)~.
			\end{aligned}
		\end{equation}
        The map $\phi$ is a chain map with regards to $\sft$ and $\hat \sft$ and hence maps $\ker(\sft)$ to $\ker(\hat \sft)$ and $\coker(\hat \sft)=(\mathrm{id}-\hat \sft \circ \hat P)$ to $\coker(\sft)$.
        
        Consider now gauge equivalence classes of local connection forms $(\hat A,\hat B)\in \hat \scB_0$ and $(A,B)\in \scB_0$. We can choose gauges (i.e.~representatives of gauge equivalence classes) such that 
        \begin{equation}
            \hat A|_{\im(\hat\sft\circ \hat P)}=0
            \eand
            A|_{\im(\sft\circ P)}=0
        \end{equation} 
        by applying a suitable 1-morphism. This gauge choice restricts the 1-morphisms such that
        \begin{equation}\label{eq:restricted_gauge}
            \hat \lambda\in \ker(\hat \sft)
            \eand 
            \lambda\in \ker(\sft)~.
        \end{equation} 
        The remaining parts of $\hat A$ and $A$ take values in $\coker(\hat \sft)$ and $\coker(\sft)$, respectively, and these spaces are isomorphic with the isomorphism given by the restriction of $\phi$.
        
        Similarly, the component $\hat B|_{\ker(\hat \sft)}$ is bijectively mapped to $B|_{\ker(\sft)}$, and it remains to consider $\hat B|_{\im(\hat\sft\circ \hat P)}$. We note that the restricted gauge transformations preserve the decomposition $\hat B=\hat B|_{\ker(\hat \sft)}+\hat B|_{\im(\hat P\circ\hat\sft)}$ due to~\eqref{eq:restricted_gauge}. Moreover, $\phi$ is generically not an isomorphism from $\im(\hat\sft\circ \hat P)$ to $\im(\sft\circ P)$. 
        
        If, however, $\hat B|_{\ker(\hat \sft)}$ and $B|_{\ker(\hat \sft)}$ were not independent, but determined by other components of the local connection forms, then this would ensure an injective map between gauge equivalence classes of components. We therefore define the following.
        \begin{definition}
            Let $(\kappa,P)$ be a special adjustment~\eqref{eq:adjustmentPathLoopGroups} of a crossed module of Lie groups. We call a local connection form $(A,B)$ \uline{$P$-flat} if the fake curvature $F$ satisfies
            \begin{equation}\label{eq:adjustedFakeCurvatureCondition}
                P(F)\ =\ 0~.
            \end{equation}
            
            The adjusted BRST 2-groupoid can be consistently restricted to the full sub-2-groupoid of $P$-flat connections. We call this the \uline{adjusted $P$-flat BRST 2-groupoid}.
        \end{definition}
        
        Our considerations above then show the following.
        \begin{proposition}
            Consider two equivalent specially adjusted crossed modules of Lie groups. Then the corresponding adjusted $P$-flat BRST 2-groupoids are equivalent.
        \end{proposition}

        We note that imposing $P$-flatness is a choice, motivated by mathematical considerations, but this condition is not necessary for the consistency of the BRST 2-groupoid.
        
        Moreover, we only showed that $P$-flatness is a sufficient condition, and it will not be necessary for certain equivalence classes of specially adjusted crossed modules of Lie groups. An explicit example highlighting the necessity of this condition is given in \ref{sec:principalAdjustedHigherPrincipal}.
                
        \subsection{Adjusted differential cohomology for higher principal bundles}\label{sec:adj_cocycles}
        
        We can now stackify the adjusted BRST 2-groupoid in order to construct the adjusted differential cohomology for principal 2-bundles with connection. 
        
        \paragraph{Covers.}
        Recall that a \uline{cover} is a surjective submersion $\sigma:Y\rightarrow X$ of manifolds, not necessarily a local diffeomorphism, together with the fibre products
        \begin{equation}\label{eq:fibreProducts}
            Y^{[n]}\ \coloneqq\ \underbrace{Y\times_X\cdots\times_XY}_{\text{$n$ factors}}\ \coloneqq\ \{(y_1,\ldots,y_n)\in Y^n\,|\,\sigma(y_1)=\cdots=\sigma(y_n)\}
        \end{equation}
        for $n\in\IN$. We will sometimes make the additional assumption of working with a \emph{good} cover by which we mean a cover in which all the $Y^{[n]}$ are disjoint unions of contractible spaces.
        
        \paragraph{Stackification.} 
        In the following, let $\caG_\kappa\coloneqq(\sfH\overset{\sft}{\longrightarrow}\sfG,\acton,\kappa)$ be an adjusted crossed module of Lie groups and $\sfLie(\caG_\kappa)\coloneqq(\frh\overset{\sft}{\longrightarrow}\frg,\acton,\kappa)$ be the associated adjusted crossed module of Lie algebras. Let also $Y$ be the cover of a manifold $X$ and $\scB$ be the BRST 2-groupoid for $\caG_\kappa$ and $Y$. In this 2-groupoid, we have assigned a local connection form to each element in $Y$, and we have local and higher local gauge transformations. 
        
        Next, we smoothly assign to each element in $Y^{[2]}$ an element in $\scB_1$, gluing the local connection data together. Finally, we assign to each element in $Y^{[3]}$ an element in $\scB_2$ that describes the composition of the gluing data over each $Y^{[2]}$. This leads to consistency relations over $Y^{[3]}$ and $Y^{[4]}$ which, together with the assigned gluing data over $Y^{[2]}$, produce the adjusted differential cocycles.
        
        To obtain the corresponding adjusted differential coboundaries, we consider local gauge transformations, together with higher gauge transformations of the gluing data on $Y^{[2]}$, leading to consistency relations on $Y^{[2]}$ and $Y^{[3]}$. 
        
        Finally, in order to obtain the adjusted differential 2-coboundaries, we consider local higher gauge transformations on $Y$ leading to consistency relations on $Y^{[2]}$.
        
        The required computations are straightforward, and we therefore only list the results in the following.
        
        \paragraph{Adjusted differential cocycles.} The above process yields the following data and relations for the cocycles.
        \begin{definition}\label{def:adj_cocycles}
            The \uline{differential cocycles} describing a \uline{principal $\caG$-bundle with adjusted connection} over a manifold $X$ subordinate to a cover $Y\rightarrow X$ consist of
            \begin{subequations}\label{eq:adjustedCocycleConditions}
                \begin{equation}
                    \begin{aligned}
                        h\ &\in\ \scC^\infty(Y^{[3]},\sfH)~,
                        \\
                        (g,\Lambda)\ &\in\ \scC^\infty(Y^{[2]},\sfG)\oplus\Omega^1(Y^{[2]},\frh)~,
                        \\
                        (A,B)\ &\in\ \Omega^1(Y^{[1]},\frg)\oplus\Omega^2(Y^{[1]},\frh)~,
                    \end{aligned}
                \end{equation}
                such that
                \begin{equation}\label{eq:adjustedCocycleConditionsB}
                    \begin{aligned}
                        h_{ikl}h_{ijk}\ &=\ h_{ijl}(g_{ij}\acton h_{jkl})~,
                        \\
                        g_{ik}\ &=\ \sft(h_{ijk})g_{ij}g_{jk}~,
                        \\
                        \Lambda_{ik}\ &=\ \Lambda_{jk}+g_{jk}^{-1}\acton \Lambda_{ij}-g_{ik}^{-1}\acton(h_{ijk}\nabla_i h_{ijk}^{-1})~,
                        \\
                        A_j\ &=\ g^{-1}_{ij}A_ig_{ij}+g^{-1}_{ij}\rmd g_{ij}-\sft(\Lambda_{ij})~,
                        \\
                        B_j\ &=\ g^{-1}_{ij}\acton B_i+\rmd\Lambda_{ij}+ A_j\acton\Lambda_{ij}+\tfrac12[\Lambda_{ij},\Lambda_{ij}]-\kappa\big(g_{ij}^{-1},F\big)\,,
                    \end{aligned}
                \end{equation}
            \end{subequations}
            for all appropriate $(i,j,\ldots)\in Y^{[n]}$. Here, $\nabla_i\coloneqq\rmd+A_i\acton$ as before, and the \uline{adjusted curvature forms} are defined as 
            \begin{equation}\label{eq:adjustedCurvatures}
                \begin{aligned}
                    F_i\ &\coloneqq\ \rmd A_i+\tfrac12[A_i,A_i]+\sft(B_i)\ \in\ \Omega^2(Y^{[1]},\frg)~,
                    \\
                    H_i\ &\coloneqq\ \rmd B_i+A_i\acton B_i-\kappa(A_i,F_i)\ \in\ \Omega^3(Y^{[1]},\frh)~,
                \end{aligned}
            \end{equation}
            for all $(i)\in Y^{[1]}$.
        \end{definition}
        
        The adjusted curvatures~\eqref{eq:adjustedCurvatures} satisfy the \uline{adjusted Bianchi identities}
        \begin{equation}\label{eq:adjustedBianchiIdentities}
            \nabla_i F_i\ =\ \sft(H_i+\kappa(A_i,F_i))
            \eand
            \nabla_i H_i\ =\ \kappa(A_i,\sft(H_i))-\kappa(F_i,F_i)~.
        \end{equation}
        The first identity in~\eqref{eq:adjustedBianchiIdentities} follows from straightforward differentiation. From~\eqref{eq:kappat_linear} and~\eqref{eq:kappa_commutator}, we find
        \begin{subequations}
            \begin{equation}
                \kappa(\sft(B_i),F_i)\ =\ -F_i\acton B_i     
            \end{equation} 
            and
            \begin{equation}
                A_i\acton\kappa(A_i,F_i)\ =\ \kappa\big(\tfrac12[A_i,A_i],F_i\big)-\kappa\big(A_i,[A_i,F_i]-\sft(\kappa(A_i,F_i))\big)\,,
            \end{equation}
        \end{subequations}
        respectively. Upon using these equations and the adjusted Bianchi identity for $F_i$, it is not too difficult to see that applying $\rmd$ to the definition of $H_i$ yields the second identity in~\eqref{eq:adjustedBianchiIdentities}.

        \paragraph{Adjusted differential 1-coboundaries.} The coboundary relations obtained from stackification are the following.
        \begin{definition}\label{def:adj_coboundaries}
            Two principal $\caG$-bundles with adjusted connections are said to be \uline{equivalent} whenever the corresponding adjusted differential cocycles $(h,g,\Lambda,A,B)$ and $(\tilde h,\tilde g,\tilde\Lambda,\tilde A,\tilde B)$ are linked by a \uline{differential 1-coboundary} consisting of maps
            \begin{subequations}\label{eq:adjustedCoboundaryConditions}
                \begin{equation}
                    \begin{aligned}
                        b\ &\in\ \scC^\infty(Y^{[2]},\sfH)~,
                        \\
                        (a,\lambda)\ &\in\ \scC^\infty(Y^{[1]},\sfG)\oplus\Omega^1(Y^{[1]},\frh)
                    \end{aligned}
                \end{equation}
                by means of
                \begin{equation}\label{eq:adjustedCoboundaryConditionsB}
                    \begin{aligned}
                        \tilde h_{ijk}\ &=\ a_i^{-1}\acton(b_{ik}h_{ijk}(g_{ij}\acton b_{jk}^{-1})b_{ij}^{-1})~,
                        \\
                        \tilde g_{ij}\ &=\ a_i^{-1}\sft(b_{ij})g_{ij}a_j~,
                        \\
                        \tilde\Lambda_{ij}\ &=\ a^{-1}_j\acton\Lambda_{ij}+\lambda_j-\tilde{g}^{-1}_{ij}\acton\lambda_i+(a_j^{-1}g_{ij}^{-1})\acton(b_{ij}^{-1}\nabla_ib_{ij})~,                    
                        \\
                        \tilde A_i\ &=\ a_i^{-1}A_ia_i+a_i^{-1}\rmd a_i-\sft(\lambda_i)~,
                        \\
                        \tilde B_i\ &=\ a_i^{-1}\acton B_i+\rmd\lambda_i+\tilde{A}_i\acton\lambda_i+\tfrac12[\lambda_i,\lambda_i]-\kappa\big(a_i^{-1},F_i\big)
                    \end{aligned}
                \end{equation}
            \end{subequations}
            for all appropriate $(i,j,\ldots)\in Y^{[n]}$. We call such coboundary transformations \uline{adjusted gauge transformations}.
        \end{definition}
        
        The adjusted curvatures~\eqref{eq:adjustedCurvatures} change under the adjusted gauge transformations~\eqref{eq:adjustedCoboundaryConditions} as
        \begin{equation}\label{eq:adjustedCurvatureGaugeTransformation}
            \tilde F_i\ =\ a_i^{-1}F_ia_i-\sft(\kappa(a_i^{-1},F_i))
            \eand
            \tilde H_i\ =\ a_i^{-1}\acton H_i-\kappa(a_i^{-1},\sft(H_i))
        \end{equation}
        for all $(i)\in Y^{[1]}$. For a special adjustment~\eqref{eq:adjustmentPathLoopGroups}, we also have
        \begin{equation}
            \sft(\tilde H_i)\ =\ \sft(H_i)
        \end{equation}
        for all $(i)\in Y^{[1]}$. Both results follow directly by using the formulas~\eqref{eq:adjustedCoboundaryConditionsB} for the gauge-transformed gauge potentials $\tilde A_i$ and $\tilde B_i$ in the definition of adjusted curvatures in~\eqref{eq:adjustedCurvatures}, and simplifying the resulting expressions using identities from \cref{app:useful,app:kappaProperties}. The derivation for $\tilde H_i$ is more involved. One first shows that
        \begin{equation}
            \tilde H_i+\kappa(\tilde A_i,\tilde F_i)\ =\ a^{-1}_i\acton(H_i+\kappa(A_i,F_i))+(a_i^{-1}F_i a_i)\acton\lambda_i-\tilde\nabla_i\kappa(a_i^{-1},F_i)~,
        \end{equation}
        and one then expands the terms with $\kappa$ using~\eqref{eq:dofkappa} and the Bianchi identity for $F_i$ from~\eqref{eq:adjustedBianchiIdentities} as well as~\eqref{eq:kappat_linear} and~\eqref{eq:kappa_galfaX}. The case of a special adjustment follows from 
        \begin{equation}
            \sft(\kappa(a_i^{-1},\sft(H_i))\ =\ \sft(P(\sft(a_i^{-1}[\sft(H_i),a_i])))\ =\ \sft(a_i^{-1}[\sft(H_i),a_i]))~.
        \end{equation}
        
        \paragraph{Adjusted differential 2-coboundaries.}Contrary to ordinary principal bundles, principal 2-bundles come with higher gauge transformations of bundle isomorphisms that link gauge transformations or bundle isomorphisms.
        \begin{definition}\label{def:adj_2-coboundaries}
            Two 1-coboundaries $(a,b,\lambda)$ and $(\tilde a,\tilde b,\tilde\lambda)$ are called \uline{equivalent} whenever they are linked by a \uline{2-coboundary} consisting of maps
            \begin{subequations}\label{eq:adjustedHigherGaugeTransformations}
                \begin{equation}
                    m\ \in\ \scC^\infty(Y^{[1]},\sfH)
                \end{equation}
                by means of
                \begin{equation}\label{eq:adjustedHigherGaugeTransformationsB}
                    \begin{aligned}
                        \tilde b_{ij}\ &=\ m_ib_{ij}(g_{ij}\acton m_j^{-1})~,
                        \\
                        \tilde a_i\ &=\ \sft(m_i)a_i~,
                        \\
                        \tilde \lambda_i\ &=\ \lambda_i+a_i^{-1}\acton(m_i^{-1}\nabla_i m_i)
                    \end{aligned}
                \end{equation}
            \end{subequations}
            for all appropriate $(i,j,\ldots)\in Y^{[n]}$. We shall refer to them as \uline{adjusted higher gauge transformations}.
        \end{definition}
        
        \paragraph{Comments on the construction.}
        Altogether, we obtain a 2-groupoid of principal 2-bundles with adjusted connection subordinate to a cover $Y\rightarrow X$, where cocycles, 1-coboundaries, and 2-coboundaries are linked schematically as follows:
        \begin{equation}
            \begin{tikzcd}
                (h,g,\Lambda,A,B) 
                \arrow[r,bend left=50,"{(a,b,\lambda)}"{name=U}]
                \arrow[r,bend right=50,"{(\tilde a,\tilde b,\tilde\lambda)}"{name=D},swap]
                \arrow[Rightarrow,"\,m", from=U, to=D,start anchor={[yshift=-1ex]},end anchor={[yshift=1ex]}]
                & (\tilde h,\tilde g,\tilde\Lambda,\tilde A,\tilde B)~,
            \end{tikzcd}
        \end{equation}
        The 2-groupoid of general principal 2-bundles is obtained as the colimit of the above 2-groupoid over the directed set of covers. 
        
        \begin{remark}
            We note that for the crossed module of Lie groups $(\sfG\overset{\sft}{\longrightarrow}\unit,\id)$, the group $\sfG=\ker(\sft)$ is necessarily Abelian. A principal $(\sfG\overset{\sft}{\longrightarrow}\unit,\id)$-bundle is then an \uline{Abelian gerbe} with connection in the sense of Hitchin--Chatterjee~\cite{Hitchin:1999fh,Chatterjee:1998}, and for these, the fake flatness issue is absent. Recall that any of the more general bundle gerbes of~\cite{Murray:9407015} is stably isomorphic to a Hitchin--Chatterjee gerbe of~\cite{Murray:2007ps} and thus also, albeit indirectly, captured by our approach.
        \end{remark}
                
        \begin{remark}
            Note that because of \ref{rem:not_all_2_groups_have_adjustments}, not all principal 2-bundles admit an adjusted connection. We plan to address the question of existence of a connection for principal 2-bundles with adjusted higher structure group elsewhere.
        \end{remark}
        
        \paragraph{Higher Poincar\'e lemma.} We note that the evident higher form of the Poincar\'e lemma has been proved for principal 2- and 3-bundles with unadjusted curvatures:
        \begin{proposition}[{\cite[Theorem 3.4]{Demessie:2014ewa}}]
            Locally, (unadjusted) flat connections on principal 2-bundles are gauge-equivalent to the trivial connection.
        \end{proposition}
        \noindent Since the adjustment does not modify the gauge transformations for flat connections, it is clear that the Poincar\'e lemma continues to hold for adjusted connections.
        
        \subsection{Unadjusted differential cocycles}
        
        Let us briefly compare the adjusted differential cocycles to the commonly used, unadjusted ones and point out the problems with these.
        
        \paragraph{Unadjusted cocycles.} Putting $\kappa=0$ in~\eqref{eq:adjustedCocycleConditionsB} recovers the cocycle relations for principal 2-bundles commonly found in the literature. These are a special case of the original definitions in~\cite{Breen:math0106083,Aschieri:2003mw}, which also contained an additional map
        \begin{equation}
            \delta\ \in\ \Omega^2(Y^{[2]},\frh)
        \end{equation}
        giving a deformation of the above cocycle relations. This additional map is put to zero in most of the subsequent literature\footnote{see e.g.~\cite{Schreiber:2008aa,Saemann:2012uq}} for a number of reasons. One of these is the fact that this additional map corresponds to an unexpected, additional higher gauge transformation which is absent in all natural local and infinitesimal pictures.\footnote{A closer look at the datum $\delta$ shows that it is akin to an (unwanted) non-trivial curvature for $\Lambda$.} The adjustment datum $\kappa$ can be interpreted as a mechanism of providing the arbitrary element $\delta$ in terms of expected cocycle data.

        \paragraph{Fake flatness.} We note that for connections that are fake-flat, i.e.~$F=0$, the term containing $\kappa$ vanishes, and there is no difference between adjusted and unadjusted cocycles. Fake flatness has been argued to be vital for consistency of the higher parallel transport in~\cite{Schreiber:0705.0452,Schreiber:2008aa}, cf.~also the discussion in~\cite{Kim:2019owc}. 
        
        Here, we can see this condition emerge in particular from the consistency of the gluing relations on $Y^{[3]}$,
        \begin{equation}\label{eq:tripleGaugeTransformations}
            \begin{tikzcd}
                & (A_2,B_2)\arrow[dr,"{(g_{23},\Lambda_{23})}"]\arrow[d,Rightarrow,"{h_{123}}"]
                \\
                (A_1,B_1) \arrow[ur,"{(g_{12},\Lambda_{12})}"] \arrow[rr,"{(g_{13},\Lambda_{13})}",swap] & \phantom{g} & (A_3,B_3) 
            \end{tikzcd}
        \end{equation}
        
        A quick computation using~\eqref{eq:adjustedCocycleConditions} shows that the commutativity of the diagram~\eqref{eq:tripleGaugeTransformations} amounts to
        \begin{equation}\label{eq:fakeCurvatureConditionFromTriangle}
            (g_{23}^{-1}g_{12}^{-1})\acton (h_{123}^{-1}(F_1\acton h_{123}))\ =\ 0~,
        \end{equation}
        see \cref{app:cocycleConsistency} for details. For a generic crossed module of Lie groups $\caG$ and generic datum $h$, this requires $F=0$. 
        
        The ramifications of fake flatness, however, are rather severe. Most importantly, we have the following result, which impedes their application in many physical situations. 
        \begin{proposition}[{\cite[Theorem 4.1]{Gastel:2018joi}, \cite[Theorem 4.3]{Saemann:2019dsl}}]\label{prop:implications_fake_flatness}
            Fake-flat connections on principal 2-bundles can be gauge-transformed locally to a connection on an Abelian gerbe.
        \end{proposition}
        
        \begin{remark}
            We stress that the evident analogue of \ref{prop:implications_fake_flatness} is obviously false for principal bundles. From a physical perspective, this would mean that in any pure gauge theory one could choose some abelian gauge such that gauge bosons become non-interacting, which is certainly not the case.
            
            In many situations in string theory or rather supergravity, the 1-form and 2-form parts of the connection are expected to mix in the curvature forms, e.g.~in the above mentioned case of heterotic supergravity.
            \Cref{prop:implications_fake_flatness} then shows that fake-flat connections on principal 2-bundles are insufficient for describing these. Moreover, interesting explicit cases of such physical connections on principal 2-bundles are simply not fake-flat, and we discuss a particular example in \ref{ssec:nonAbelianSelfDualString}.
        \end{remark}
        
        \section{Examples: Principal bundles as adjusted higher principal bundles}\label{sec:principal_G_bundles_as_adjusted_higher_bundles}
        
        As a very instructive example, let us explain how ordinary principal $\sfG$-bundles with connection for $\sfG$ some Lie group can be seen as principal 2-bundles with adjusted connection and with the structure 2-group given by the adjusted crossed module of Lie groups $\caL\sfG\cong\sfG_{\rm cm}$, see~\eqref{eq:definitionLieGroupLieCrossed} and~\eqref{eq:crossedModuleLg}.\footnote{Note that in the case of $\sfG_{\rm cm}$, condition~\eqref{eq:fakeCurvatureConditionFromTriangle} is trivial, and there is no need to impose a fake-flatness condition.}
        
        In this example of a special adjustment, we will see in particular the freedom involved in choosing and adjustment (in particular, adjustments are not unique) and how the equivalence of differential cocycles works in the case of $P$-flat connections.
        
        Beyond its instructive value, this example also underlies the later discussion of string structures in \ref{sec:stringBundlesWithConnection}.
        
        \subsection{General setting}\label{sec:principalAdjustedHigherPrincipal}
        
        The differential cocycles for a principal $\caL\sfG_\kappa$-bundle with $\kappa$ as in~\eqref{eq:adjustmentPathLoopGroups} is given in \cref{def:adj_cocycles}. The adjusted curvatures~\eqref{eq:adjustedCurvatures} reduce to
        \begin{equation}\label{eq:curvature_LG_kompressed}
            F_i\ =\ \rmd A_i+\tfrac12[A_i,A_i]+B_i
            \eand
            H_i\ =\ (\id-\wp\cdot\boundary)(\rmd F_i)
        \end{equation} 
        for all $(i)\in Y^{[1]}$ which, in turn, allow us to rewrite the adjusted Bianchi identities~\eqref{eq:adjustedBianchiIdentities} as
        \begin{equation}
            \nabla_i F_i\ =\ H_i+\kappa(A_i,F_i)
            \eand
            \rmd H_i\ =\ 0
        \end{equation}
        for all $(i)\in Y^{[1]}$. Finally, the $P$-flatness condition~\eqref{eq:adjustedFakeCurvatureCondition} becomes
        \begin{equation}\label{eq:adjustedFakeCurvatureConditionPathLoopGroups}
            (\id-\wp\cdot\boundary)(F_i)\ =\ 0
            \quad\Rightarrow\quad
            H_i\ = 0
        \end{equation}
        for all $(i)\in Y^{[1]}$. We then expect the following statement.
        
        \begin{proposition}
            The differential cohomology of principal $\sfG$-bundles with connections and the differential cohomology of principal $\caL\sfG$-bundles with adjusted connections satisfying the $P$-flatness condition~\eqref{eq:adjustedFakeCurvatureConditionPathLoopGroups} are isomorphic.
        \end{proposition}
        
        \begin{proof}
            We shall establish this claim in both directions in two steps. Firstly, we map the cocycle for one type of bundle to the cocycle for the other type. We then show that isomorphism classes are indeed mapped to isomorphism classes. 
            
            \vspace{5pt}
            \noindent
            \uline{Principal bundles $\Rightarrow$ higher principal bundles:} Consider the cocycle describing a principal $\sfG$-bundle over $X$ subordinate to a good cover $Y\rightarrow X$,
            \begin{subequations}\label{eq:descentDataOrdinaryPrincipalBundle}
                \begin{equation}
                    g\ \in\ \scC^\infty(Y^{[2]},\sfG)
                    \eand
                    A\ \in\ \Omega^1(Y^{[1]},\frg)
                \end{equation}
                with
                \begin{equation}
                    g_{ik}\ =\ g_{ij}g_{jk}
                    \eand
                    A_j\ =\ g_{ij}^{-1}A_i g_{ij}+g_{ij}^{-1}\rmd g_{ij}
                \end{equation}
            \end{subequations}
            for all appropriate $(i,j,\ldots)\in Y^{[n]}$. Since the endpoint evaluation map $\boundary:P_0\sfG\rightarrow \sfG$ defined in~\eqref{eq:endpointEvaluation} is a surjective submersion, the pullbacks along $\boundary$ exist, and since also $Y^{[2]}$ is the disjoint union of contractible subsets of $X$ by assumption, there exists a lift of $g\in\scC^\infty(Y^{[2]},\sfG)$, i.e.~a map $g^\circ\in\scC^\infty(Y^{[2]},P_0\sfG)$ such that 
            \begin{equation}
                \begin{tikzcd}
                    & P_0\sfG \arrow[d,"\boundary"]
                    \\
                    Y^{[2]} \arrow[ru,dotted,"g^\circ"] \arrow[r,"g"] & \sfG
                \end{tikzcd}
            \end{equation}
            is commutative.\footnote{Because $Y$ is a good cover, $Y^{[2]}$ is a disjoint union of contractible spaces, and hence the usual lifting lemma (e.g.~\cite[Lemma 79.1]{Munkres:2014aa}) works on each connected component of $Y^{[2]}$.}
            
            The map $g^\circ$ only satisfies the cocycle condition for a principal $P_0\sfG$-bundle up to elements in $L_0\sfG$. To extend this to a cocycle for a principal $\caL\sfG$-bundle with adjusted connection, we hence introduce $h^\circ\in\scC^\infty(Y^{[3]},L_0\sfG)$ by 
            \begin{subequations}\label{eq:lifted_cocycles_to_LSpin}
                \begin{equation}
                    h^\circ_{ijk}\ \coloneqq\ g^\circ_{ik}(g^\circ_{jk})^{-1}(g^\circ_{ij})^{-1}
                    \eforall(i,j,k)\ \in\ Y^{[3]}~.
                \end{equation}
                This indeed defines a loop in $\sfG$, that is, $\boundary(h^\circ_{ijk})=\unit$, because of the gluing condition in~\eqref{eq:descentDataOrdinaryPrincipalBundle}. It is possible and convenient (but not necessary) to define the lift $g^\circ$ in such a way that
                \begin{equation}
                    g^\circ_{ii}\ =\ \unit
                    \eand
                    g^\circ_{ij}\ =\ (g^\circ_{ji})^{-1}
                \end{equation}
                with $\unit\in P_0\sfG$ the constant path, which also implies the normalisations
                \begin{equation}
                    h^\circ_{iij}\ =\ h^\circ_{ijj}\ =\ h^\circ_{iji}\ =\ \unit~.
                \end{equation}
                For the connection forms, we generalise the usual prescription for transporting Maurer--Cartan elements in an $L_\infty$-algebra along quasi-isomorphisms, cf.~\cite{Saemann:2017rjm}:
                \begin{equation}\label{eq:connection_loop_embedded}
                    A^\circ_i\ \coloneqq\ \wp\cdot A_i
                    \eand
                    B^\circ_i\ \coloneqq\ \tfrac12(\wp-\wp^2)\cdot[A_i,A_i]
                    \eforall
                    (i)\ \in\ Y^{[1]}
                \end{equation}
                so that
                \begin{equation}\label{eq:Lambda-lift}
                    \Lambda_{ij}^\circ\ \coloneqq\ (g_{ij}^\circ)^{-1}A_i^\circ g_{ij}^\circ+(g_{ij}^\circ)^{-1}\rmd g_{ij}^\circ-A_j^\circ
                    \eforall
                    (i,j)\ \in\ Y^{[2]}~.
                \end{equation}
            \end{subequations}        
            The right-hand side here is a loop in $\frg$ due to the gluing relation~\eqref{eq:descentDataOrdinaryPrincipalBundle} between $A_i$ and $A_j$. It is now not too difficult to verify that with all these definitions, all the cocycle conditions~\eqref{eq:adjustedCocycleConditions} including the $P$-flatness condition~\eqref{eq:adjustedFakeCurvatureConditionPathLoopGroups} are satisfied. Altogether, we have obtained the cocycle $(h^\circ,g^\circ,\Lambda^\circ,A^\circ,B^\circ)$ of a principal $\caL\sfG$-bundle with adjusted connection satisfying the $P$-flatness condition.
            
            \vspace{5pt}
            \noindent
            \uline{Principal bundles $\Leftarrow$ higher principal bundles:}
            Conversely, we can apply the endpoint evaluation~\eqref{eq:endpointEvaluation} to project the cocycle $(h^\circ,g^\circ,\Lambda^\circ,A^\circ,B^\circ)$ for a principal $\caL\sfG$-bundle over $X$ with adjusted connection and with the $P$-flatness condition~\eqref{eq:adjustedFakeCurvatureConditionPathLoopGroups} imposed to the data $(g,A)$ of a principal $\sfG$-bundle by means of
            \begin{equation}\label{eq:recover_ordinary_cocycles}
                g_{ij}\ \coloneqq\ \boundary(g_{ij}^\circ)
                \eand
                A_i\ \coloneqq\ \boundary(A_i^\circ)
            \end{equation}
            for all appropriate $(i,j,\ldots)\in Y^{[n]}$. Again, it follows that this data fulfils the required cocycle conditions~\eqref{eq:descentDataOrdinaryPrincipalBundle}.
            
            \vspace{5pt}
            \noindent
            \uline{Gauge and higher gauge equivalence:}
            Above, we have established morphisms $\phi:(g,A)\mapsto(h^\circ,g^\circ,\Lambda^\circ,A^\circ,B^\circ)$ and $\psi:(h^\circ,g^\circ,\Lambda^\circ,A^\circ,B^\circ)\mapsto(g,A)$. Evidently, $\psi\circ\phi=\id$, and it remains to show that there is a bundle isomorphism $\phi\circ\psi\Rightarrow\id$. Consider some cocycle $(h^\circ,g^\circ,\Lambda^\circ,A^\circ,B^\circ)$ and set
            \begin{equation}
                (\tilde h^\circ,\tilde g^\circ,\tilde\Lambda^\circ,\tilde A^\circ,\tilde B^\circ)\ \coloneqq\ (\phi\circ\psi)(h^\circ,g^\circ,\Lambda^\circ,A^\circ,B^\circ)~.
            \end{equation}
            We note that there is a coboundary $(a^\circ,b^\circ,\lambda^\circ)$ with
            \begin{equation}\label{eq:coboundaryAlpha}
                a^\circ_i\ \coloneqq\ \unit~,
                \quad
                b^\circ_{ij}\ \coloneqq\ \tilde g^\circ_{ij}(g^\circ_{ij})^{-1}~,
                \eand
                \lambda^\circ_i\ \coloneqq\ (\id-\wp\cdot\boundary)(A^\circ_i)
            \end{equation}
            for all appropriate $(i,j,\ldots)\in Y^{[n]}$ connecting $(h^\circ,g^\circ,\Lambda^\circ,A^\circ,B^\circ)$ with $(\tilde h^\circ,\tilde g^\circ,\tilde\Lambda^\circ,\tilde A^\circ,\tilde B^\circ)$. In particular, the difference between $g^\circ_{ij}$ and $\tilde g^\circ_{ij}$ is a loop and thus in the image of the bundle isomorphism with the given $b^\circ_{ij}$. Furthermore, the maps $\tilde h^\circ_{ijk}$ are fully determined by the maps $\tilde g^\circ_{ij}$ because $\sft$ is injective. Similarly, the difference between $A^\circ_i$ and $\tilde A^\circ_i$ is a loop and thus in the image of the bundle isomorphism with the given $\lambda^\circ_i$, which then fully determines the maps $\tilde \Lambda^\circ_{ij}$. The $P$-flatness condition~\eqref{eq:adjustedFakeCurvatureConditionPathLoopGroups} moreover fully determines $\tilde B^\circ_i$ in terms of $\tilde A^\circ_i$ again because $\sft$ is injective. Hence, for cocycles satisfying the $P$-flatness condition~\eqref{eq:adjustedFakeCurvatureConditionPathLoopGroups}, the coboundary data~\eqref{eq:coboundaryAlpha} indeed provide a bundle isomorphism $\phi\circ\psi\Rightarrow\id$. In fact, the full gauge orbits are mapped to each other:
            \begin{subequations}
                \begin{equation}
                    \begin{tikzcd}[row sep=2cm,column sep=3cm]
                        (g,A)\arrow[r,bend left=20,"\phi"]\arrow[d,"a"] & (h^\circ,g^\circ,\Lambda^\circ,A^\circ,B^\circ)\arrow[l,bend left=20,"\psi"{below}]\arrow[d,bend left=30,"{(\tilde a^\circ,\tilde b^\circ,\tilde\lambda^\circ)}"{name=U,right}]
                        \arrow[d,bend right=30,"{(a^\circ,b^\circ,\lambda^\circ)}"{name=D,left}]
                        \\
                        (\tilde g,\tilde A)\arrow[r,bend left=20,"\phi"] & (\tilde h^\circ,\tilde g^\circ,\tilde\Lambda^\circ,\tilde A^\circ,\tilde B^\circ)\arrow[l,bend left=20,"\psi"{below}]\arrow[Rightarrow,"m^\circ", from=D, to=U,start anchor={[xshift=1ex]},end anchor={[xshift=-1ex]}]
                    \end{tikzcd}
                \end{equation}
                Explicitly, $a^\circ$ and $\tilde a^\circ$ are lifts of $a$ such that 
                \begin{equation}\label{eq:boundary_kond_1}
                    \boundary(a^\circ_i)\ =\ a_i\ =\ \boundary(\tilde a_i^\circ)
                    \eforall
                    (i)\ \in\ Y^{[1]}~.
                \end{equation}
                The maps $b^\circ$ and $\tilde b^\circ$ are then determined by the explicit form of $\tilde g^\circ$,
                \begin{equation}
                    b^\circ_{ij}\ \coloneqq\ a^\circ_i\tilde g^\circ_{ij}a^{\circ-1}_jg_{ij}^{\circ-1}
                    \eand
                    \tilde b^\circ_{ij}\ \coloneqq\ \tilde a^\circ_i\tilde g^\circ_{ij}\tilde a^{\circ-1}_j(g_{ij}^\circ)^{-1}
                    \eforall
                    (i,j)\ \in\ Y^{[2]}~,
                \end{equation}
                where the expressions on the right-hand sides are indeed loops due to~\eqref{eq:boundary_kond_1}. Furthermore, comparing $A^\circ$ with $\tilde A^\circ$ fixes the maps $\lambda^\circ$ and $\tilde \lambda^\circ$,
                \begin{equation}
                    \begin{aligned}
                        \lambda^\circ_i\ &\coloneqq\ (a_i^\circ)^{-1}A^\circ_ia^\circ_i+(a_i^\circ)^{-1}\rmd a^\circ_i-\tilde A^\circ_i~,
                        \\
                        \tilde \lambda^\circ_i\ &\coloneqq\ (\tilde a^\circ_i)^{-1} A^\circ_i\tilde a^\circ_i+(\tilde a^\circ_i)^{-1}\rmd\tilde a^\circ_i-\tilde A^\circ_i
                    \end{aligned}
                \end{equation}
            \end{subequations}
            for all $(i)\in Y^{[1]}$. Having fixed the lifted gauge transformations, one can directly verify that the coboundary relations between the cocycle $(h^\circ,g^\circ,\Lambda^\circ,A^\circ,B^\circ)$ and $(\tilde h^\circ,\tilde g^\circ,\tilde\Lambda^\circ,\tilde A^\circ,\tilde B^\circ)$ hold.
            
            It remains to check that the two constructed gauge transformations are linked by a higher gauge transformation~\eqref{eq:adjustedHigherGaugeTransformations}. We put 
            \begin{equation}
                m^\circ_i\ \coloneqq\ \tilde a^\circ_i(a^\circ_i)^{-1}
                \eforall
                (i)\ \in\ Y^{[1]}
            \end{equation}
            which ensures that $a^\circ$ and $\tilde a^\circ$ are related appropriately. The corresponding relations between $b^\circ$ and $\tilde b^\circ$ as well as $\lambda^\circ$ and $\tilde\lambda^\circ$ are then straightforwardly checked.
        \end{proof}
        
        \subsection{The instanton--anti-instanton pair as an \texorpdfstring{$\caL\sfSpin(4)$}{LSpin(4)}-bundle}\label{sec:Spin4Bundle}
        
        \paragraph{Motivation.} Let us also give a concrete and explicit example of how a principal $\sfG$-bundle can be equivalently described as a principal $\caL\sfG$-bundle. For physical applications, an interesting candidate is certainly the principal $\sfSU(2)$-bundle over $S^4$ underlying the elementary instanton solution on $S^4$. Recall that this bundle is the principal $\sfSU(2)$-bundle given by the quaternionic Hopf fibration $S^3\hookrightarrow S^7\rightarrow S^4$. This bundle has a non-vanishing second Chern class whose image in de~Rham cohomology integrates to the unit instanton charge. 
        
        Because $\sfSpin(4)\cong \sfSU(2)\times \sfSU(2)$, it is very natural to extend this bundle to that of an instanton--anti-instanton pair, which is the $\sfSpin(4)\cong \sfSU(2)\times \sfSU(2)$-bundle 
        \begin{equation}\label{eq:Spin4-bundle}
            \sfSU(2)\times\sfSU(2)\ \hookrightarrow\ P\ \coloneqq\ \sfSpin(5)\ \rightarrow\ \sfSpin(5)/\sfSpin(4)\ \cong\ \sfSO(5)/\sfSO(4)\ \cong\ S^4~.
        \end{equation}
        This gives the unique spin structure on $S^4$, cf.~\cite{Dabrowski:1986en}.
        
        The connections of both the elementary instanton and the instanton--anti-instanton pair can be obtained as the canonical connection on the corresponding reductive homogeneous spaces; see \cref{app:YMCoset} for a brief review of this construction. The resulting connection on the bundle~\eqref{eq:Spin4-bundle} turns out to consist of the sum of an $\sfSU(2)$-instanton and an $\sfSU(2)$-anti-instanton.
        
        This extension is also natural from a physical perspective. As explained in~\cite{Saemann:2017rjm}, self-dual gauge configurations are part of non-Abelian self-dual strings that naturally arise when M2-branes end on M5-branes. Moreover, the quaternionic Hopf fibration underlying the elementary instanton also features prominently in the so-called hypothesis H, cf.~e.g.~\cite{Fiorenza:2019usl}, which states that the charge quantisation in M-theory happens in a particular cohomology theory.

        \paragraph{$S^4$ as the coset $\sfSpin(5)/\sfSpin(4)$.}
        The explicit description of the fibration~\eqref{eq:Spin4-bundle} is found e.g.~in~\cite[Diagram 24.4]{Porteous:1995eh}, and we briefly review this description in the following. Recall that 
        \begin{equation}\label{eq:def_Sp2}
            \sfSpin(5)\ \cong\ \sfSp(2)\ =\ \sfU(2,\IH)\ \coloneqq\ \left\{g\in\sfMat(2,\IH)\,|\,g^\dagger g=gg^\dagger=\unit_2\right\},
        \end{equation}
        and thus, the elements of $\sfSpin(5)$ can be identified with matrices
        \begin{equation}\label{eq:Spin5-elements}
            g\ =\ \begin{pmatrix}
                a & b
                \\
                c & d
            \end{pmatrix},
            \quad
            a,b,c,d\in\IH~,
            \quad
            \begin{array}{ll}
                |a|^2+|b|^2\ =\ 1~,&a\bar c+b\bar d\ =\ 0~,
                \\
                |c|^2+|d|^2\ =\ 1~,&c\bar a+d\bar b\ =\ 0~,
                \\
                |a|^2+|c|^2\ =\ 1~,&\bar a b+\bar c d\ =\ 0~,
                \\
                |b|^2+|d|^2\ =\ 1~,&\bar b a+\bar d c\ =\ 0~.
            \end{array}
        \end{equation}
        Note that the above eight relations contain six independent conditions, so that $\sfSpin(5)$ is ten-dimensional. We further identify $\sfSU(2)$ with the unit quaternions $\sfSp(1)\cong\sfU(1,\IH)$, and consider the embedding
        \begin{equation}
            \begin{aligned}
                \sfSU(2)\times\sfSU(2)\ &\hookrightarrow\ \sfSpin(5)~,
                \\
                \left(\frac{e}{|e|},\frac{f}{|f|}\right)\ &\mapsto\ 
                \begin{pmatrix}
                    \frac{e}{|e|} & 0 
                    \\
                    0 & \frac{f}{|f|}
                \end{pmatrix}.
            \end{aligned}
        \end{equation}
        We then have a canonical projection $\sfSpin(5)\rightarrow\sfSpin(5)/\sfSpin(4)$, where we identify
        \begin{equation}\label{eq:Spin4Action}
            \begin{pmatrix}
                a & b
                \\
                c & d
            \end{pmatrix}\ \sim\ 
            \begin{pmatrix}
                a & b
                \\
                c & d
            \end{pmatrix}
            \begin{pmatrix}
                \frac{e}{|e|} & 0 
                \\
                0 & \frac{f}{|f|}
            \end{pmatrix}
            \eforall
            e,f\ \in\ \IH~.
        \end{equation}
        Given an element $g\in\sfSpin(5)$ of the form~\eqref{eq:Spin5-elements}, we note that the product
        \begin{equation}
            \begin{pmatrix}
                a & b
                \\
                c & d
            \end{pmatrix}
            \begin{pmatrix}
                1 & 0
                \\
                0 & -1
            \end{pmatrix}
            \begin{pmatrix}
                a & b
                \\
                c & d
            \end{pmatrix}^\dagger
            \ = \
            \begin{pmatrix}
                2|a|^2-1 & 2a\bar c
                \\
                2 c\bar a & 2|a|^2-1
            \end{pmatrix}
        \end{equation}
        is invariant under the right-action~\eqref{eq:Spin4Action} of $\sfSpin(4)$ on $g$. Moreover, we can identify 
        \begin{equation}
            x^1+\rmi x^2+\rmj x^3+\rmk x^4\ \coloneqq\ 2 c\bar a~\eand
            x^5\ \coloneqq\ 2|a|^2-1~,
        \end{equation}
        where $\rmi$, $\rmj$, and $\rmk$ are the standard quaternionic units and $x^1,\ldots,x^5\in\IR$. Because of the constraints in~\eqref{eq:Spin5-elements}, we obtain $\sum_{i=1}^5(x^i)^2=1$, that is, $S^4\hookrightarrow\IR^5$. This yields a projection 
        \begin{equation}
            \pi\,:\,\sfSpin(5)\ \rightarrow\ S^4~,
        \end{equation}
        which we shall use below. 
        
        Similarly, we have a projection $\sfSpin(5)\rightarrow S^7\cong\sfSpin(5)/\sfSU(2)$ that is given by 
        \begin{equation}
            \begin{pmatrix}
                a & b
                \\
                c & d
            \end{pmatrix}
            \ \mapsto\ (a,c)~,
        \end{equation}
        where we regard $S^7$ as the unit vectors in $\IH^2\cong \IR^8$ and the quotient group $\sfSU(2)$ is identified with the second factor in $\sfSpin(4)\cong \sfSU(2)\times \sfSU(2)$.
        
        \paragraph{Instantons on $S^4$.}
        We write elements of $\frspin(5)$ as quaternionic matrices
        \begin{equation}
            X\ =\ 
            \begin{pmatrix}
                x & z
                \\
                -\bar z & y
            \end{pmatrix}
            \eforall
            x,y,z\ \in\ \IH
            \ewith
            \Re(x)\ =\ \Re(y)\ =\ 0~.
        \end{equation}
        Here, $\frspin(4)\cong\frsu(2)\oplus\frsu(2)$ is identified with the two imaginary quaternions $x$ and $y$, which yields the decomposition 
        \begin{subequations}
            \begin{equation}\label{eq:split_spin5}
                \frspin(5)\ \cong\ \frspin(4)\oplus\frm~,
                \quad
                \begin{pmatrix}
                    x & z
                    \\
                    -\bar z & y
                \end{pmatrix}
                \ =\
                \begin{pmatrix}
                    x & 0
                    \\
                    0 & y
                \end{pmatrix}
                +
                \begin{pmatrix}
                    0 & z
                    \\
                    -\bar z & 0
                \end{pmatrix}
            \end{equation}
            with the evident relations
            \begin{equation}\label{eq:algebraic structure}
                [\frspin(4),\frspin(4)]\ \subseteq\ \frspin(4)~,
                \quad
                [\frspin(4),\frm]\ \subseteq\ \frm~,
                \eand
                [\frm,\frm]\ \subseteq\ \frspin(4)~.
            \end{equation}
        \end{subequations}
        Our discussion in \cref{app:YMCoset} then yields the canonical connection on the bundle $\sfSpin(5)\rightarrow S^4$. As we shall see below, this connection is the direct sum of a fundamental instanton and a fundamental anti-instanton.
        
        \paragraph{Clutching construction.}
        Following our discussion in \cref{app:YMCoset}, we can describe the principal $\sfSpin(4)$-bundle $\sfSpin(5)\rightarrow S^4$ subordinate to the cover given by the bundle itself with the transition function 
        \begin{equation}\label{eq:transition_function_1}
            \begin{aligned}
                g\,:\,(\sfSpin(5))^{[2]}\ &\rightarrow\ \sfSU(2)\times\sfSU(2)~,
                \\
                (g_1,g_2)\ &\mapsto\ g_1^{-1}g_2~,
            \end{aligned}
        \end{equation}
        where we have again used the notation~\eqref{eq:fibreProducts}.
        
        In order to make contact with the usual clutching construction of principal bundles over $S^4$, we embed $S^4$ into $\IR^5$, set $U_\pm\coloneqq S^4\setminus\{x^5=\pm1\}$, and derive a cover $U_+\sqcup U_-\to S^4$ from the stereographic projections
        \begin{subequations}\label{eq:stereographicProjections}
            \begin{equation}
                \begin{aligned}
                    \pi_\pm\,:\,U_\pm\ &\rightarrow\ \pi_\pm(U_\pm)\ \cong\ \IH~,
                    \\
                    (x^1,\ldots,x^5)\ &\mapsto\ q_\pm\ \coloneqq\ \frac{1}{1\mp x^5}(x^1\pm\rmi x^2\pm\rmj x^3\pm\rmk x^4)~,
                \end{aligned}
            \end{equation}
            where, as before, $\rmi$, $\rmj$, and $\rmk$ are the standard quaternionic units. Consequently, $q_+q_-=1$ on $\pi_+(U_+\cap U_-)\cong\pi_+(U_+\cap U_-)\cong\IH\setminus\{0\}$.
        \end{subequations}
        We then consider the commutative diagram
        \begin{equation}
            \begin{tikzcd}
                U_+\sqcup U_-\arrow[rr,"\phi"]\arrow[rd] & &\sfSpin(5)\arrow[ld] 
                \\
                & S^4
            \end{tikzcd}
        \end{equation}
        where $\phi$ reads in local coordinates as
        \begin{equation}\label{eq:comp1}
            \begin{gathered}
                (\pi_\pm^{-1})^*\phi\,:\,
                \pi_\pm(U_\pm)\ \rightarrow\ \sfSpin(5)~,
                \\
                q_+\ \mapsto\ \frac{1}{\sqrt{1+|q_+|^2}}
                \begin{pmatrix}
                    \bar q_+ & 1 
                    \\
                    1 & -q_+
                \end{pmatrix}
                \eand
                q_-\ \mapsto\ \frac{1}{\sqrt{1+|q_-|^2}}
                \begin{pmatrix}
                    1 & q_-
                    \\
                    \bar q_- & -1
                \end{pmatrix}.
            \end{gathered}
        \end{equation}
        The pullback $g_{+-}\coloneqq\phi^*g$ of the transition function~\eqref{eq:transition_function_1} to $(U_+\sqcup U_-)^{[2]}$ then reads in local coordinates as 
        \begin{subequations}\label{eq:instanton_bundle}
            \begin{equation}
                \begin{aligned}
                    (\pi_\pm^{-1})^*g_{+-}|_{\pi_\pm(U_+\cap U_-)}\,:\,\pi_\pm(U_+\cap U_-)\ &\rightarrow\ \sfSU(2)\times\sfSU(2)~,
                    \\
                    q_\pm\ &\mapsto\ \left(\frac{q_+}{|q_+|},\frac{\bar q_+}{|q_+|}\right)\ =\ \left(\frac{\bar q_-}{|q_-|},\frac{q_-}{|q_-|}\right).
                \end{aligned}
            \end{equation}
            This principal bundle can be endowed with a connection given by the gauge potentials $A_\pm$ on $U_\pm$ with
            \begin{equation}
                A_-\ =\ g_{+-}^{-1}A_+g_{+-}+g_{+-}^{-1}\rmd g_{+-}
                \eon
                U_+\cap U_-~.
            \end{equation}
            This connection describes a pair of an elementary instanton with an elementary anti-instanton, which, in quaternionic notation, is given by
            \begin{equation}
                (\pi_\pm^{-1})^*A_\pm\ \coloneqq\ \frac12\left(\frac{q_\pm\rmd\bar q_\pm-\rmd q_\pm\bar q_\pm}{1+|q_\pm|^2},\frac{\bar q_\pm\rmd q_\pm-\rmd\bar q_\pm q_\pm}{1+|q_\pm|^2}\right)
                \eon
                \pi_\pm(U_\pm)~,
            \end{equation}
            cf.~\cite{Atiyah:1979iu}. Consequently, the curvature is
            \begin{equation}
                (\pi_\pm^{-1})^*F_\pm\ =\ \left(\frac{\rmd q_\pm\wedge\rmd\bar q_\pm}{(1+|q_\pm|^2)^2},\frac{\rmd\bar q_\pm\wedge\rmd q_\pm}{(1+|q_\pm|^2)^2}\right)
                \eon
                \pi_\pm(U_\pm)
            \end{equation}
        \end{subequations}
        with the first component describing an instanton with self-dual field strength and the second component describing an anti-instanton with anti-self-dual field strength. The total first Pontryagin class (as is evident from its image in de~Rham cohomology) vanishes, and this will become important later.
        
        \paragraph{Choice of cover for a lift.} In order to lift the cocycle data~\eqref{eq:instanton_bundle} of the above defined $\sfSpin(4)$-bundle $P$ to that of a principal $\caL\sfSpin(4)$-bundle, we have to choose a different cover.\footnote{For a recent application of based loop groups as gauge groups in Yang--Mills theory, see also~\cite{Popov:2015wsa}.}
        The cover $U_+\sqcup U_-$ introduced above is not suitable for a lift of $P$ to a principal $\caL\sfSpin(4)$-bundle, because the transition function $g_{+-}$ in~\eqref{eq:instanton_bundle} does not factor through $P_0\sfSpin(4)$. This is easy to see: $U_+\sqcup U_-\cong S^3\times \IR$, and the equivalence classes of transition functions are thus given by $\pi_3(\sfSpin(3)\times \sfSpin(3))\cong \IZ\times \IZ$, but the path space $P_0\sfSpin(4)$ is contractible, and hence its homotopy groups and in particular $\pi_3$ vanishes. Therefore, the topological information has to be encoded in the $L_0\sfSpin(4)$-valued function. Since we only have a 2-patch cover and since $U_+\cap U_-\cong S^3\times\IR$, the homotopy classes of this map are $\pi_3(L_0\sfSpin(4))\cong\pi_4(\sfSpin(4))\cong\pi_4(S^3\times S^3)\cong\IZ_2\oplus\IZ_2$, which, evidently is not large enough.\footnote{The identification $\pi_{n}(\sfG)\cong\pi_{n-1}(L_0\sfG)$ is as follows. Let $X$ be a topological space and $p_0\in X$ fixed. The \uline{reduced suspension} of $(X,p_0)$ is the quotient $\Sigma X\coloneqq(X\times[0,1])/\!\!\sim$ with the equivalence relation given by $(p,0)\sim(p,1)\sim(p_0,t)$ for all $p\in X$ and for all $t\in[0,1]$. Then, for $Y$ another topological space, there is the natural identification $\scC_0(\Sigma X,Y)\cong\scC_0(X,L_0Y)$ of continuous maps that preserve the base point. In the context of spheres, we have a homeomorphism between $\Sigma S^{n-1}$ and $S^n$. Indeed, we mark poles $x_0\in S^{n-1}$ and $y_0\in S^n$ and identify $x_0$ as well as all points $S^{n-1}\times\{0,1\}$ with $y_0$. We then identify the remaining points $(0,1)$ with $\IR$ and map $S^{n-1}\setminus\{x_0\}$ to $\IR^{n-1}$ by stereographic projection at $x_0$. The Cartesian product $\IR^{n-1}\times\IR$ is then identified with $S^n\setminus\{y_0\}$ by inverse stereographic projection at $y_0$.\label{fn:reduced_suspension}}
        
        We can, however, construct a suitable 3-patch replacement cover of the 2-patch cover $U_+\sqcup U_-$ rather directly by hand. First of all, we note that $\pi_3(\sfSpin(4))\cong\pi_2(L_0\sfSpin(4))$ so that we would like our triple overlaps to be homeomorphic to $S^2\times\IR^2$. We can reduce $\pi_+(U_+\cap U_-)\cong\pi_-(U_+\cap U_-) \cong\IH\setminus\{0\}$ to such a space by removing a straight line through the origin from $\IH$ since, topologically, $\IH\setminus\IR\cong S^2\times\IR^2$. In particular, let
        \begin{equation}
            \begin{aligned}
                L_+\ &\coloneqq\ \{q\in\IH\,|\,\Re(q)>0\text{ and }\Im(q)=0\}~,
                \\
                L_-\ &\coloneqq\ \{q\in\IH\,|\,\Re(q)<0\text{ and }\Im(q)=0\}~,
            \end{aligned}
        \end{equation}
        and set $L\coloneqq L_-\cup\{0\}\cup L_+$; evidently $L\cong\IR$. Next, we define the 3-patch cover $U_1\sqcup U_2\sqcup U_3\rightarrow S^4$ by\footnote{Note that this is not an open cover. We can refine this cover to an open cover; however, to keep the formulas simple, we shall work with this 3-patch cover.}
        \begin{equation}\label{eq:3patchCover}
            \begin{gathered}
                (U_1,\pi_1)
                \ewith
                U_1\ \coloneqq\ \pi_+^{-1}(\underbrace{\pi_+(U_+)\setminus L_+}_{\cong\,\IH\setminus L_+})
                \eand
                \pi_1\ \coloneqq\ \pi_+|_{U_1}~,
                \\
                (U_2,\pi_2)
                \ewith
                U_2\ \coloneqq\ \pi_+^{-1}(\underbrace{\pi_+(U_+)\setminus L_-}_{\cong\,\IH\setminus L_-})
                \eand
                \pi_2\ \coloneqq\ \pi_+|_{U_2}~,
                \\
                (U_3,\pi_3)
                \ewith
                U_3\ \coloneqq\ U_-
                \eand
                \pi_3\ \coloneqq\ \pi_-~.
            \end{gathered}
        \end{equation}
        Note that $U_1\cup U_2=U_+$ and $U_1\cap U_2\cap U_3\cong S^2\times\IR^2$. With respect to this cover, the bundle $\sfSpin(5)\rightarrow S^4$ is described by transition functions
        \begin{subequations}\label{eq:cocycle_Y_tilde}
            \begin{equation}
                g_{12}\ \coloneqq\ -\unit~,
                \quad
                g_{13}\ \coloneqq\ -g_{+-}~,
                \eand
                g_{23}\ \coloneqq\ g_{+-}~,
            \end{equation}
            where we inserted a constant change of frame between patches $U_1$ and $U_2$. This does not affect the gauge potential, and we have
            \begin{equation}
                \begin{aligned}
                    (\pi_{1,2}^{-1})^*A_{1,2}\ &\coloneqq\ \frac12\left(\frac{q_{1,2}\rmd\bar q_{1,2}-\rmd q_{1,2}\bar q_{1,2}}{1+|q_{1,2}|^2},\frac{\bar q_{1,2}\rmd q_{1,2}-\rmd\bar q_{1,2} q_{1,2}}{1+|q_{1,2}|^2}\right),
                    \\
                    (\pi_3^{-1})^*A_3\ &\coloneqq\ \frac12\left(\frac{q_3\rmd\bar q_3-\rmd q_3\bar q_3}{1+|q_3|^2},\frac{\bar q_3\rmd q_3-\rmd\bar q_3 q_3}{1+|q_3|^2}\right),
                \end{aligned}
            \end{equation}
        \end{subequations}  
        with $q_{1,2,3}\in\pi_{1,2,3}(U_{1,2,3})$, where $q_{1,2}$ are restrictions of $q_+$ and $q_3=q_-$.
        
        \paragraph{Lifted cocycle data.} 
        The cocycle~\eqref{eq:cocycle_Y_tilde} can now be straightforwardly lifted, following the prescription in \cref{sec:adj_cocycles}. Relative to the cover~\eqref{eq:3patchCover}, we can define the following lift of the cocycles:
        \begin{subequations}
            \begin{equation}
                g^\circ_{12}\ \coloneqq\ \big(t\mapsto(\rme^{\rmi\pi \wp(t)},\rme^{-\rmi\pi \wp(t)})\big)
                \eand
                g^\circ_{i3}\ \coloneqq\ \big(t\mapsto(u_i(t),\bar u_i(t))\big)
            \end{equation}
            for all $i=1,2$ with
            \begin{equation}\label{eq:cocycles_LSpin_bundle_gui}
                u_i(t)\ \coloneqq\ \left(q_i \mapsto \frac{1-\wp(t)+(-1)^i\wp(t)q_i}{|1-\wp(t)+(-1)^i\wp(t)q_i|}\right)
            \end{equation}
            and for all $t\in[0,1]$, which induces
            \begin{equation}
                h^\circ_{123}\ \coloneqq\ g_{13}^\circ(g_{23}^\circ)^{-1}(g_{12}^\circ)^{-1}~.
            \end{equation}
            Furthermore, the lift of the connection reads as
            \begin{equation}
                \begin{gathered}
                    A^\circ_1\ \coloneqq\ \wp\cdot A_1~,
                    \quad
                    A^\circ_2\ \coloneqq\ \wp\cdot A_2~,
                    \quad
                    A^\circ_3\ \coloneqq\ \wp\cdot A_3~,
                    \\
                    B^\circ_1\ \coloneqq\ \tfrac12(\wp-\wp^2)\cdot[A_1,A_1]~,
                    \quad
                    B^\circ_2\ \coloneqq\ \tfrac12(\wp-\wp^2)\cdot[A_2,A_2]~,
                    \\
                    B^\circ_3\ \coloneqq\ \tfrac12(\wp-\wp^2)\cdot[A_3,A_3]~,
                \end{gathered}
            \end{equation}
            and
            \begin{equation}
                \begin{aligned}
                    \Lambda_{12}^\circ\ &\coloneqq\ \bigg(t\mapsto\wp(t)\cdot(\rme^{-\rmi\pi \wp(t)},\rme^{\rmi\pi \wp(t)}) \left[A_+,(\rme^{\rmi\pi \wp(t)},\rme^{-\rmi\pi \wp(t)})\right]\bigg)\,,
                    \\
                    \Lambda_{i3}^\circ\ &\coloneqq\ \bigg(t\mapsto\frac{1}{1+|q_+|^2}\bigg(\wp(t)|q_+|^2\bar u_i(t)u_i(1)\rmd\big(\bar u_i(1)u_i(t)\big)+|q_+|^2(1-\wp(t))\bar u_i(t)\rmd u_i(t)
                    \\
                    &\kern4.5cm+\bar u_i(t)\rmd u_i(t)-\wp(t)\bar u_i(1)\rmd u_i(1),u_i(t)\leftrightarrow\bar u_i(t)\bigg)\bigg)
                \end{aligned}
            \end{equation}
        \end{subequations}
        for all $i=1,2$ with $u_i$ as given in~\eqref{eq:cocycles_LSpin_bundle_gui} and for all $t\in[0,1]$. Here, $u_i(t)\leftrightarrow\bar u_i(t)$ denotes the same expression but with $u_i(t)$ and $\bar u_i(t)$ interchanged.
        
        This is the complete cocycle data of a non-trivial and non-Abelian gerbe, albeit one which is equivalent to an ordinary principal bundle. The original cocycles are recovered by endpoint evaluation of the lifted cocycles as detailed in~\eqref{eq:recover_ordinary_cocycles}.
        
        \paragraph{Alternative lift.}
        The above lift has the disadvantage of being computationally rather involved. With a view towards our later discussion, we therefore also introduce a simplifying lift. To this end, we replace the 3-patch cover $U_1\sqcup U_2\sqcup U_3$ by the 2-patch cover $\bar U_1\sqcup \bar U_2\rightarrow S^4$ with 
        \begin{equation}\label{eq:2patchCover}
            \begin{gathered}
                (\bar U_1,\bar\pi_1)
                \ewith
                \bar U_1\ \coloneqq\ \pi_+^{-1}(\pi_+(U_+)\setminus(L_+\cup L_-))
                \eand
                \bar\pi_1\ \coloneqq\ \pi_+|_{\bar U_1}
                \\
                (\bar U_2,\bar\pi_2)
                \ewith
                \bar U_2\ \coloneqq\ U_-
                \eand
                \bar\pi_2\ \coloneqq\ \pi_-~.
            \end{gathered}
        \end{equation}
        Note that $\bar U_1\cap\bar U_2\cong S^2\times\IR^2$. With respect to this cover, the bundle $\sfSpin(5)\rightarrow S^4$  is described by the transition function
        \begin{subequations}\label{eq:cocycle_Y_tilde2}
            \begin{equation}
                g_{12}\ \coloneqq\ g_{+-}\big|_{(\bar U_1\sqcup\bar U_2)^{[2]}}~.
            \end{equation}
            The local connection 1-form reads as
            \begin{equation}
                (\pi_{1,2}^{-1})^*A_{1,2}\ \coloneqq\ \left(\frac{\Im{(q_{1,2}\rmd\bar q_{1,2})}}{1+|q_{1,2}|^2},\frac{\Im{(\bar q_{1,2}\rmd q_{1,2})}}{1+|q_{1,2}|^2}\right),
            \end{equation}
        \end{subequations}  
        with $q_{1,2}\in\bar\pi_{1,2}(\bar U_{1,2})$, where $q_1$ is a restriction of $q_+$ and $q_2=q_-$.
        
        The cocycle~\eqref{eq:cocycle_Y_tilde2} can then be straightforwardly lifted, following the prescription in \cref{sec:adj_cocycles}. Relative to the cover~\eqref{eq:3patchCover}, we can define
        \begin{subequations}\label{eq:cocycles_LSpin_bundle}
            \begin{equation}
                g^\circ_{12}\ \coloneqq\ (u_{q},u_{\bar q})
                \eand
                g^\circ_{21}\ \coloneqq\ (v_{q},v_{\bar q})~,
            \end{equation}
            where $ q=q_+ $ and for all $ t \in [0,1] $
            \begin{equation}
                \begin{aligned}
                    u_q(t)\ &\coloneqq\ \cos(\theta_q\wp(t))+\frac{\Im(q)}{|\Im(q)|}\sin(\theta_q \wp(t))~,
                    \\
                    v_q(t)\ &\coloneqq\ \cos(\theta'_q\wp(t))+\frac{\Im(q)}{|\Im(q)|}\sin(\theta'_q \wp(t))
                \end{aligned}
            \end{equation}
            with
            \begin{equation}\label{eq:cocycles_LSpin_bundle_theta}
                \theta_q\ \coloneqq\ \frac{\pi}{2}-\arctan\left(\frac{\Re(q)}{|\Im(q)|}\right) \eand \theta'_q\ \coloneqq\ 2\pi - \theta_q ~. 
            \end{equation}
            This induces
            \begin{equation}
                \begin{aligned}
                    h^\circ_{121}\ &\coloneqq\ (g_{12}^\circ g_{21}^\circ)^{-1}
                    \\
                    &\kern2.5pt=\ \left(t\mapsto \left(\cos(2\pi \wp(t))+\frac{\Im(\bar q)}{|\Im(q)|}\sin(2\pi \wp(t)),\cos(2\pi \wp(t))+\frac{\Im(q)}{|\Im(q)|}\sin(2\pi \wp(t))\right)\right),
                    \\
                    h^\circ_{212}\ &\coloneqq\ (g_{21}^\circ g_{12}^\circ)^{-1}\ =\ h^\circ_{121}~,
                \end{aligned}
            \end{equation}
            where the last equality follows because $g^\circ_{12}$ and $g^\circ_{21}$ commute. We note that the above two expressions only depend on $\frac{\Im(q)}{|\Im(q)|}$, which describes the embeddings $S^2\hookrightarrow\IR^3\hookrightarrow\IR^4$, as expected from the abstract discussion involving reduced suspension mentioned previously.
            
            The lift of the connection reads as
            \begin{equation}
                \begin{gathered}\label{eq:cocycles_LSpin_bundle_AB2}
                    A^\circ_1\ \coloneqq\ \wp\cdot A_1~,
                    \quad
                    A^\circ_2\ \coloneqq\ \wp\cdot A_2~,
                    \\
                    B^\circ_1\ \coloneqq\ \tfrac12(\wp-\wp^2)\cdot[A_1,A_1]~,
                    \quad
                    B^\circ_2\ \coloneqq\ \tfrac12(\wp-\wp^2)\cdot[A_2,A_2]~,
                \end{gathered}
            \end{equation}
            and, for all $ t \in [0,1] $, 
            \begin{equation}
                \begin{aligned}
                    \Lambda_{12}^\circ\ &\coloneqq\ \bigg(t\mapsto\frac{1}{1+|q|^2}\bigg(|q|^2\wp(t)\bar u_q(t)u_q(1)\rmd\big(\bar u_q(1)u_q(t)\big)+|q|^2(1-\wp(t))\bar u_q(t)\rmd u_q(t)
                    \\
                    &\kern4.5cm+\bar u_q(t)\rmd u_q(t)-\wp(t)\bar u_q(1)\rmd u_q(1),u_q(t)\leftrightarrow\bar u_q(t)\bigg)\bigg)~.
                \end{aligned}
            \end{equation}
            Explicitly,
            \begin{equation}
                \begin{aligned}
                    \Lambda_{12}^\circ(t)\ &=\ \frac{1}{1+|q|^2}\left(Q_q(t) \frac{\Im\big(\Im(q)\rmd \Im(\bar q)\big)}{2|\Im(q)|^2},Q_{\bar q}(t)\frac{\Im\big(\Im(\bar q)\rmd \Im(q)\big)}{2|\Im(q)|^2}\right)
                    \\
                    &=\ \big(Q_q(t),Q_{\bar q}(t)\big)\frac{\Im\big(\Im(q)\rmd \Im(q)\big)}{2|\Im(q)|^2(1+|q|^2)}~,
                \end{aligned}                    
            \end{equation}   
            where
            \begin{equation}
                Q_q(t)\ \coloneqq\ \big(1+|q|^2+\wp(t)(q^2-|q|^2)\big)\bar u_q^2(t)-\left(1+|q|^2-\wp(t)\left(1-\frac{\bar q^2}{|q|^2}\right)\right).
            \end{equation}  
            Similarly, we find
            \begin{equation}
                \begin{aligned}
                    \Lambda_{21}^\circ\ &\coloneqq\ \bigg(t\mapsto\frac{1}{1+|q|^2}\bigg(\wp(t)\bar{v}_{q}(t)v_q(1)\rmd\big(\bar{v}_{q}(1)v_q(t)\big)+(1-\wp(t))\bar{v}_{q}(t)\rmd v_q(t)
                    \\
                    &\kern4.5cm+|q|^2(\bar{v}_{q}(t)\rmd v_q(t)-\wp(t)\bar{v}_{q}(1)\rmd v_q(1)),~v_q(t)\leftrightarrow\bar{v}_q(t)\bigg)\bigg)
                    \\
                    &\kern2.5pt=\ \bigg(t\mapsto\bigg(Q'_q(t),Q'_{\bar q}(t)\bigg)\frac{\Im\big(\Im(q)\rmd\Im(q)\big)}{2|\Im(q)|^2(1+|q|^2)}\bigg)\,,          
                \end{aligned}
            \end{equation}  
            where
            \begin{equation}
                Q'_q(t)\ \coloneqq\ \left(1+|q|^2-\wp(t)\left(1-\frac{\bar q^2}{|q|^2}\right)\right)\bar{v}_{q}^2(t)-\big(1+|q|^2+\wp(t)(q^2-|q|^2)\big)\,.
            \end{equation}  
        \end{subequations}
        The form of the lift of the gauge potentials $A$ and $B$ given in~\eqref{eq:cocycles_LSpin_bundle_AB2} was convenient from the local, infinitesimal perspective, and it leads to the general lifting formulas~\eqref{eq:lifted_cocycles_to_LSpin}. 
        
        \begin{remark}\label{rem:better_cover}
            We note that using a suitable path space as a cover may lead to much more natural description of the lift of the spin structure on $S^4$. This, however, may seem too esoteric to physicists, and we therefore prefer to work in this slightly more complicated setting.
        \end{remark}
        
        \section{Examples: String bundles with connections}\label{sec:stringBundlesWithConnection}
        
        We now turn to truly higher principal bundles with connection that are not mere reformulations of ordinary principal bundles with connections. We will focus on principal 2-bundles with structure 2-group a 2-group model of the string group which encode string structures, higher generalisations of spin structures; see \cref{app:group_structures} for a concise review. We will then show how to endow these bundles with adjusted connections. As an explicit example, we shall lift the $\sfSpin(4)$-structure on $S^4$ constructed in \cref{sec:Spin4Bundle} to a string structure, and endow the resulting principal 2-bundle with an adjusted connection.
        
        Lifting a generically non-flat spin bundle to a string bundle, we obtain a generically non-fake-flat principal 2-bundle, and hence we have to use adjusted connections. We note that at the local and infinitesimal level, such connections had first been identified in the context of supergravity~\cite{Bergshoeff:1981um,Chapline:1982ww} and before the invention of string structures. They were then put into the context of higher gauge theory in~\cite{Sati:2008eg,Waldorf:2009uf,Sati:2009ic,Fiorenza:2010mh} either at the local level, or in an abstract integrated form, or using an alternative perspective on string structures. The latter arises from the observation that isomorphism classes of string structures with connections on a principal bundle $P$ are in bijection with isomorphism classes of trivialisations of the Chern--Simons 2-gerbe with connection associated to $P$~\cite{Waldorf:2009uf}, and this set is a torsor for the Deligne cohomology group $H^3_{\rm D}(X,\IZ)$. For very recent work in this direction, see also~\cite{Tellez-Dominguez:2023wwr}, where our adjusted connections also emerged in this context.
        
        \subsection{Strict Lie 2-group model}\label{sec:Lie2GroupModel}
        
        As mentioned in \cref{app:group_structures}, the string group is defined up to $A_\infty$-equivalence, and various models exist. Here, we focus on the strict Lie 2-group model constructed in~\cite{Baez:2005sn}, see also~\cite{Carey:1989ck,Murray:2001xq} as well as~\cite[Section 4]{Pressley:1988qk} and~\cite[Section 4]{Mickelsson:1989hp} for discussions of the involved central extension. Recall that a string 2-group model for a compact, simply-connected, simple Lie group\footnote{We make this assumption for $\sfG$ for the remainder of this section.} $\sfG$ in the sense of~\cite{Nikolaus:2011zg} is a Lie 2-group $\caG=(\sfG_1\multirightarrow{2}\sfG_0)$ with a 3-connected cover $\pi:\caG\rightarrow\sfG$ such that the group of isomorphism classes of objects is isomorphic to $\sfG$ and the group of automorphisms of $\unit\in \sfG_0$ is isomorphic to $\sfU(1)$.
        
        \paragraph{String Lie 2-group as 2-group extension.}
        Consider the based path and loop groups defined in~\eqref{eq:basedPathLoopGroups}. Given the close relationship between string structures and spin structures on loop spaces and the fact that the spin group is a central extension of another group, it is perhaps not surprising that a 2-group model $\sfString(\sfG)$ for a Lie group $\sfG$ can be built as the central extension of 2-groups
        \begin{equation}\label{eq:centralExtension2Groups}
            \unit\ \longrightarrow\ \sfB\sfU(1)\ \longrightarrow\ \sfString(\sfG)\ \longrightarrow\ \caL\sfG\ \longrightarrow\ \unit~,
        \end{equation}
        where $\sfB\sfU(1)$ is the crossed module $(\sfU(1)\longrightarrow*,\id)$ and the 2-group $\caL\sfG$ is defined in~\eqref{eq:definitionLoopLieCrossed}. Let us now explain the construction of $\sfString(\sfG)$ in detail, reviewing and slightly expanding the discussion in~\cite{Baez:2005sn}, which, at the time of writing, contained some sign errors.
        
        \paragraph{Kac--Moody extension as principal circle bundle.}
        Our starting point is the \uline{Kac--Moody central extension}\footnote{See \cref{app:centralExtensions} for a brief review of central extensions.}
        \begin{equation}\label{eq:Kac-Moody-group-extension}
            \unit\ \longrightarrow\ \sfU(1)\ \overset{\iota}{\longrightarrow}\ \widehat{L_0\sfG}\ \overset{\pi}{\longrightarrow}\ L_0\sfG\ \longrightarrow\ \unit
        \end{equation}
        of $L_0\sfG$. To construct a group product on $\widehat{L_0\sfG}$, we follow~\cite{Murray:1987ua}, see also~\cite{Mickelsson:1987:173-183,Murray:2001eu}. In particular, it is convenient to regard the trivialisation of the principal $\sfU(1)$-bundle $\widehat{L_0\sfG}$ over the path space\footnote{All our definitions of based, parametrised path and loop spaces are the same as in \cref{sec:higherPrincipalBundles}.} $P_0L_0\sfG$. Here, we identify $L_0G\cong P_0L_0\sfG/L_0L_0\sfG$, generalising~\eqref{eq:identification}, and $\widehat{L_0\sfG}\cong \big(P_0L_0\sfG\times\sfU(1)\big)/\sfN$, with $N$ defined in the following proposition: 
        
        \begin{proposition}
            The subset $\sfN\subseteq L_0L_0\sfG\times\sfU(1)$ with
            \begin{equation}\label{eq:subgroupN}
                \sfN\ \coloneqq\ \left\{(f,z)\in L_0L_0\sfG\times\sfU(1)\,\middle|\,z=\sfhol^{-1}(f)=\exp\left(\int_{D_f}\omega\right)\right\},
            \end{equation}
            where $D_f$ is a disc in $L_0\sfG$ with boundary $f$, is a normal subgroup of $P_0L_0\sfG\times\sfU(1)$.
        \end{proposition}
        
        \begin{proof}
            The subgroup property follows because of
            \begin{equation}
                \left(f_1,\exp\left(\int_{D_{f_1}}\omega\right)\right)\left(f_2,\exp\left(\int_{D_{f_2}}\omega\right)\right)\ =\ \left(f_1f_2,\exp\left(\int_{D_{f_1f_2}}\omega\right)\right)
            \end{equation}
            due to 
            \begin{equation}
                c(f_1,f_2)\ =\ \exp\left(-\int_{D_{f_1}}\omega-\int_{D_{f_2}}\omega+\int_{D_{f_1f_2}}\omega\right),
            \end{equation}
            which, in turn, is a direct consequence of the definition~\eqref{eq:cocyclePLGU} specialised to loops $f_{1,2}\in L_0L_0\sfG$. 
            
            Moreover, the normality of $\sfN$ follows because of
            \begin{equation}
                \begin{aligned}
                    (f_1,z_1)(f_2,\sfhol^{-1}(f_2))(f_1,z_1)^{-1}\ &=\ (f_1f_2 f_1^{-1},\sfhol^{-1}(f_2)\,c(f_1,f_2)\,c(f_1f_2,f_1^{-1})\,c^{-1}(f_1,f_1^{-1}))
                    \\
                    &=\ (f_1f_2 f_1^{-1},\sfhol^{-1}(f_1f_2 f_1^{-1}))
                \end{aligned}
            \end{equation}
            for all $f_1\in P_0L_0\sfG$, for all $f_2\in L_0L_0\sfG$, and for all $z_1\in\sfU(1)$. Here, the last equality follows from combining the holonomy $\sfhol^{-1}(f_2)$ with the holonomies in the cocycles, comparing the boundaries and invoking Stokes' theorem, and using~\eqref{eq:holonomy_shift} once.
        \end{proof}
        
        Altogether, we obtain the following commutative diagram
        \begin{equation}\label{eq:KacMoodyExtensionLoopGroup}
            \begin{tikzcd}
                P_0L_0\sfG\times\sfU(1)\arrow[r,"\hat \boundary"] & \boundary^*\widehat{L_0\sfG}\arrow[r]\arrow[d] & \big(P_0L_0\sfG\times\sfU(1)\big)/\sfN\arrow[d]\arrow[r,"\cong"] & \widehat{L_0\sfG}\arrow[d,"\pi"]
                \\
                & P_0L_0\sfG\arrow[r]\arrow[rr,bend right=20,"\boundary"] & P_0L_0\sfG/L_0L_0\sfG\arrow[r,"\cong"] & L_0\sfG
            \end{tikzcd}
        \end{equation}
        where the pullback bundle $\boundary^*\widehat{L_0\sfG}$ along the endpoint evaluation map $\boundary$ is given by the fibre product
        \begin{equation}
            \boundary^*\widehat{L_0\sfG}\ \coloneqq \ P_0L_0\sfG~\times_{L_0\sfG}\widehat{L_0\sfG}~,
        \end{equation}
        and the isomorphism $\hat \boundary$ reads as
        \begin{equation}
            \hat \boundary (f,z)\ \coloneqq\ (f,\hat f(1)z)
        \end{equation}
        for all $(f,z)\in P_0L_0\sfG\times\sfU(1)$, where $\hat f$ is the horizontal lift\footnote{with respect to the connection constructed in \cref{app:proofs}} of $f\in P_0L_0\sfG$ to $\widehat{L_0\sfG}$ with $\hat f(0)=\unit\in \widehat{L_0\sfG}$.
        
        \paragraph{Multiplication on $\widehat{L_0\sfG}$.}
        In order to construct the multiplication on $\widehat{L_0\sfG}$, we require a suitable group cocycle $c$, cf.~\cref{app:centralExtensions}. Here, $c$ is a map
        \begin{equation}\label{eq:group_cocycle}
            c\,:\,P_0L_0\sfG\times P_0L_0\sfG\ \rightarrow\ \sfU(1)
            \ewith 
            c(f_1,f_2)c(f_1f_2,f_3)\ =\ c(f_1,f_2f_3)c(f_2,f_3)
        \end{equation}
        for all $f_{1,2,3}\in P_0L_0\sfG$, which results in the associative, unital product 
        \begin{equation}\label{eq:productOnP0L0G}
            (f_1,z_1)(f_2,z_2)\ \coloneqq\ \big(f_1f_2,z_1z_2c(f_1,f_2)\big)
        \end{equation}
        for all $(f_{1,2},z_{1,2})\in P_0L_0\sfG\times\sfU(1)$. We would like the group cocycle $c$ to arise as the integrated form of the Kac--Moody 2-cocycle on the Lie algebra $L_0\frg$ with $\frg$ the Lie algebra of $\sfG$. 
        \begin{equation}\label{eq:Lie_algebra_cocycle}
            \begin{aligned}
                \omega\,:\,L_0\frg\times L_0\frg\ &\rightarrow\ \fru(1)~,
                \\
                (\beta_1,\beta_2)\ &\mapsto\ \frac{\rmi}{2\pi}\int_0^1 \rmd r\,\innerLarge{\beta_1(r)}{\parder[\beta_2(r)]{r}}.
            \end{aligned}
        \end{equation}
        Here, $\inner{-}{-}$ is the inner product on $\frg$ which is normalised in a standard fashion~\cite{Pressley:1988qk}.\footnote{Explicitly, $\inner{h_\alpha}{h_\alpha}=2$ when $h_\alpha$ is the co-root corresponding to the highest root. In particular, we have $\inner{U}{V}=-\tr(UV)$ for $\fru(n)$ with anti-Hermitian generators and $\inner{U}{V}=-\tfrac12\tr(UV)$ for $\frso(2n)$ with antisymmetric generators.} Note that the cocycle~\eqref{eq:Lie_algebra_cocycle} also gives rise to the standard expression for the left-invariant 2-form curvature $\omega\in\Omega^2(L_0\sfG,\fru(1))$ of the bundle $\widehat{L_0\sfG}\rightarrow L_0\sfG$~\cite{Pressley:1988qk},
        \begin{equation}
            \omega_g(g\beta_1,g\beta_2)\ \coloneqq\ \omega(\beta_1,\beta_2)
        \end{equation}
        for all $g\in L_0\sfG$ and for all $g\beta_{1,2}\in T_gL_0\sfG$, or 
        \begin{equation}\label{eq:2FormCurvatureKacMoody}
            \omega_g \ =\ \frac{\rmi}{4\pi}\int_0^1\rmd r\,\innerLarge{\theta_{g(r)}}{\parder[\theta_{g(r)}]{r}},
        \end{equation}
        where $\theta$ is the left-invariant Maurer--Cartan form on $L_0\sfG$. This is a closed 2-form. The integrated cocycle is then obtained from the holonomy of $\omega$. Explicitly, given paths $f_{1,2}\in P_0L_0\sfG$, we construct the loop given by the triangle
        \begin{equation}\label{eq:triangle}
            \ell(f_1,f_2)\ \coloneqq\ \overline{f_1f_2}\circ f_1(1)f_2\circ f_1\ \in\ L_0L_0\sfG~,
        \end{equation}
        where $\overline{f}$ denotes the path $f\in P_0L_0\sfG$ with reversed orientation, that is, $\overline{f}(t,r)\coloneqq f(1-t,r)$. Then, we have the following result.
        
        \begin{proposition} (\cite{Murray:1987ua}) 
            Let $f_{1,2}\in P_0L_0\sfG$ and $D_{\ell(f_1,f_2)}$ be an arbitrary disc in $L_0\sfG$ with boundary $\ell(f_1,f_2)$ given by~\eqref{eq:triangle}. Then,
            \begin{equation}\label{eq:cocyclePLGU}
                \begin{aligned}
                    c(f_1,f_2)\ &\coloneqq\ \sfhol(\ell(f_1,f_2))
                    \\
                    &\coloneqq\ \exp\left(-\int_{D_{\ell(f_1,f_2)}}\omega\right)
                    \\
                    &\,=\ \exp\left(-\frac{\rmi}{2\pi}\int_0^1\rmd r\int_0^1\rmd s\int_0^s \rmd t\,\innerLarge{f_1^{-1}\parder[f_1]{s}}{f_2\left\{\parder{r}\left(f_2^{-1}\parder[f_2]{t}\right)\right\}f_2^{-1}}\right)
                \end{aligned}
            \end{equation}
            for all $f_{1,2}\in P_0L_0\sfG$ which we parametrised as $f_{1,2}=f_{1,2}(s,r)$ with $s,t$ the path parameters and $r$ the loop parameter. 
        \end{proposition}
        
        \noindent
        Note that using the identity
        \begin{equation}\label{eq:helpful_identity}
            \parder{t}\left(\parder[f_2(t,r)]{r}f_2^{-1}(t,r)\right)\ =\ f_2(t,r)\left\{\parder{r}\left(f_2^{-1}(t,r)\parder[f_2(t,r)]{t}\right)\right\}f_2^{-1}(t,r)~,
        \end{equation}
        we can simplify the expression~\eqref{eq:cocyclePLGU} to~\cite{Baez:2005sn}
        \begin{equation}\label{eq:cocyclePLGUBaez}
            c(f_1,f_2)\ =\ \exp\left(-\frac{\rmi}{2\pi}\int_0^1\rmd r\int_0^1\rmd s\,\innerLarge{f_1^{-1}(s,r)\parder[f_1(s,r)]{s}}{\parder[f_2(s,r)]{r}f_2^{-1}(s,r)}\right),
        \end{equation}
        where, again, $s$ and $r$ are the path and loop parameters, respectively. The group cocycle condition~\eqref{eq:group_cocycle} for $c$ is then straightforwardly verified using the form~\eqref{eq:cocyclePLGUBaez}.
        
        \begin{remark}
            We note that the left-invariance of the curvature $\omega$ implies that
            \begin{equation}\label{eq:holonomy_shift}
                \sfhol(\ell(f_1,f_2))\ =\ \sfhol(g\ell(f_1,f_2))
            \end{equation}
            for all $g\in L_0\sfG$. The cocycle condition~\eqref{eq:group_cocycle} then amounts to the integral of the curvature $\omega$ over a tetrahedron vanishing, and the sides of this tetrahedron are given by the triangles
            \begin{equation}
                \begin{aligned}
                    &(\overline{f_1f_2},f_1(1)f_2,f_1)~,~
                    &&(\overline{f_1f_2f_3},f_1(1)f_2f_3,f_1)~,
                    \\
                    &(\overline{f_1f_2f_3},(f_1f_2)(1)f_3,f_1f_2)~,~
                    &&(\overline{f_1(1)f_2f_3},(f_1f_2)(1)f_3,f_1(1)f_2)~,
                \end{aligned}
            \end{equation}
            cf.~\cite{Murray:1987ua}.
        \end{remark}
        
        Altogether, by means of~\eqref{eq:productOnP0L0G} and~\eqref{eq:cocyclePLGUBaez}, we have constructed a group product on $P_0L_0\sfG\times\sfU(1)$ and thus on $\big(P_0L_0\sfG\times\sfU(1)\big)/\sfN$. Furthermore, we obtain a group homomorphism
        \begin{subequations}\label{eq:tMomorphismString2Group}
            \begin{equation}
                \sft\,:\,P_0L_0\sfG\times\sfU(1)\big)/\sfN\ \rightarrow\ P_0\sfG
            \end{equation}
            via the composition of the evident projection $\big(P_0L_0\sfG\times\sfU(1)\big)/\sfN\rightarrow P_0L_0\sfG/L_0L_0\sfG\rightarrow L_0\sfG$ with the embedding $L_0\sfG\hookrightarrow P_0\sfG$. Explicitly, for all $[(f,z)]\in\big(P_0L_0\sfG\times\sfU(1)\big)/\sfN$ we have
            \begin{equation}
                \sft\big([(f,z)]\big)\ =\ \boundary(f)\ \in\ L_0\sfG~,
            \end{equation}
        \end{subequations}
        which is indeed well-defined. 
        
        \paragraph{Different group cocycles.} Let us briefly leave the main thread of our presentation in order to answer a question raised in~\cite{Murray:1987ua}.
        
        The starting point of our discussion around~\eqref{eq:cocyclePLGU} is the group cocycle\footnote{Note that in the following, we fix some sign errors appearing in the original literature.} 
        \begin{equation}
            c(f_1,f_2)\ =\ \exp\left(-\frac{\rmi}{2\pi}\int_0^1\rmd r\int_0^1\rmd s\int_0^s \rmd t\,\innerLarge{f_1^{-1}\parder[f_1]{s}}{f_2\left\{\parder{r}\left(f_2^{-1}\parder[f_2]{t}\right)\right\}f_2^{-1}}\right)
        \end{equation}
        given by Murray~\cite{Murray:1987ua} for all $f_{1,2}\in P_0L_0\sfG$ parametrised as $f_1=f_1(s,r)$ and $f_2=f_2(t,r)$ with $s,t$ the path parameters and $r$ the loop parameter. Furthermore, with the help of the identity~\eqref{eq:helpful_identity}, Murray's group cocycle becomes 
        \begin{equation}\label{eq:MurrayBaez}
            c(f_1,f_2)\ =\ \exp\left(-\frac{\rmi}{2\pi}\int_0^1\rmd r\int_0^1\rmd s\,\innerLarge{f_1^{-1}(s,r)\parder[f_1(s,r)]{s}}{\parder[f_2(s,r)]{r}f_2^{-1}(s,r)}\right),
        \end{equation}
        and this is the form of the group cocycle given in Baez--Stevenson--Crans--Schreiber~\cite{Baez:2005sn}. 
        
        Next, Mickelsson provides another group cocycle given by~\cite{Mickelsson:1987:173-183}
        \begin{equation}
            \begin{aligned}
                \tilde c(f_1,f_2)\ &\coloneqq\ \exp\left(-\frac{\rmi}{4\pi}\int_0^1\rmd r\int_0^1\rmd s\left\{\innerLarge{f_1^{-1}(s,r)\parder[f_1(s,r)]{s}}{\parder[f_2(s,r)]{r}f_2^{-1}(s,r)}\right.\right.
                \\
                &\kern5cm\left.\left.-\innerLarge{f_1^{-1}(s,r)\parder[f_1(s,r)]{r}}{\parder[f_2(s,r)]{s}f_2^{-1}(s,r)}\right\}\right).
            \end{aligned}
        \end{equation}
        The group cocycles $c$ and $\tilde c$ are related by a coboundary transformation
        \begin{equation}\label{eq:coboundaryDifferentGroupCocycles}
            \tilde c(f_1,f_2)\ =\ c(f_1,f_2)(d(f_1f_2))^{-1}d(f_1)d(f_2)~,
        \end{equation}
        cf.~\eqref{eq:coboundaryCondition}, with 
        \begin{equation}
            \begin{aligned}
                d\,:\,P_0L_0\sfG\ &\rightarrow\ \sfU(1)~,
                \\
                f\ &\mapsto\ \exp\left(-\frac{\rmi}{4\pi}\int_{0}^{1}\rmd r\int_0^1\rmd s\,\innerLarge{f^{-1}(s,r)\parder[f(s,r)]{s}}{f^{-1}(s,r)\parder[f(s,r)]{r}}\right).
            \end{aligned}
        \end{equation}
        This answers the question raised in~\cite{Murray:1987ua} as to what the relation between $c$ and $\tilde c$ is.
        
        \paragraph{Lift of the action in $\caL\sfG$.} 
        Returning to the main discussion, we finally have to lift the action of $P_0\sfG$ on $L_0\sfG$ in $\caL\sfG$ to an automorphism action of $P_0\sfG$ on $P_0L_0\sfG\times\sfU(1)\big)/\sfN$. Consider the left-invariant 1-form $\xi_g\in\Omega^1(L_0\sfG,\fru(1))$ given by (cf.~\cite{Baez:2005sn} and~\cite[Section 4]{Pressley:1988qk})
        \begin{equation}\label{eq:Baez1Form}
            (\xi_g)_\ell\ \coloneqq\ \frac{\rmi}{2\pi}\int_0^{1}\rmd r\,\innerLarge{\theta_{\ell(r)}}{g(r)^{-1}\parder[g(r)]{r}}
        \end{equation}
        for all $g\in P_0\sfG$ and for all $\ell\in L_0\sfG$. We have $(\Ad_g)^*\omega=\omega+\rmd \xi_g$, see \cref{app:proofs}, which implies that for loops $h\in L_0L_0\sfG$,
        \begin{equation}\label{eq:simp_loops}
            \exp\left(-\int_0^1\rmd s\,\xi_g\left(h^{-1}(s)\parder[h(s)]{s}\right)\right)\ =\ \sfhol(ghg^{-1})\sfhol^{-1}(h)~.
        \end{equation}
        Here, $h(s)\in L_0\sfG$ for all $s\in[0,1]$. 
        
        \begin{proposition}
            The group action~\cite{Baez:2005sn}
            \begin{equation}\label{eq:actionOfP0G}
                \begin{aligned}
                    g\acton(f,z)\ &\coloneqq\ \left(gfg^{-1},z\exp\left(\int_0^1\rmd s\,\xi_g\left(f^{-1}(s)\parder[f(s)]{s}\right)\right)\right)
                    \\
                    &\,=\ \left(gfg^{-1},z\exp\left(\frac{\rmi}{2\pi}\int_0^1\rmd r\int_0^1\rmd s\,\innerLarge{f^{-1}(s,r)\parder[f(s,r)]{s}}{g^{-1}(r)\parder[g(r)]{r}}\right)\right)
                \end{aligned}
            \end{equation}
            for all $g\in P_0\sfG$ and for all $(f,z)\in P_0L_0\sfG\times\sfU(1)$ is compatible with the group product on $P_0L_0\sfG\times\sfU(1)$. It also closes on the subgroup $\sfN$ defined in~\eqref{eq:subgroupN} and so, it descents to an action of $P_0\sfG$ on $\big(P_0L_0\sfG\times\sfU(1)\big)/\sfN$. The group homomorphism~\eqref{eq:tMomorphismString2Group} and the automorphism action~\eqref{eq:actionOfP0G} satisfy the crossed module conditions~\eqref{eq:crossedModuleConditions}.
        \end{proposition}
        
        \begin{proof}
            Compatibility with the group product follows from a direct calculation using the explicit expressions~\eqref{eq:cocyclePLGUBaez} and~\eqref{eq:Baez1Form} of the group cocycle $c$ and of $\xi_g$, respectively. Furthermore, with~\eqref{eq:simp_loops}, it is also clear that the action closes on $\sfN$. 
            
            We also note that the action trivially satisfies $\sft(g\acton\hat h)=g\sft(\hat h)g^{-1}$ for all $g\in P_0\sfG$ and for all $\hat h\in P_0L_0\sfG\times\sfU(1)$. It only remains to show that $(\sft(\hat h_1)\acton\hat h_2)\hat h_1\hat h_2^{-1}\hat h_1^{-1}\in\sfN$ for all $\hat h_{1,2}\in P_0L_0\sfG\times\sfU(1)$ as this implies the Peiffer identity. The proof can be found in \cref{app:proofs}.
        \end{proof}
        
        \paragraph{String 2-group model.} In conclusion, our above discussion now yields the following definition.
        
        \begin{definition} (\cite{Baez:2005sn})
            The \uline{string 2-group model} $\sfString(\sfG)$ is the crossed module of Lie groups
            \begin{equation}\label{eq:string2GroupDefinition}
                \sfString(\sfG)\ \coloneqq\ \big(\underbrace{\big(P_0L_0\sfG\times\sfU(1)\big)/\sfN}_{\cong\,\widehat{L_0\sfG}}\overset{\sft}{\longrightarrow}P_0\sfG,\acton\big)~.
            \end{equation}
            We shall also write $\sfString(n)$ as a shorthand notation for $\sfString(\sfSpin(n))$.
        \end{definition}
        
        \paragraph{String Lie 2-algebra model.}
        In order to derive the corresponding crossed module of Lie algebras, we identify the Lie algebra of $P_0L_0\sfG\times\sfU(1)$ with $P_0L_0\frg\oplus\fru(1)$. The Lie bracket is obtained in the usual way from the group commutator using~\eqref{eq:cocyclePLGUBaez},
        \begin{equation}\label{eq:P0L0G+R_commutator}
            \begin{aligned}
                [(\gamma_1,q_1),(\gamma_2,q_2)]\ &=\ \bigg([\gamma_1,\gamma_2],-\frac{\rmi}{2\pi}\int_0^1\rmd r\int_0^1\rmd s\,\bigg\{\innerLarge{\parder[\gamma_1(s,r)]{s}}{\parder[\gamma_2(s,r)]{r}}
                \\
                &\kern6cm -\innerLarge{\parder[\gamma_2(s,r)]{s}}{\parder[\gamma_1(s,r)]{r}}\bigg\}\bigg)
                \\
                &=\ \left([\gamma_1,\gamma_2],-\frac{\rmi}{2\pi}\int_0^1\rmd    r\,\innerLarge{\boundary(\gamma_1(s,r))}{\parder[\,\boundary(\gamma_2(s,r))]{r}}\right),
            \end{aligned}
        \end{equation}
        where we have integrated by parts in the second line. Note that $\boundary(\gamma_{1,2}(s,r))=\gamma_{1,2}(1,r)$. The extra minus sign in the above formula looks a bit odd, comparing it with the expression for the Lie algebra cocycle~\eqref{eq:Lie_algebra_cocycle}, but it is expected on general grounds. This follows as the difference between the horizontal lift of a commutator of vector fields and the commutator of horizontal lifts of vector fields is given by the curvature.\footnote{More explicitly, given a $\frg$-valued connection 1-form $\mu$ on a $\sfG$-bundle whose curvature is $\omega$, we have for any vector fields $V_{1,2}$ on the base space, the formula $\omega(V_1^h,V_2^h)=-\mu([V_1^h,V_2^h])$, where ${}^h$ denotes the horizontal lift. Upon identifying the Lie algebra $\frg$ with the space of vertical vector fields it generates, we can write this as $\omega(V_1^h,V_2^h)=[V_1,V_2]^h-[V_1^h,V_2^h]$, where we have used that $[V_1^h,V_2^h]_H=[V_1,V_2]^h$, where ${}_H$ denotes the horizontal part of a vector (field). For the case of a central extension $\hat\sfK$ of the group $\sfK$ by $\sfG=\sfU(1)$, this formula at the identity gives the Lie algebra commutator $[(V_1,0),(V_2,0)]=([V_1,V_2],0)-(0,\omega(V_1^h,V_2^h))=([V_1,V_2],-\omega(V_1,V_2))$, where we have used the connection to split the Lie algebra $\hat\frk=\frk\oplus\fru(1)$ and construct the horizontal lifts. We emphasise this point here, as it had been the source of sign inconsistencies plaguing our calculations.}
        
        Furthermore, because of
        \begin{equation}\label{eq:derhol}
            \dder{t}\bigg|_0\sfhol(\exp(t\gamma))\ =\ 0
        \end{equation}
        for all $\gamma\in L_0L_0\frg$, which we verify in \cref{app:proofs}, the Lie algebra $\frn$ of $\sfN$ is just
        \begin{equation}
            \frn\ =\ \{(\gamma,0)\in L_0L_0\frg\oplus\fru(1)\}\ \cong\ L_0L_0\frg~.
        \end{equation}
        From~\eqref{eq:P0L0G+R_commutator} we immediately see that $ \frn $ is indeed an ideal. Consequently,
        \begin{equation}
            (P_0L_0\frg\oplus\fru(1))/\frn\ \cong\ L_0\frg\oplus\fru(1)~.
        \end{equation}
        In particular, since the second component in the Lie algebra commutator depends only on the endpoint, we recover the Kac--Moody 2-cocycle~\eqref{eq:Lie_algebra_cocycle} in the factor algebra. In addition, the linearisation of the action~\eqref{eq:actionOfP0G} is
        \begin{equation}
            \alpha\,\acton\,(\gamma,q)\ =\ \left([\alpha,\gamma],\frac{\rmi}{2\pi}\int_0^1\rmd r\,\innerLarge{\parder[\alpha]{r}}{\gamma}\right)
        \end{equation}
        for all $\alpha\in P_0\frg$ and for all $(\gamma,q)\in L_0\frg\oplus\fru(1)$. In summary, we thus have obtained the crossed module of Lie algebras~\cite{Baez:2005sn}
        \begin{equation}\label{eq:strict_string_Lie_2_algebra}
            \frstring(\frg)\ \coloneqq\ \big(L_0\frg\oplus\fru(1)\overset{\sft}{\longrightarrow}P_0\frg,\acton\big)
        \end{equation}
        as a strict model of the string Lie 2-algebra. Here, $\sft$ is just the projection onto $L_0 \frg$ followed by the embedding into $P_0\frg$. 
        
        \paragraph{Minimal model.}
        We recall that a minimal model of the string Lie 2-algebra is given by the 2-term $L_\infty$-algebra 
        \begin{equation}
            \frstring^\circ(\frg)\ \coloneqq\ \fru(1)\oplus\frg~,
        \end{equation}
        where $\frstring^\circ(\frg)_{-1}\coloneqq\fru(1)$ and $\frstring^\circ(\frg)_0\coloneqq\frg$ with the non-trivial higher products
        \begin{equation}
            \begin{aligned}
                \mu_2\,:\,\frg\times\frg\ &\rightarrow\ \frg~~,
                \\
                (V_1,V_2)\ &\mapsto\ [V_1,V_2]~,
                \\
                \mu_3\,:\,\frg\times\frg\times\frg\ &\rightarrow\ \fru(1)~,
                \\
                (V_1,V_2,V_3)\ &\mapsto\ \rmi\inner{V_1}{[V_2,V_3]}
            \end{aligned}
        \end{equation}
        for all $V_{1,2,3}\in\frg$. We can now extend the quasi-isomorphism~\eqref{eq:quasi-iso_Lie_algebras} between the Lie algebra $\frg$, trivially regarded as a Lie 2-algebra, and the Lie 2-algebra $\caL\frg$ to a quasi-isomorphism 
        \begin{subequations}
            \begin{equation}
                \frstring^\circ(\frg)\ \xrightarrow{~\phi~}\ \frstring(\frg)\ \xrightarrow{~\psi~}\ \frstring^\circ(\frg)~.
            \end{equation}
            Explicitly, we have the chain maps
            \begin{equation}
                \begin{gathered}
                    \begin{tikzcd}
                        \fru(1)\arrow[d,"0"]\arrow[r,"\phi_1"] & L_0\frg\oplus\fru(1)\arrow[d,"\sft"]\arrow[r,"\psi_1"] & \fru(1)\arrow[d,"0"]
                        \\
                        \frg \arrow[r,"\phi_1\coloneqq\cdot_\wp"] & P_0\frg\arrow[r,"\psi_1\coloneqq\boundary"]  & \frg
                    \end{tikzcd}
                \end{gathered}
            \end{equation}
            where in the top row $\phi_1$ and $\psi_1$ are the evident embedding and projection maps. In this case, both $\phi_2$ and $\psi_2$ are non-trivial,
            \begin{equation}
                \begin{aligned}
                    \phi_2\,:\,\frg\times\frg\ &\rightarrow\ L_0\frg~,
                    \\
                    (V_1,V_2)\ &\mapsto\ (\wp-\wp^2)\cdot[V_1,V_2]
                \end{aligned}
            \end{equation}
            for all $V_{1,2}\in\frg$ and
            \begin{equation}
                \begin{aligned}
                    \psi_2\,:\,P_0\frg\times P_0\frg\ &\rightarrow\ \fru(1)~,
                    \\
                    (V_1,V_2)\ &\mapsto\ -\rmi\int_0^1\rmd r\left(\innerLarge{\parder[V_1]{r}}{V_2}-\innerLarge{V_1}{\parder[V_2]{r}}\right)
                \end{aligned}
            \end{equation}
            for all $V_{1,2}\in P_0\frg$.
        \end{subequations}
        
        \begin{remark}
            We notice that the string 2-group model crucially depends on the correct normalisation of the defining cocycle. Usually, $\sfString(\sfG)$ is therefore defined only for simply-connected compact simple Lie groups $\sfG$. The normalisation of the cocycle translates to the correct Lie 2-group extension of $\caL\sfG$ by $\sfB\sfU(1)$, the Lie 2-group corresponding to the crossed module $(\sfU(1)\rightarrow *,\id)$. These extensions are characterised by $H^3(\sfG,\IZ)$, and we have $H^3(\sfG,\IZ)\cong \IZ$ for simply-connected compact simple Lie groups with preferred element the generator $+1$. This limits us to $\sfSpin(n)$ for $n=3$ and $n\geq 5$ for $\sfString(n)$. As we shall see, however, there are also two equivalent, preferred choices for $n=4$, where $H^3(\sfSpin(4),\IZ)\cong \IZ\times \IZ$, which are the generators $(\pm 1,\mp 1)$ of opposite signs, so that we can speak of $\sfString(4)$.
        \end{remark}
        
        \subsection{Adjusted differential cocycles}\label{sec:adj_cocyclesStringGroup}
        
        Let us now discuss adjustment for the differential cocycles based on $\sfString(\sfG)$ for a Lie group $\sfG$. We note that the local kinematical data for higher gauge theory with gauge Lie 2-algebra $\frstring(\frg)$ was already given in~\cite{Saemann:2017rjm,Saemann:2019dsl}.
        
        \paragraph{Adjustment of $\sfString(\sfG)$.}
        Recall the adjustment $\kappa$ defined in~\eqref{eq:adjustmentPathLoopGroups}. We can now lift this adjustment datum from $\caL\sfG$ to $\sfString(\sfG)$ as follows.
        
        \begin{proposition}
            The map 
            \begin{equation}\label{eq:adjustmentStringGroup}
                \begin{aligned}
                    \hat\kappa\,:\,P_0\sfG\times P_0\frg\ &\rightarrow\ L_0\frg\oplus\fru(1)~,
                    \\
                    (g,V)\ &\mapsto\ \left(\kappa(g,V),\frac{\rmi}{2\pi}\int_0^1\rmd r\,\innerLarge{g^{-1}\parder[g]{r}}{V}\right)
                \end{aligned}
            \end{equation}
            defines an adjustment on $\sfString(\sfG)$. It also induces an adjustment at the level of crossed modules of Lie algebras with 
            \begin{equation}\label{eq:adjustmentStringAlgebra}
                \begin{aligned}
                    \hat\kappa\,:\,P_0\frg\times P_0\frg\ &\rightarrow L_0\frg\oplus\fru(1)~,
                    \\
                    (V_1,V_2)\ &\mapsto\ \left(\kappa(V_1,V_2),\frac{\rmi}{2\pi}\int_0^1\rmd r\,\innerLarge{\parder[V_1]{r}}{V_2}\right)
                \end{aligned}
            \end{equation}
            for all $V_{1,2}\in P_0\frg$ where $\kappa(V_1,V_2)$ is the adjustment datum defined in~\eqref{eq:adjustmentPathLoopAlgebras}. This infinitesimal adjustment is the one derived in~\cite{Saemann:2017rjm,Saemann:2019dsl}. 
        \end{proposition}
        
        \begin{proof}
            This follows from using the results from \cref{app:kappaProperties,app:proofs}. In particular,~\eqref{eq:alternativeAdjustmentCondition_a} directly follows from~\eqref{eq:hAhinv}. Moreover, using~\eqref{eq:StringkappaIdentities} and~\eqref{eq:StringkappaIdentity_b}, we ultimately obtain~\eqref{eq:alternativeAdjustmentCondition_b}. The infinitesimal adjustment is then derived by a short computation.
        \end{proof}
        
        Let us now briefly list the adjusted differential cocycles from \cref{def:adj_cocycles}, specialised to structure 2-group $\sfString(\sfG)$ with adjustment~\eqref{eq:adjustmentStringGroup}, and simplified. Such differential cocycles consists of maps
        \begin{subequations}
            \begin{equation}
                \begin{aligned}
                    \hat h=[(h,h')]\ &\in\ \scC^\infty\Big(Y^{[3]},\big(P_0L_0\sfG\times\sfU(1)\big)/\sfN\Big)~,
                    \\
                    (g,\hat\Lambda=(\Lambda,\Lambda'))\ &\in\ \scC^\infty(Y^{[2]},P_0\sfG)~\oplus~\Omega^1(Y^{[2]},L_0\frg\oplus\fru(1))~,
                    \\
                    (A,\hat B=(B,B'))\ &\in\ \Omega^1(Y^{[1]},P_0\frg)~\oplus~\Omega^2(Y^{[1]},L_0\frg\oplus\fru(1))~,
                \end{aligned}
            \end{equation}
            such that for all appropriate $(i,j,\ldots)\in Y^{[n]}$, with $Y^{[n]}$ given in~\eqref{eq:fibreProducts}, we have
            \begin{equation}
                \begin{gathered}
                    g_{ik}\ =\ \sft(h_{ijk})g_{ij}g_{jk}~,
                    \\
                    (h_{ikl}h_{ijk})^{-1}h_{ijl}(g_{ij}\acton h_{jkl})\ \in\ L_0L_0\sfG~,
                    \\
                    h'_{ikl}h'_{ijk}\ =\ h'_{ijl}h'_{jkl} \exp\left(\frac{\rmi}{2\pi}\int_0^1\rmd s\int_0^1\rmd r\,\innerLarge{g_{ij}^{-1}\parder[g_{ij}]{r}}{h_{jkl}^{-1}\parder[h_{jkl}]{s}}\right)\exp\left(-\int_\Box\omega\right),
                \end{gathered}
            \end{equation}
            where $\Box$ denotes the loop $\overline{h_{ikl}}\circ\big(\boundary(h_{ikl})\overline{h_{ijk}}\big)\circ\big(\boundary(h_{ijl})(g_{ij}\acton h_{jkl})\big)\circ h_{ijl}$. The cocycle conditions for the differential refinement specialise to 
            \begin{equation}
                \begin{aligned}
                    \Lambda_{ik}\ &=\ \Lambda_{jk}+g_{jk}^{-1}\acton\Lambda_{ij}-g_{ik}^{-1}\acton\big(\boundary(h_{ijk})\nabla_i\boundary(h_{ijk}^{-1})\big)\,,
                    \\
                    \Lambda'_{ik}\ &=\ \Lambda'_{jk}+\Lambda'_{ij}-h'_{ijk}\rmd h'^{-1}_{ijk}
                    \\
                    &\kern1cm-\frac{\rmi}{2\pi}\int_0^1\rmd r\,\left\{\innerLarge{\parder[g_{jk}]{r}g_{jk}^{-1}}{\Lambda_{ij}}-\innerLarge{\parder[g_{ik}]{r}g_{ik}^{-1}}{\boundary(h_{ijk})\nabla_i\boundary(h_{ijk}^{-1})}\right.
                    \\
                    &\kern2cm+\left.\innerLarge{\boundary(h_{ijk}^{-1})\parder[\,\boundary(h_{ijk})]{r}}{A_i}+\int_{0}^{1}\rmd s\,\innerLarge{\parder[h_{ijk}]{r}h_{ijk}^{-1}}{\parder[(h_{ijk}\rmd h_{ijk}^{-1})]{s}}\right\},\\
                    A_j\ &=\ g^{-1}_{ij}A_ig_{ij}+g^{-1}_{ij}\rmd g_{ij}- \Lambda_{ij}~,\\
                    B_j\ &=\ B_i+\rmd\Lambda_{ij}+A_j\acton \Lambda_{ij}+\tfrac12[\Lambda_{ij},\Lambda_{ij}]-\kappa(g^{-1}_{ij},F_i-B_i)~,\\
                    B'_j\ &=\ B'_i+\rmd\Lambda'_{ij} +\frac{\rmi}{2\pi}\int_0^1\rmd r\left\{\innerLarge{\parder[A_j]{r}}{\Lambda_{ij}}+\frac12\innerLarge{\parder[\Lambda_{ij}]{r}}{\Lambda_{ij}}+\innerLarge{\parder[g_{ij}]{r}g_{ij}^{-1}}{F_i-B_i}\right\}.
                \end{aligned}
            \end{equation}
        \end{subequations}
        Furthermore, using the linearisation~\eqref{eq:adjustmentStringAlgebra}, the adjusted curvatures~\eqref{eq:adjustedCurvatures} reduce to
        \begin{subequations}\label{eq:curvature_LG_compressed_extended}
            \begin{equation}
                \begin{aligned}
                    F_i\ &=\ \rmd A_i+\tfrac12[A_i,A_i]+\sft(\hat B_i)~,
                    \\
                    H_i\ &=\ \rmd B_i+A_i\acton B_i-\kappa(A_i,F_i)\ =\ (\id-\wp\cdot\boundary)(\rmd F_i)~,
                    \\
                    H'_i\ &=\ \rmd B'_i-\frac{\rmi}{2\pi}\int_0^1\rmd r\,\innerLarge{\parder[A_i]{r}}{F_i-\sft(B_i)}\ =\ \rmd\big(B'_i-\tfrac12\omega(A_i,A_i)\big)-\tfrac{\rmi}{4\pi}{\rm cs}(\boundary(A_i))
                \end{aligned}
            \end{equation}
            for all $(i)\in Y^{[1]}$. Here, in analogy to writing $\hat B=(B,B')$ we did write for the 3-form curvature $\hat H=(H,H')$, and we have made us of the Chern--Simons 3-form
            \begin{equation}\label{eq:chern-simons}
                {\rm cs}(A_i)\ \coloneqq\ \inner{A_i}{\rmd A_i}+\tfrac13\inner{A_i}{[A_i,A_i]}~
            \end{equation}
        \end{subequations}
        for all $(i)\in Y^{[1]}$.  The adjusted Bianchi identities~\eqref{eq:adjustedBianchiIdentities} thus become
        \begin{equation}\label{eq:Bianchi_identities_extended}
            \begin{aligned}
                \nabla_iF_i\ &=\ \sft(H_i+\kappa(A_i,F_i))~,
                \\
                \rmd H_i\ &=\ 0~,
                \\
                \rmd H'_i\ &=\ -\tfrac{\rmi}{4\pi}\inner{\boundary(F_i-\sft(B_i))}{\boundary(F_i-\sft(B_i))}\ =\ -\tfrac{\rmi}{4\pi}\inner{\boundary(F_i)}{\boundary(F_i)}
            \end{aligned}
        \end{equation}
        for all $(i)\in Y^{[1]}$. 
        
        \subsection{The string 2-bundle over \texorpdfstring{$S^4$}{S4}}\label{sec:instantonAntiInstantonStringLift}
        
        Having a lift of the spin group $\sfSpin(n)$ to a 2-group $\sfString(n)$, it is natural to ask, if there is a lift of
        \begin{equation}
            \sfSpin(5)\ \rightarrow\ \sfSpin(5)/\sfSpin(4)~~~\mbox{to}~~~\sfString(5)\ \rightarrow\ \sfString(5)/\sfString(4)~.
        \end{equation}
        As we will show, such a lift exists, and we first learnt about this construction, albeit without connection, from~\cite{Roberts:2022wwl}. It turns out that both quotient spaces can be identified with $S^4$, the latter, with the manifold $S^4$ trivially regarded as the discrete Lie groupoid $(S^4\multirightarrow{2}S^4)$, and we begin by commenting on this point.
        
        \paragraph{Quotients of crossed modules.}
        Recall that the categories of crossed modules of Lie groups and of strict Lie 2-groups are equivalent, cf.~\cite{Brown:1976:296-302,Baez:0307200,Porst:0812.1464}. Explicitly, given a crossed module of Lie groups $\caG=(\sfH\overset{\sft}{\longrightarrow}\sfG,\acton)$, we associated to $\caG$ a strict Lie 2-group, denoted by $\grp{\caG}$ and given by the data
        \begin{equation}\label{eq:crossedModuleGroupoid}
            \begin{gathered}
                \grp{\caG}\ \coloneqq\ (\sfG\ltimes\sfH\multirightarrow{2}\sfG)
                \ewith
                \kern-.8cm
                \begin{tikzcd}[column sep=2.0cm,row sep=large]
                    \phantom{\sft(h)}g & \sft(h^{-1})\,g\arrow[l,bend left,swap,out=-20,in=200]{}{(g,h)}
                \end{tikzcd}\!,
                \\
                (g_1,h_1)\circ(\sft(h^{-1}_1)g_1,h_2)\ =\ (g_1,h_1h_2)~,
                \\
                (g_1,h_1)\otimes (g_2,h_2)\ \coloneqq\ (g_1g_2,(g_1\acton h_2)h_1)
            \end{gathered}
        \end{equation}
        for all $g,g_{1,2}\in\sfG$ and for all $h,h_{1,2}\in\sfH$. Note that our conventions here are those of~\cite{Jurco:2014mva}, which are slightly different from~\cite{Baez:0307200}.
        
        By an \uline{embedding} of a crossed module $\caG_1\coloneqq(\sfH_1\overset{\sft_1}{\longrightarrow}\sfG_1,\acton_1)$ into a crossed module $\caG_2\coloneqq(\sfH_2\overset{\sft_2}{\longrightarrow}\sfG_2,\acton_2)$, we mean a (strict) injective morphism of crossed modules $\sfe:\caG_1\hookrightarrow\caG_2$ which is given by embeddings of groups $\sfe_0:\sfG_1\hookrightarrow\sfG_2$ and $\sfe_1:\sfH_1\hookrightarrow\sfH_2$ such that
        \begin{equation}\label{eq:consistency_2_group_embedding}
            \sfe_0(\sft_1(h_1))\ =\ \sft_2(\sfe_1(h_1))
            \eand 
            \sfe_0(g_1)\acton_2\sfe_1(h_1)\ =\ \sfe_1(g_1\acton_1 h_1)
        \end{equation}
        for all $h_1\in\sfH_1$ and for all $g_1\in\sfG$. Such an embedding then induces a right action of the morphisms of $\grp{\caG_1}$ on the morphisms of $\grp{\caG_2}$ by means of
        \begin{equation}
            \caR_{(g_1,h_1)}(g_2,h_2)\ \coloneqq\ (g_2,h_2)\otimes (\sfe_0(g_1),\sfe_1(h_1))\ =\ \big(g_2 \sfe_0(g_1),(g_2\acton_2\sfe_1(h_1))h_2\big)
        \end{equation}
        for all $(g_1,h_1)\in\sfG_1\ltimes\sfH_1$ and for all $(g_2,h_2)\in\sfG_2\ltimes\sfH_2$. Upon identifying morphisms under this action, we obtain the \uline{quotient Lie groupoid}\footnote{A more sophisticated definition in terms of quotient stacks is certainly possible, but not necessary here.}
        \begin{equation}\label{eq:quotionLieGroupoid}
            \caG_2/\caG_1\ \coloneqq\ \grp{\caG_2}/\grp{\caG_1}\ \coloneqq\ \Big(\big(\sfG_2\ltimes\sfH_2)/(\sfG_1\ltimes\sfH_1)\multirightarrow{2}\sfG_2/\sfG_1\Big)~,
        \end{equation}
        where, again, the equivalence relation is
        \begin{equation}
            g_2\ \sim\ g_2\sfe_0(g_1)
            \eand
            (g_2,h_2)\ \sim\ (g_2 \sfe_0(g_1),(g_2\acton_2\sfe_1(h_1))h_2)
        \end{equation}
        for all $(g_1,h_1)\in\sfG_1\ltimes\sfH_1$ and for all $(g_2,h_2)\in\sfG_2\ltimes\sfH_2$. The composition of morphisms in this groupoid is induced from the vertical composition $\circ$ in~\eqref{eq:crossedModuleGroupoid}.
        
        \begin{example}\label{ex:bu1_quotient}
            A useful and simple example of the above construction for the purposes of our discussion is given by the crossed modules $\sfB\sfU(1)=(\sfU(1)\longrightarrow*,\id)$ and $\sfString(\sfG)$ as defined in~\eqref{eq:string2GroupDefinition}. We can embed $\sfB\sfU(1)$ into $\sfString(\sfG)$ as follows:
            \begin{subequations}\label{eq:first_coset_example}
                \begin{equation}
                    \begin{tikzcd}
                        \sfU(1)\arrow[r,"\sfe_1"]\arrow[d] & \widehat{L_0\sfG}\arrow[d]
                        \\
                        *\arrow[r,"\sfe_0"] & P_0\sfG
                    \end{tikzcd}
                    \quad
                    \begin{array}{c}
                        {}
                        \\[-5pt]
                        \sfe_1\,:\,\sfU(1)\ \ni\ z\ \mapsto\ [(\unit,z)]\ \in\ \big(P_0L_0\sfG\times\sfU(1)\big)/\sfN\ \cong\ \widehat{L_0\sfG}~,
                        \\
                        \kern-4pt\sfe_0\,:\,*\ \ni\ \unit\ \mapsto\ \unit\ \in\ P_0\sfG~.    
                    \end{array}
                \end{equation}
                The resulting quotient is
                \begin{equation}\label{eq:cosetStringBU}
                    \sfString(\sfG)/\sfB\sfU(1)\ \cong\ \caL\sfG
                \end{equation}
            \end{subequations}        
            with $\caL\sfG$ the crossed module defined in~\eqref{eq:definitionLoopLieCrossed}. This is, in fact, the essence of the short exact sequence~\eqref{eq:centralExtension2Groups}. Since $\sfG$ can be trivially identified with the crossed module $(*\hookrightarrow\sfG,\id)$ and because of the latter's identification with $\caL\sfG$ by virtue of~\eqref{eq:loopGroupButterfly}, any Lie group $\sfG$ can thus be viewed as the coset $\sfString(\sfG)/\sfB\sfU(1)$.
        \end{example}
        
        We note that the string 2-group $\sfString(1)$ can be identified with the group $\sfB\sfU(1)$, so that \cref{ex:bu1_quotient} can be regarded as the quotient $\sfString(\sfG)/\sfString(1)$. We shall consider more general quotients of the string group later on.
        
        Also, we note that a special case of \cref{ex:bu1_quotient} is the categorification of the Hopf fibration $S^1\hookrightarrow S^3\rightarrow S^2$ to
        \begin{equation}\label{eq:categorified_Hopf}
            \sfString(1)\ \hookrightarrow\ \sfString(3)\ \rightarrow\ \caL\sfSpin(3) \ \cong\ \sfSpin(3)~.
        \end{equation}
        This fibration underlies the Abelian self-dual string, see e.g.~\cite{Saemann:2017rjm} for more details.
        
        \paragraph{Quotients of string 2-groups.} Consider two string groups $\sfString(\sfH)$ and $\sfString(\sfG)$ of the form~\eqref{eq:string2GroupDefinition} for some Lie groups $\sfH$ and $\sfG$ together with an embedding $\sfH\hookrightarrow\sfG$. Such an embedding induces an embedding of crossed modules $\sfString(\sfH)\hookrightarrow \sfString(\sfG)$, and we have resulting group homomorphisms
        \begin{equation}
            P_0\sfH\ \hookrightarrow\ P_0\sfG~,
            \quad
            P_0L_0\sfH\ \hookrightarrow\ P_0L_0\sfG~,
            \quad
            L_0L_0\sfH\ \hookrightarrow\ L_0L_0\sfG~.
        \end{equation}
        From the definition~\eqref{eq:subgroupN}, it is furthermore clear that we obtain an embedding of the normal subgroup $\sfN_\sfH$ into $\sfN_\sfG$. These embeddings are compatible with the structure maps, and we have
        \begin{equation}
            \begin{aligned}
                \sfString(\sfG)/\sfString(\sfH)\ &= \ \left(\frac{P_0\sfG\ltimes((P_0L_0\sfG\times\sfU(1))/\sfN_\sfG)}{P_0\sfH\ltimes((P_0L_0\sfH\times \sfU(1))/\sfN_\sfH)}\multirightarrow{2}P_0\sfG/P_0\sfH\right)
                \\
                \ &\cong \ \left(\frac{P_0\sfG\ltimes L_0\sfG}{P_0\sfH\ltimes L_0\sfH}\multirightarrow{2}P_0\sfG/P_0\sfH\right)
                \\
                \ &\cong \ \caL\sfG/\caL\sfH
                \\
                \ &\cong \ \sfG/\sfH~.
            \end{aligned}
        \end{equation}
        In the last equivalence, we have used the fact that our notion of quotient of Lie groups is compatible with the equivalences $\caL\sfG\cong\sfG_{\rm cm}$, see~\eqref{eq:crossedModuleLg} and~\eqref{eq:definitionLieGroupLieCrossed}, which immediately follows from the fact that underlying this equivalence are surjective morphisms $\caL\sfG\rightarrow \sfG_{\rm cm}$ and $\caL\sfH\rightarrow\sfH_{\rm cm}$ which map orbits to orbits. The above discussion has the following corollary.
        
        \begin{corollary}
            The quotient $\sfString(5)/\sfString(4)$ is a Lie groupoid, which is equivalent to $S^4$, regarded as the discrete Lie groupoid $(S^4\multirightarrow{2}S^4)$.
        \end{corollary}
        
        \paragraph{Total space description.} 
        Before presenting explicit cocycles, let us briefly sketch the total space description of the string bundle over $S^4$, which is 
        \begin{equation}\label{eq:nonAbelian_coset_bundle}
            \sfString(5)\ \rightarrow\ S^4\ \cong\ \sfString(5)/\sfString(4)~.
        \end{equation}
        Together with the groupoid equivalence
        \begin{equation}
            (S^4\multirightarrow{2}S^4)\ \cong\ \big(L_0S^4\multirightarrow{2}P_0S^4\big)\,,
        \end{equation}
        we then have the following diagram
        \begin{equation}
            \begin{tikzcd}
                \widehat{L_0\sfSpin(5)} \arrow[d]&
                \\
                L_0\sfSpin(5)\arrow[r,shift left]\arrow[r,shift right] \arrow[d] & P_0\sfSpin(5)\arrow[d]
                \\
                L_0S^4\arrow[r,shift left]\arrow[r,shift right] & P_0S^4\arrow[d,"\boundary"]
                \\
                & S^4
            \end{tikzcd}
        \end{equation}
        This diagram underlies the description of the principal $\sfString(4)$-bundle as a non-Abelian bundle gerbe, cf.~\cite{Aschieri:2003mw,Jurco:2005qj} as well as~\cite{Nikolaus:2011ag}, but we refrain from giving further details.
        
        \paragraph{Cocycle description.} 
        The above total space description is not suitable for the direct extraction of a cocycle description: the bundle $\widehat{L_0\sfSpin(5)}$ is a non-trivial higher principal bundle over $L_0S^4$. We therefore need to switch to a different cover. We also note that the different lifts form a torsor for $H^3(S^4,\IZ)$. Because this group is trivial, the lift is unique up to isomorphism. This allows us to work with the same simple cover $\bar U_1\sqcup\bar U_2\rightarrow S^4$ defined in~\eqref{eq:2patchCover}, which proved to be sufficient in the description of the coset space $\caL\sfSpin(5)/\caL\sfSpin(4)\cong\sfSpin(5)/\sfSpin(4)\cong S^4$ in \cref{sec:Spin4Bundle}. The cocycles for this bundle will then evidently be an extension of the cocycles introduced in~\eqref{eq:cocycles_LSpin_bundle}.
        
        Out second main result is then the following.
        
        \begin{theorem}
            Subordinate to the cover $\bar U_1\sqcup\bar U_2$, the string bundle with an adjusted connection describing a string structure of $S^4$ is given by the following cocycle data 
            \begin{subequations}
                \begin{equation}
                    \begin{aligned}
                        \hat h_{ijk}\,:\,\bar Y^{[3]}\ &\rightarrow\ \big(P_0L_0\sfSpin(4)\times\sfU(1)\big)/\sfN~,
                        \\
                        q_+\ &\mapsto\ \big[\big((s,t)\mapsto h^\circ_{ijk}(q_+,st),h'_{ijk}\big)\big]\,,
                    \end{aligned}
                \end{equation}
                as well as
                \begin{equation}
                    h'_{111}\ \coloneqq\ h'_{222}\ \coloneqq\ h'_{112}\ \coloneqq\ h'_{122}\ \coloneqq\ h'_{121}\ \coloneqq\ 1~,
                \end{equation}
                and
                \begin{equation}
                    \Lambda'_{11}\ \coloneqq\ \Lambda'_{22}\ \coloneqq\ \Lambda'_{12}\ \coloneqq\ 0 \eand B'_{1}\ \coloneqq\ 0~,
                \end{equation}
                together with the data in~\eqref{eq:cocycles_LSpin_bundle}. The cocycle relations~\eqref{eq:adjustedCocycleConditions} then reduce to 
                \begin{equation}
                    \begin{aligned}
                        \hat h_{212}\ &=\ g_{21}\acton \hat h_{121}~,
                        \\
                        \hat \Lambda_{21}\ &=\ \hat h_{121}\nabla_1\hat h_{121}^{-1}-g_{21}^{-1}\acton \hat \Lambda_{12}~,
                        \\
                        \hat B_2\ &=\ \hat B_1+\rmd\hat \Lambda_{12}+A_2\acton \hat \Lambda_{12}+\tfrac12[\hat \Lambda_{12},\hat \Lambda_{12}]-\hat \kappa(g^{-1}_{12},F_1-B_1)
                    \end{aligned}
                \end{equation}
                thus providing the missing data. The additional components arising in the lift to a string bundle read as
                \begin{equation}
                    \begin{aligned}
                        h_{212}\ &=\ h_{121}~,
                        \\
                        h'_{212}\ &=\ \exp\left(\frac{\rmi}{2\pi}\int_{0}^{1}\rmd t \int_{0}^{1}\rmd s\innerLarge{g_{21}^{\circ-1}(q,t)\parder[g^\circ_{21}(q,t)]{t}}{h^{\circ-1}_{121}(q,st)\parder[h^{\circ}_{121}(q,st)]{s}}\right) 
                        \\
                        &=\ \exp(\rmi\theta_q)\ =\ \frac{\Re(q)}{|q|}+\rmi\frac{|\Im(q)|}{|q|}~,
                        \\
                        \Lambda'_{21}\ &=\ -\frac{\rmi}{2\pi}\int_0^1\rmd r\left\{\innerLarge{\parder[g^\circ_{12}]{r}g_{12}^{\circ-1}}{A^\circ_1}+\innerLarge{\parder[g^\circ_{21}]{r}g_{21}^{\circ-1}}{A^\circ_2}\right.
                        \\
                        &\kern2cm-\left.\innerLarge{\parder[g^\circ_{21}]{r}g_{21}^{\circ-1}}{g_{12}^{\circ-1} \rmd g^\circ_{12}}-\int_{0}^{1}\rmd s\,\innerLarge{\parder[h^\circ_{121}]{r}h_{121}^{\circ-1}}{\parder[(h^\circ_{121}\rmd h_{121}^{\circ-1})]{s}}\right\},
                        \\
                        B'_{2}\ &=\ \frac{\rmi}{4\pi}\left.\left(\innerLarge{g_{12}^{\circ-1}\rmd g^\circ_{12} + g_{12}^{\circ-1} A^\circ_1 g^\circ_{12}}{A^\circ_2} +\innerLarge{g_{12}^{\circ-1} A^\circ_1 g^\circ_{12}}{g_{12}^{\circ-1}\rmd g^\circ_{12}}\right)\right|_{r=1}
                        \\
                        &\kern1cm+\frac{\rmi}{2\pi}\rmd\int_0^1\rmd r\innerLarge{\parder[g^\circ_{12}]{r}g_{12}^{\circ-1}}{A^\circ_1} 
                        \\
                        &\kern1cm+\frac{\rmi}{4\pi}\int_0^1\rmd r\left\{\innerLarge{\parder[A^\circ_1]{r}}{A^\circ_1}-\innerLarge{\parder[A^\circ_2]{r}}{A^\circ_2}+\innerLarge{\parder[g_{12}^{\circ-1}\rmd g^\circ_{12}]{r}}{g_{12}^{\circ-1}\rmd g^\circ_{12}}\right\}.
                    \end{aligned}
                \end{equation}
            \end{subequations}
        \end{theorem}
        \begin{proof}
            The proof is simply a verification of the adjusted cocycle relations for $\sfString(4)$-bundles, which is a tedious but straightforward exercise.
        \end{proof}
        
        \noindent Finally, we note that \cref{rem:better_cover} also applies here, and much more suitable covers may exist, that are, however, less suitable for immediate physical applications. A physical interpretation of the above string structure on $S^4$ is given in \ref{ssec:nonAbelianSelfDualString}.
        
        \section{Physical applications}
        
        We close with some remarks on the physical application of our constructions. 
        
        \subsection{Supergravities with gauge potentials}
        
        Adjusted connections appear widely in the context of heterotic and gauged supergravities. The latter involve principal $n$-bundles with $n=2,\ldots,6$, but the technology we developed above limits us to the case $n=2$. Let us therefore focus on heterotic supergravity; for a detailed account of the local situation in terms of adjusted connections, see~\cite{Borsten:2021ljb}.
        
        In this theory, the string theory Kalb--Ramond $B$-field needs to be coupled to a Lie-algebra valued one-form gauge potential. The local, infinitesimal description is already found in~\cite{Bergshoeff:1981um,Chapline:1982ww}, and the deformation of curvature form due to the adjustment and the deformation of 2-form gauge transformations were already determined in these papers, see e.g.~\cite[Eqs.~(5),(7)]{Chapline:1982ww}. 
        
        Our construction now allows now a) to describe explicit finite gauge transformations and b) to obtain a global description of this type of supergravity on arbitrary manifolds. The latter is particularly interesting for discussing T-duality, which requires compact space-time directions.
        
        Explicitly, the underlying gauge Lie 2-algebra is a weak form of the string Lie 2-algebra we consider in this paper. Using the morphism from this weak Lie 2-algebra to our strict Lie 2-algebra found e.g.~in~\cite{Baez:2005sn} or~\cite{Saemann:2019dsl}, one can now in principle formulate this theory globally. 
        
        We stress that in order to obtain the local cocycle description usually required in physics computations makes using a formalism such as our indispensable: the $B$-field here is not simply a connective structure on an abelian gerbe.
        
        \subsection{Non-Abelian self-dual strings}\label{ssec:nonAbelianSelfDualString}
        
        More concretely, and very close to our above discussion is the non-abelian version of the self-dual string equation as found in~\cite{Saemann:2017rjm} or~\cite{Akyol:2012cq}.
        
        \paragraph{Abelian self-dual strings.}
        Recall that a \uline{Dirac} or \uline{$\sfU(1)$-monopole} is a principal circle bundle with connection $A$ and curvature $F$ on $\IR^3\setminus\{0\}$ together with a Higgs field given by a function $\phi:\IR^3\setminus\{0\}\rightarrow\fru(1)$ with $F={\star\rmd\phi}$. The fundamental (i.e.~charge~1) Dirac monopole is topologically given by the Hopf fibration. In string theory, Dirac monopoles can be obtained from D1-branes ending on D3-branes. A categorified version of this is the \uline{Abelian self-dual string}, which describes M2-branes ending on an M5-brane~\cite{Howe:1997ue}. Mathematically, it is given by an Abelian gerbe with connection $B$ and curvature $H$ on $\IR^4\setminus\{0\}$ together with a Higgs field or function $\phi:\IR^4\setminus\{0\}\rightarrow\fru(1)$ such that $H={\star\rmd\phi}$. The fundamental self-dual string is topologically described by the categorified Hopf fibration~\eqref{eq:categorified_Hopf}.
        
        \paragraph{Non-Abelian generalisation.}
        The generalisation of the Dirac monopole to a non-Abelian principal bundle is straightforward and usually called the \uline{Bogomolny monopole}. This is a principal $\sfG$-bundle with connection over\footnote{In the non-Abelian case, non-trivial solutions exist for trivial principal bundles.} $\IR^3$ and a $\frg$-valued Higgs field $\phi$ satisfying $F={\star\nabla\phi}$. 
        
        As argued in~\cite{Saemann:2017rjm}, the appropriate non-Abelian generalisation of the self-dual string is a principal $\sfString(\sfG)$-bundle with adjusted connection together with a Higgs field $\phi:\IR^4\rightarrow \widehat{L_0\frg}$ such that 
        \begin{equation}\label{eq:non_Abelian_SDS}
            H\ =\ {\star\nabla\phi}~,
            \quad
            \pr_1(\phi)\ =\ 0~,
            \eand
            F\ =\ {\star F}~,
        \end{equation}
        where $F$ and $H$ are the 2- and 3-form components of the curvature, $\nabla=\rmd+A\acton$, and $\pr_1:\widehat{L_0\frg}\rightarrow L_0\frg$ is the evident projection. As shown in~\cite{Saemann:2017rjm}, these equations posses a natural dimensional reduction of the non-Abelian self-dual string equation to the Bogomolny monopole equation. Moreover, these configurations are BPS equations of an interesting six-dimensional superconformal field theory~\cite{Saemann:2017rjm}, see also~\cite{Akyol:2012cq}. This is expected to be related to the description of M2-branes ending on M5-branes.
        
        We emphasise that in~\cite{Saemann:2017rjm}, only infinitesimal gauge transformations and the local bundle structure were available. With the techniques developed in \cref{sec:stringBundlesWithConnection}, we now also have finite gauge transformations and the total bundle structure under full control. This is vital, e.g., if one aims for a higher form of the Nahm transform, which describes the duality between Nahm data and monopoles. 
        
        \subsection{Outline of the twistor space description of self-dual strings}\label{sec:commentsTwistorApproach}
        
        Twistor descriptions of higher gauge theories have been studied before in the unadjusted setting~\cite{Saemann:2011nb,Saemann:2012uq,Saemann:2013pca,Jurco:2014mva,Jurco:2016qwv,Samann:2017dah}, but fake-flatness was a fundamental obstruction to obtaining interesting solutions. One may now hope that adjusted connections solve also this problem, but the situation is more subtle as we explain in the following.
        
        \paragraph{Penrose--Ward transform.} 
        Many important field equations can be recast as the partial flatness of the curvature of a gauge potential when restricted to certain subspaces of space-time $X$. The moduli space of the relevant subspaces is then called \uline{twistor space} and commonly denoted by $Z$. Fibred over both is the \uline{correspondence space} $Y$, the (disjoint) union of relevant subspaces containing a particular space-time point over all space-time points:
        \begin{equation}\label{eq:twistorDoubleFibration}
            \begin{tikzcd}
                & Y\arrow[dl,"\pi_2",swap] \arrow[dr,"\pi_1"]
                \\
                Z & & X
            \end{tikzcd}
        \end{equation}
        The projection $\pi_1$ forgets the subspaces considered, whilst the projection $\pi_2$ forgets the contained space-time point. We thus have the \uline{geometric twistor correspondence},
        \begin{equation}
            Z\ \ni\ z\ \leftrightarrow\ \pi_1(\pi_2^{-1}(z))\ \subseteq\ X 
            \eand
            Z\ \supseteq\ \pi_2(\pi_1^{-1}(x))\ \leftrightarrow\ x\ \in\ X~.
        \end{equation}
        
        Under suitable topological conditions on the double fibration~\eqref{eq:twistorDoubleFibration},\footnote{see e.g.~\cite{Eastwood:1985aa}} the \uline{Penrose--Ward transform} then identifies the gauge orbits of solutions to some field equations with isomorphism classes of a certain principal bundle $P$ over the corresponding twistor space $Z$. Examples include self-dual Yang--Mills theory~\cite{Ward:1977ta,Hitchin:1982gh} and its supersymmetric extensions~\cite{Witten:2003nn,Popov:2004rb,Popov:2005uv} as well as Yang--Mills theory~\cite{Isenberg:1978kk,Witten:1978xx,Khenkin:1980ff,Pool:1981aa,Buchdahl:1985aa,Popov:2021mfl} and its supersymmetric extensions~\cite{Witten:1978xx,Saemann:2005ji,Saemann:2012rr}.\footnote{We note that such constructions also exist for self-dual gravity~\cite{Penrose:1976js,Atiyah:1978wi} and its supersymmetric extensions~\cite{Merkulov:1992,Wolf:2007tx,Mason:2007ct}, as well as Einstein gravity~\cite{0264-9381-2-4-020} and its supersymmetric extensions~\cite{Merkulov:1992qa}.} For instance, in the context of four-dimensional self-dual (supersymmetric) Yang--Mills theory, the relevant subspaces are self-dual (super) surfaces in (super) space-time whereas for four-dimensional $\caN=3$ supersymmetric Yang--Mills theory, the relevant subspaces are super null-geodesics in $\caN=3$ super space-time, respectively.
        
        Roughly speaking, we can pull back the bundle $P$ along $\pi_2$ and perform a trivialising isomorphism $\pi_2^*P\rightarrow\tilde P$. This will lead to a relative connection along $\pi_2$ on the bundle $\tilde P$ over $Y$, and this connection is relatively flat. It turns out that this connection is necessarily of a form that allows for a push-down along $\pi_1$, leading to a connection on a principal bundle $\pi_1^!\tilde P$ on space-time $X$. Relative flatness of the connection on $\tilde P$ now implies flatness of the connection along the subspaces $\pi_1(\pi_2^{-1}(z))\subseteq X$ for all $z\in Z$, which, in turn, implies the desired field equations on space-time.
        
        \paragraph{Generalisation to higher gauge theory.}
        Starting from the Abelian setting in~\cite{Saemann:2011nb,Mason:2011nw}, this construction was generalised to higher gauge theories in~\cite{Saemann:2011nb,Saemann:2012uq,Saemann:2013pca,Jurco:2014mva,Jurco:2016qwv,Samann:2017dah}. These generalisations even allow for replacing the manifolds appearing in the double fibration~\eqref{eq:twistorDoubleFibration} by categorified spaces~\cite{Jurco:2016qwv}. 
        
        However, what all these generalisations have in common is that the fake-flatness condition for the connection is built in. Unfortunately, adjusted connections do not alleviate this problem. There is a tension between the fibres of $\pi_2:Y\rightarrow Z$ being large enough to produce the full information of a 2-form $B$ on space-time, together with a constraint on $H$, and them being small enough to allow for a non-flat relative connection with a non-trivial relative 2-form curvature along $\pi_2:Y\rightarrow Z$. The latter, however, is required for our generalisation of adjusted string bundles to be noticeable. We show in the following how this issue can be circumvented for the non-Abelian self-dual string equation~\eqref{eq:non_Abelian_SDS}.
        
        \paragraph{Twistor space diagram.} 
        For convenience, we consider a flat space-time, which we complexify, and we also work with a complexified gauge group. We then have the following fibrations of twistor space:
        \begin{equation}
            \begin{tikzcd}
                & & \IC^4\times \IC P^1\times \IC P^1 \arrow[dl,"\pi_2",swap] \arrow[ddrr,"\pi_1"]
                \\[10pt]
                & Z_{\rm Ins}\times\IC P^1\arrow[dl,"\pi_{3}",swap] \arrow[dr,"\pi_{4}"] & & 
                \\
                Z_{\rm Hyp} & & Z_{\rm Ins} & & \IC^4
            \end{tikzcd}
        \end{equation}
        Here, $Z_{\rm Ins}$ denotes the Penrose twistor space for an instanton, which is the total space of the holomorphic vector bundle $\caO(1)\oplus\caO(1)\rightarrow\IC P^1$ and $Z_{\rm Hyp}$ is the hyperplane twistor space introduced in~\cite{Saemann:2011nb}, which is the total space of the holomorphic vector bundle $\caO(1,1)\rightarrow\IC P^1\times\IC P^1$. Since $\sfSpin(4)\cong\sfSU(2)\times\sfSU(2)$, we can use spinorial indices $\alpha,\beta,\ldots,\dot\alpha,\dot\beta,\ldots=1,2$. We take $x^{\alpha\dot\beta}$ as coordinates on $\IC^4$, $\lambda_{\dot\alpha}$ and $\mu_\alpha$ are homogeneous coordinates on the two projective spaces $\IC P^1$, and $y$ and $z^\alpha$ and $z$ are the fibre coordinates of $\caO(1,1)\rightarrow\IC P^1\times\IC P^1$ and $\caO(1)\oplus\caO(1)\rightarrow\IC P^1$, respectively. Then, the various projections are given as follows,
        \begin{equation}
            \begin{gathered}
                \pi_1\,:\,(x^{\alpha\dot\beta},\mu_\alpha,\lambda_{\dot\alpha})\ \mapsto\ (x^{\alpha\dot\beta})~,
                \\
                \pi_2\,:\,(x^{\alpha\dot\beta},\mu_\alpha,\lambda_{\dot\alpha})\ \mapsto\ (z^\alpha,\mu_\alpha,\lambda_{\dot\alpha})\ \coloneqq\ (x^{\alpha\dot\beta}\lambda_{\dot\beta},\mu_\alpha,\lambda_{\dot\alpha})~,
                \\
                \pi_3\,:\,(z^\alpha,\mu_\alpha,\lambda_{\dot\alpha})\ \mapsto\ (y,\mu_\alpha,\lambda_{\dot\alpha})\ \coloneqq\ (z^\alpha\mu_\alpha,\mu_\alpha,\lambda_{\dot\alpha})~,
                \\
                \pi_4\,:\,(z^\alpha,\mu_\alpha,\lambda_{\dot\alpha})\ \mapsto\ (z^\alpha,\lambda_{\dot\alpha})~.
            \end{gathered}
        \end{equation}
        For more details and the conventions we use, see e.g.~\cite{Popov:2004rb,Wolf:2010av} as well as~\cite{Saemann:2011nb}. 
        
        \paragraph{Iterative Penrose--Ward transform.}
        Consider a smoothly trivial holomorphic principal $\sfG$-bundle $P_0$ over $Z_{\rm Ins}$ which is holomorphically trivial when restricted to $\IC P^1\cong(\pi_4\circ\pi_2)(\pi_1^{-1}(x))\subseteq Z_{\rm Ins}$ for all $x\in\IC^4$ and which has vanishing second Chern class. As is well-known~\cite{Ward:1977ta}, this bundle describes a solution to the self-dual Yang--Mills equations on space-time with gauge group $\sfG$ via the Penrose--Ward transform along the span given by $\pi_4\circ\pi_2$ and $\pi_1$. We define the pullback 
        \begin{equation}
            P_1\ \coloneqq\ \pi^*_4 P_0~,
        \end{equation}
        and this bundle can be lifted to a holomorphic principal $\sfString(\sfG)$-bundle $\scP_1$ as the second Chern class of $P_1$ necessarily vanishes.
        
        Next, consider a holomorphic Abelian gerbe $\scG$ over $Z_{\rm Hyp}$. As shown in~\cite{Saemann:2011nb}, this gerbe by itself describes an Abelian self-dual string on $\IC^4$ by means of a Penrose--Ward transform along the span given by $\pi_4\circ\pi_3$ and $\pi_1$. We can now tensor the principal $\sfString(\sfG)$-bundle $\scP_1$ by the pull-back $\pi_3^*\scG$, which yields the holomorphic principal $\sfString(\sfG)$-bundle 
        \begin{equation}
            \scP_2\ \coloneqq\ \scP_1\otimes\pi^*_3\scG~.
        \end{equation}
        Whilst the lift of $P_1$ to $\scP_1$ is not unique, the different choices are parametrised precisely by the choices of $\scG$. Consequently, we have to provide an additional, non-canonical choice for defining the lift. Nevertheless, we are certain to cover the solution space with any such choice, because the possible self-dual strings form a torsor for the group of the pull-backs of the possible gerbes $\scG$. The resulting holomorphic principal string bundle $\scP_2$ can then be taken as input data for a straightforward higher Penrose--Ward transform as in~\cite{Saemann:2012uq}, and the result is a solution to the non-Abelian self-dual string equations~\eqref{eq:non_Abelian_SDS} for $\sfString(\sfG)$ over $\IC^4$. We note that a supersymmetric extension is also readily implemented, cf.~\cite{Saemann:2012uq}.
        
        \appendix
        \addappheadtotoc
        \appendixpage
        
        \section{Useful relations and detailed computations}
        
        In this appendix, we collect a number of very useful identities and detailed computations that would have distracted from the key arguments in the main part of this paper.
        
        \subsection{Computations with crossed modules}\label{app:useful}
        
        Let $(\sfH\overset{\sft}{\longrightarrow}\sfG,\acton)$ be a crossed module of Lie groups and $(\frh\overset{\sft}{\longrightarrow}\sfg,\acton)$ the corresponding crossed module of Lie algebras, respectively. 
        
        \paragraph{Algebraic relations.} 
        Firstly, recall that $g\acton\unit=\unit$ and $\unit\acton h=h$ for all $g\in\sfG$ and for all $h\in\sfH$. We have the following implications:
        \begin{subequations}
            \begin{equation}\label{eq:app_identity_t(hA1/h)}
                \begin{gathered}
                    \sft(h^{-1}(g\acton h))\ =\ \sft(h^{-1})g\sft(h)g^{-1}
                    \\
                    \Rightarrow\quad
                    \sft(h^{-1}(V\acton h))\ =\ \sft(h^{-1})V\sft(h)-V~,
                \end{gathered}
            \end{equation}
            \begin{equation}
                \begin{gathered}
                    g_1\acton(g_2\acton h)\ =\ g_2\acton((g_2^{-1}g_1g_2)\acton h)
                    \\
                    \Rightarrow\quad 
                    V\acton(g\acton h)\ =\ g\acton((g^{-1}Vg)\acton h)~,
                \end{gathered}
            \end{equation}
            \begin{equation}
                \begin{gathered}
                    (gVg^{-1})\acton W\ =\ g\acton(V\acton(g^{-1}\acton W))
                    \\
                    \Rightarrow\quad
                    [V_1,V_2]\acton W\ =\ V_1\acton V_2\acton W -  V_2\acton V_1\acton W~,
                \end{gathered}
            \end{equation}
            \begin{equation}
                \begin{gathered}
                    g\acton(\sft(h_1)\acton h_2)\ =\ \sft(g\acton h_1)\acton(g\acton h_2)
                    \\
                    \Rightarrow\quad
                    V\acton[W_1,W_2]\ =\ [V\acton W_1,W_2]+[W_1,V\acton W_2]~,
                \end{gathered}
            \end{equation}
            \begin{equation}\label{eq:app_identity_AhA/h}
                \begin{gathered}
                    g_1\acton(h(g_2\acton h^{-1}))\ =\ \big[(g_1\acton h)h^{-1}\big]h((g_1 g_2 g_1^{-1})\acton h^{-1})\big[(g_1 g_2 g_1^{-1})\acton(h(g_1\acton h^{-1}))\big]
                    \\
                    \Rightarrow\quad V_1\acton(h(V_2\acton h^{-1}))\ =\ V_2\acton(h(V_1\acton h^{-1}))+h([V_1,V_2]\acton h^{-1})
                    \\
                    -[h(V_1\acton h^{-1}),h(V_2\acton h^{-1})]
                \end{gathered}
            \end{equation}
        \end{subequations}
        for all $g,g_{1,2}\in\sfG$, for all $h,h_{1,2}\in\sfH$, for all $V,V_{1,2}\in\frg$, and for all $W,W_{1,2}\in\frh$. These implications are derived by appropriate linearisation. In particular, one would write $g=\exp(\eps V)$ or $h=\exp(\eps W)$, for instance, and then linearise in $\eps$.
        
        We also have 
        \begin{subequations}
            \begin{align}
                \sft(h_1)\acton(h_2(V\acton h_2^{-1}))\ &=\ h_1h_2(V\acton (h_1h_2)^{-1})-h_1(V\acton h_1^{-1})~,\label{eq:app_identity_ht(B)/h}
                \\
                h(\sft(W)\acton h^{-1})\ &=\ \sft(h)\acton W-W
            \end{align}
        \end{subequations}
        for all $h,h_{1,2}\in\sfH$, for all $V\in\frg$, and for all $W\in\frh$.
        
        \paragraph{Differential properties.}
        Let $g\in\scC^\infty(U,\sfG)$ and $h\in\scC^\infty(U,\sfH)$ for $U\subseteq X$ open. Then,
        \begin{equation}
            (g\acton h)^{-1}\rmd(g\acton h)\ =\ g\acton[ h^{-1}((g^{-1}\rmd g)\acton h)]+g\acton( h^{-1}\rmd h)
        \end{equation}
        which follows from formally writing $\rmd(g\acton h)=\rmd g\acton h+g\acton\rmd h$. Therefore, 
        \begin{subequations}
            \begin{align}
                \rmd(\alpha\acton h)\ &=\ \rmd\alpha\acton h+(-1)^p\alpha\acton\rmd h~,\label{eq:app_identity_d(gLambda)}
                \\
                \rmd(g\acton\beta)\ &=\ g\acton((g^{-1}\rmd g)\acton\beta)+g\acton\rmd\beta~,\label{eq:app_identity_d(hA/h)}
                \\
                \rmd(h^{-1}(\alpha\acton h))\ &=\ h^{-1}(\rmd\alpha\acton h)+(-1)^p[h^{-1}(\alpha\acton h),h^{-1}\rmd h]+(-1)^p\alpha\acton h^{-1}\rmd h 
            \end{align}
        \end{subequations}
        for all $\alpha\in\Omega^p(U,\frg)$ and for all $\beta\in\Omega^q(U,\frh)$.
        
        \subsection{Properties of the adjustment datum}\label{app:kappaProperties}
        
        Let again $(\sfH\overset{\sft}{\longrightarrow}\sfG,\acton)$ be a crossed module of Lie groups and $(\frh\overset{\sft}{\longrightarrow}\sfg,\acton)$ the corresponding crossed module of Lie algebras, respectively. 
        
        \paragraph{Algebraic properties.}
        Note that from~\eqref{eq:alternativeAdjustmentCondition_b} it follows that
        \begin{equation}\label{eq:kappa_gginv}
            g\acton\kappa(g^{-1},V)\ =\ -\kappa(g,g^{-1}Vg-\sft(\kappa(g^{-1},V)))
        \end{equation}
        for all $g\in\sfG$ and for all $V\in\frg$. Recall the linearised $\kappa$ defined in~\eqref{eq:linearisedKappa} by
        \begin{equation}
            \kappa(V_1,V_2)\ \coloneqq\ \lim_{\eps\to0}\tfrac1\eps\kappa(\exp(\eps V_1),V_2)
        \end{equation}
        for all $V_{1,2}\in\frg$. Consequently, for all $V\in\frg$ and for all $W\in\frh$ the condition~\eqref{eq:alternativeAdjustmentCondition_a} implies that 
        \begin{equation}\label{eq:kappat_linear}
            \kappa(\sft(W),V)\ =\ -V\acton W~.
        \end{equation}
        It is then immediate that
        \begin{equation}\label{eq:kappa_tt}
            \kappa(\sft(W_1),\sft(W_2))\ =\ -\kappa(\sft(W_2),\sft(W_1))\ =\ [W_1,W_2]
        \end{equation}
        for all $W_{1,2}\in\frh$.
        
        Recall from~\cref{def:generalAdjustment} that $\kappa$ is assumed to be linear in the second slot. Linearity in the first slot of the linearised $\kappa$ follows from~\eqref{eq:alternativeAdjustmentCondition_b} and~\eqref{eq:kappa_gginv}. Hence, the linearised $\kappa$ is a bilinear map $\frg\times\frg\rightarrow\frh$.
        
        Next, consider~\eqref{eq:alternativeAdjustmentCondition_b} and set $g_2=g$ and $g_1=\exp(\eps V_1)$ for $V_1\in\frg$. Similarly, one can consider the same identity but with $g_{1,2}$ interchanged. Upon linearising in $\eps$ and using the fact that $g\exp(\eps V_1)g^{-1}=\exp(\eps gV_1g^{-1})$, we get 
        \begin{equation}\label{eq:kappa_galfaX}
            \begin{aligned}
                &(gV_1g^{-1})\acton\kappa(g,V_2)+\kappa(gV_1g^{-1},gV_2g^{-1}-\sft(\kappa(g,V_2)))
                \\ 
                &\kern4cm=\ g\acton\kappa(V_1,V_2)+\kappa(g,[V_1,V_2]-\sft(\kappa(V_1,V_2)))
            \end{aligned}
        \end{equation}
        for all $V_2\in\frg$.
        
        Let us now see how $\kappa$ acts on a commutator in the first slot, starting from the identity
        \begin{equation}\label{eq:kappa_ggginv}
            \begin{aligned}
                \kappa(g_1g_2g^{-1}_1,V)\ &=\ g_1g_2\acton\kappa(g_1^{-1},V)+g_1\acton\kappa(g_2,g_1^{-1}Vg_1-\sft(\kappa(g_1^{-1},V)))
                \\
                &\kern1cm+\kappa\big(g_1,g_2g_1^{-1}Vg_1g_2^{-1}-\sft(g_2\acton\kappa(g_1^{-1},V))
                \\
                &\kern3cm-\sft(\kappa(g_2,g_1^{-1}Vg_1-\sft(\kappa(g_1^{-1},V))))\big)\,,
            \end{aligned}
        \end{equation}
        which again follows from~\eqref{eq:alternativeAdjustmentCondition_b}. Setting $g_{1,2}=\exp(\eps_{1,2}V_{1,2})$ for $V_{1,2}\in\frg$ and keeping only terms of order $\eps_1\eps_2$, the identity~\eqref{eq:kappa_ggginv} yields 
        \begin{equation}\label{eq:kappa_commutator}
            \begin{aligned}
                \kappa([V_1,V_2],V_3)\ &=\ V_1\acton\kappa(V_2,V_3)-V_2\acton\kappa(V_1,V_3) 
                \\
                &\kern1cm+\kappa(V_1,[V_2,V_3])-\kappa(V_2,[V_1,V_3])
                \\
                &\kern1cm-\kappa(V_1,\sft(\kappa(V_2,V_3)))+\kappa(V_2,\sft(\kappa(V_1,V_3)))
            \end{aligned}
        \end{equation}
        for all $V_3\in\frg$.
        
        \paragraph{Differential properties.}
        Suppose now that $g\in\scC^\infty(U,\sfG)$ and $\alpha\in\Omega^p(U,\frg)$ for $U\subseteq X$ open. Then,
        \begin{equation}\label{eq:dofkappa}
            \rmd\kappa(g,\alpha)\ =\ \kappa(g,\rmd\alpha)+g\acton\kappa(g^{-1}\rmd g,\alpha)+\kappa(g,[g^{-1}\rmd g,\alpha]-\sft(\kappa(g^{-1}\rmd g,\alpha)))~.
        \end{equation}
        This follows by formally writing $\rmd\kappa(g,\alpha)=\kappa(\rmd g,\alpha)+\kappa(g,\rmd\alpha)$ and then using~\eqref{eq:alternativeAdjustmentCondition_b} in $\kappa(g\cdot g^{-1}\rmd g,\alpha)$, utilising the fact that $g^{-1}\rmd g$ is valued in $\frg$. Note that upon combining~\eqref{eq:kappa_galfaX} and~\eqref{eq:dofkappa}, we find that
        \begin{equation}
            \begin{aligned}
                \nabla\kappa(g,\alpha)\ \coloneqq\ (\rmd+A\acton)\kappa(g,\alpha)\ &=\ \kappa(g,\nabla\alpha)+g\acton\kappa(g^{-1}Ag+g^{-1}\rmd g,\alpha) 
                \\
                &\kern0.5cm+\kappa(g,[g^{-1}Ag+g^{-1}\rmd g,\alpha]-\sft(\kappa(g^{-1}Ag+g^{-1}\rmd g,\alpha)))
                \\
                &\kern0.5cm-\kappa(A,g\alpha g^{-1}-\sft(\kappa(g,\alpha)))~. 
            \end{aligned}
        \end{equation}
        
        \paragraph{Specialisation to $\sfString(\sfG)$.}
        The adjustment datum $\hat\kappa$ for the Lie 2-group $\sfString(\sfG)$ given in~\eqref{eq:adjustmentStringGroup} (and thus also $\kappa$ for $\caL\sfG$ given in~\eqref{eq:adjustmentPathLoopGroups}) satisfies extra identities. Let $g,g_{1,2}\in\scC^\infty(X,P_0\sfG)$, $\alpha,\alpha_{1,2}\in\Omega^p(U,P_0\frg)$, and $\beta\in\Omega^q(U,L_0\frg\oplus\fru(1))$. First of all, we have 
        \begin{subequations}\label{eq:StringkappaIdentities}
            \begin{equation}
                \hat\kappa(g,\sft(\beta))\ =\ g\acton \beta - \beta~,
            \end{equation}
            from which it follows that
            \begin{equation}
                \hat\kappa(\alpha,\sft(\beta))\ =\ \alpha\acton\beta \ =\ -(-1)^{pq}\hat\kappa(\sft(\beta),\alpha)~.
            \end{equation}
        \end{subequations}
        \begin{subequations}
            In addition, $\hat\kappa$ also obeys
            \begin{align}
                \rmd\hat\kappa(g,\alpha)\ &=\ \hat\kappa(g,\rmd\alpha)+\hat\kappa(\rmd g\,g^{-1},g\alpha g^{-1})~,
                \\
                \label{eq:StringkappaIdentity_b}
                \hat\kappa(g_2g_1,\alpha)\ &=\ \hat\kappa(g_2,g_1\alpha g_1^{-1})+\hat\kappa(g_1,\alpha)~,
                \\
                \hat\kappa(g,[\alpha_1,\alpha_2])\ &=\ \hat\kappa(g\alpha_1 g^{-1},g\alpha_2 g^{-1})-\hat\kappa(\alpha_1,\alpha_2)~.
            \end{align}
        \end{subequations}
        The above identities follow directly from the definition of $\hat\kappa$ and $\acton$ for $\sfString(\sfG)$.  
        
        \subsection{Consistency of cocycle conditions}\label{app:cocycleConsistency}
        
        Here, we give the computation of the composition of the gluing relations from \cref{sec:higherPrincipalBundles}. On triple overlaps $(i,j,k)\in Y^{[3]}$, we can express $A_k$ in terms of $A_i$ in two different ways that have to be consistent
        \begin{equation}
            \begin{aligned}
                A_k\ &=\ g^{-1}_{ik}A_ig_{ik}+g^{-1}_{ik}\rmd g_{ik}-\sft(\Lambda_{ik})
                \\
                &\overset{!}{=}\ g^{-1}_{jk}\left(g^{-1}_{ij}A_ig_{ij}+g^{-1}_{ij}\rmd g_{ij}-\sft(\Lambda_{ij})\right)g_{jk}+g^{-1}_{jk}\rmd g_{jk}-\sft(\Lambda_{jk})~.
            \end{aligned}
        \end{equation}
        Using the cocycle condition 
        \begin{equation}\label{eq:cocycleForG}
            g_{ij}g_{jk}\ =\ \sft(h^{-1}_{ijk}) g_{ik}
        \end{equation}
        and the properties of $\acton$ and $\sft$, we can rewrite the second line as
        \begin{equation}
            g^{-1}_{ik}\sft(h_{ijk})A_i\sft(h^{-1}_{ijk})g_{ik}+\sft\big(g^{-1}_{ik}\acton (h_{ijk} \rmd h^{-1}_{ijk})\big)+g^{-1}_{ik}\rmd g_{ik}-\sft(\Lambda_{jk}+g^{-1}_{jk}\acton\Lambda_{ij})~.
        \end{equation}
        Finally, using identity~\eqref{eq:app_identity_t(hA1/h)} and assuming $\sft$ to be injective, we get the cocycle condition for $\Lambda$
        \begin{equation}\label{eq:cocycleForLambda}
            \Lambda_{ik}\ =\ \Lambda_{jk}+g_{jk}^{-1}\acton\Lambda_{ij}-g_{ik}^{-1}\acton(h_{ijk}\nabla_i h_{ijk}^{-1})~,
        \end{equation} 
        where again $\nabla_i=\rmd+A_i\acton$. In order to check that the cocycle condition for $\Lambda$ is consistent one needs to work over quadruple overlaps. This is a lengthy but straightforward calculation, so we leave it as an exercise.
        
        \paragraph{Derivation of the fake-flatness condition.}
        We repeat the consistency check on triple overlaps of the cocycle condition for $ B $, as we did above for $ A $: 
        \begin{equation}\label{eq:BConsistencyCheck}
            \begin{aligned}
                B_k\ &=\ g^{-1}_{ik}\acton B_i+\nabla_k\Lambda_{ik}+\tfrac12[\Lambda_{ik},\Lambda_{ik}]
                \\
                &\overset{!}{=}\ g^{-1}_{jk}\acton (g^{-1}_{ij}\acton B_i+\nabla_j\Lambda_{ij}+\tfrac12[\Lambda_{ij},\Lambda_{ij}])+\nabla_k\Lambda_{jk}+\tfrac12[\Lambda_{jk},\Lambda_{jk}]~,
            \end{aligned}
        \end{equation}
        where in the second line we have first expressed $B_k$ in terms of $B_j$ and then $B_j$ in terms of $B_i$. Thus, we have to check whether this identity holds. Using the cocycle condition~\eqref{eq:cocycleForLambda} for $\Lambda$, we can immediately cancel a few terms appearing on both sides of~\eqref{eq:BConsistencyCheck}. After employing~\eqref{eq:cocycleForG} and~\eqref{eq:app_identity_ht(B)/h}, we are left with checking that
        \begin{equation}
            \begin{aligned}
                g^{-1}_{ik}\acton(h_{ijk}\sft(B_i)\acton h^{-1}_{ijk})\ &\overset{!}{=}\ \nabla_k\left(g^{-1}_{jk}\acton \Lambda_{ij}-g^{-1}_{ik}\acton(h_{ijk}\nabla_i h^{-1}_{ijk})\right)-g^{-1}_{jk}\acton(\nabla_j\Lambda_{ij})
                \\
                &\kern1cm+[\Lambda_{jk},g^{-1}_{jk}\acton \Lambda_{ij}-g^{-1}_{ik}\acton(h_{ijk}\nabla_i h^{-1}_{ijk})]
                \\
                &\kern1cm-[g^{-1}_{jk}\Lambda_{ij},g^{-1}_{ik}\acton(h_{ijk}\nabla_i h^{-1}_{ijk})]
                \\
                &\kern1cm+\tfrac12 g^{-1}_{ik}\acton[h_{ijk}\nabla_i h^{-1}_{ijk},h_{ijk}\nabla_i h^{-1}_{ijk}]~.
            \end{aligned}
        \end{equation} 
        We first take care of the first term on the right-hand side. Using the identity
        \begin{equation}
            \nabla_A(g\acton \Lambda)\ =\ g\acton(\nabla_{g^{-1}Ag+g^{-1}\rmd g}~\Lambda)~,
        \end{equation}
        which follows from~\eqref{eq:app_identity_d(gLambda)}, and the properties of $\acton$ and $\sft$ we get
        \begin{equation}
            \begin{aligned}
                \nabla_k(g^{-1}_{jk}\acton\Lambda_{ij})\ &=\ g^{-1}_{jk}\acton(\nabla_j\Lambda_{ij})-[\Lambda_{jk},g^{-1}_{jk}\acton\Lambda_{ij}]~,
                \\
                \nabla_k\left(g^{-1}_{ik}\acton(h_{ijk}\nabla_i h^{-1}_{ijk})\right)\ &=\ g^{-1}_{ik}\acton\left(\nabla_i(h_{ijk}\nabla_i h^{-1}_{ijk})\right)-[\Lambda_{ik},g^{-1}_{ik}\acton(h_{ijk}\nabla_i h^{-1}_{ijk})]~, 
            \end{aligned}
        \end{equation} 
        where we have also used the cocycle condition for $A$. Furthermore, from~\eqref{eq:app_identity_AhA/h},~\eqref{eq:app_identity_d(hA/h)}, and using $\rmd(h\rmd h^{-1})=-\tfrac12[h\rmd h^{-1},h\rmd h^{-1}]$, we get
        \begin{equation}
            \nabla_A(h\nabla_A h^{-1})\ =\ h\left(\rmd A+\tfrac12[A,A]\right)\acton h^{-1}-\tfrac12\left[h\nabla_A h^{-1},h\nabla_A h^{-1}\right].
        \end{equation}
        Putting everything together we get after some gymnastics
        \begin{equation}
            \begin{aligned}
                &g^{-1}_{ik}\acton\left(h_{ijk}\left(\rmd A_i+\tfrac12[A_i,A_i]+\sft(B_i)\right)\acton h^{-1}_{ijk}\right)
                \\
                &\kern1cm=\ \left[\Lambda_{ik}-\Lambda_{jk}-g_{jk}^{-1}\acton \Lambda_{ij}+g_{ik}^{-1}\acton(h_{ijk}\nabla_i h_{ijk}^{-1}),g_{ik}^{-1}\acton(h_{ijk}\nabla_i h_{ijk}^{-1})\right].
            \end{aligned}
        \end{equation}
        Again using~\eqref{eq:cocycleForLambda} the right side vanishes, and we are left with
        \begin{equation}\label{eq:appFakeFlatUnadjusted}
            g^{-1}_{ik}\acton\left(h_{ijk}(F_i\acton h^{-1}_{ijk})\right)\ =\ 0~.
        \end{equation}
        Using the cocycle condition~\eqref{eq:cocycleForG} again, this is equivalent to
        \begin{equation}
            (g_{ij}g_{jk})^{-1}\acton\left(h^{-1}_{ijk}(F_i\acton h_{ijk})\right)\ =\ 0~.
        \end{equation}
        In the adjusted case, formulas~\eqref{eq:BConsistencyCheck} get the adjustment term $\kappa(g^{-1},F)$, which modifies this condition to
        \begin{equation}\label{eq:appFakeFlatAdjusted}
            (g_{ij}g_{jk})^{-1}\acton\left(h^{-1}_{ijk}(F_i\acton h_{ijk})\right)\ =\ \kappa(g^{-1}_{ik},F_i)-g_{jk}^{-1}\acton\kappa(g^{-1}_{ij},F_i)-\kappa(g^{-1}_{jk},F_j)~.
        \end{equation}
        
        Finally, let us remark that the above calculations can be straightforwardly adapted also to the case of gauge transformations. Therefore, all the formulas are consistent if and only if~\eqref{eq:appFakeFlatUnadjusted} holds in the unadjusted case or~\eqref{eq:appFakeFlatAdjusted} holds in the adjusted case. In particular, the cocycles glue correctly and (higher) gauge transformations compose.
        
        \section{Helpful lightening reviews}
        
        Below, we provide concise reviews of three important topics relevant to our discussion, to provide a more self-contained account: Yang--Mills theory on reductive homogeneous spaces, group structures on manifolds (including string structures), and central group extensions.
        
        \subsection{Yang--Mills theory on reductive homogeneous spaces}\label{app:YMCoset}
        
        To support the discussion in \cref{sec:Spin4Bundle}, we briefly review the construction of the canonical connection on a reductive homogeneous space, see also~\cite[Section II.11]{Kobayashi:1963}. This construction is particularly appealing, as it results in solutions to the Yang--Mills equations~\cite{Trautman:1977im,Harnad:1980ct} that are, in many cases, Bogomolny--Prasad--Sommerfield (BPS) solutions~\cite{Trautman:1977im,Bais:1985ns}.\footnote{For a recent discussion of instantons on coset spaces, see also~\cite{Ivanova:2009yi}.} 
        
        \paragraph{Reductive coset spaces.}
        Let $\sfG$ be a Lie group, $\sfH$ a closed Lie subgroup of $\sfG$, and let $\frg$ and $\frh$ be the corresponding Lie algebras. We further assume that $\sfG/\sfH$ is a \uline{reductive homogeneous space}, that is, there is a decomposition
        \begin{equation}\label{eq:reductiveCosetDecomposition}
            \frg\ \cong\ \frh\oplus\frm
        \end{equation}
        with $\frm$ being $\Ad(\sfH)$-invariant. The latter condition implies that $[\frh,\frm]\subseteq\frm$, and the converse is true if $\sfH$ is connected. Altogether, we have the relations
        \begin{equation}\label{eq:reductiveCosetConditions}
            [\frh,\frh]\ \subseteq\ \frh~,
            \quad
            [\frh,\frm]\ \subseteq\ \frm~,
            \eand
            [\frm,\frm]\ \subseteq\ \frh\oplus\frm~.
        \end{equation}
        We denote the Cartan--Killing form on $\frg$ by $\inner{-}{-}$ and assume that the decomposition~\eqref{eq:reductiveCosetDecomposition} is orthogonal with respect to $\inner{-}{-}$.
        
        Let $\theta$ be the (left-invariant) Maurer--Cartan form on $\sfG$. Because of the decomposition~\eqref{eq:reductiveCosetDecomposition}, we may write
        \begin{equation}
            \theta\ =\ A+m
            \ewith
            A\ \in\ \Omega^1(\sfG,\frh)
            \eand
            m\ \in\ \Omega^1(\sfG,\frm)~.
        \end{equation}
        Under the right $\sfH$-action $g\mapsto gh$ on $\sfG$ for all $g\in\sfG$ and for all $h\in\sfH$, we have
        \begin{equation}
            A\ \mapsto\ h^{-1}Ah+h^{-1}\rmd h
            \eand
            m\ \mapsto\ h^{-1}mh~.
        \end{equation}
        Furthermore, the Maurer--Cartan equation $\rmd\theta+\frac12[\theta,\theta]=0$ then decomposes as
        \begin{equation}\label{eq:decompositionMCFormReductive}
            \rmd A+\tfrac12[A,A]\ =\ -\tfrac12[m,m]_\frh
            \eand
            \nabla_Am\ =\ -\tfrac12[m,m]_\frm~,
        \end{equation}
        where $\rmd$ denotes the exterior derivative on $\sfG$ and the subscript indicates the projection onto the respective subspaces.
        
        \paragraph{Yang--Mills connection.}
        Recall that a principal $\sfG$-bundle $\pi:P\rightarrow X$ over some manifold $X$ can conveniently be described subordinate to itself, because the pullback bundle $\pi^*P\coloneqq P\times_XP=P^{[2]}$ in
        \begin{equation}
            \begin{tikzcd}
                P\times_X P\arrow[r]\arrow[d] & P\arrow[d,"\pi"]
                \\
                P\arrow[r,"\pi"] & X
            \end{tikzcd}
        \end{equation}
        is trivial. In particular, there is a global section $P\rightarrow P^{[2]}$ given by $p\mapsto(p,p)$ for all $p\in P$. Hence, $P^{[2]}\cong P\times\sfG$ and, more generally, $P^{[n]}\cong P\times\sfG^{n-1}$. The \v Cech cocycles subordinate to the cover $\pi$ are then given by maps $(p_1,p_2)\mapsto g_{p_1p_2}$ for all $(p_1,p_2)\in P^{[2]}$, where $g_{p_1p_2}\in \sfG$ is uniquely defined by
        \begin{equation}
            p_2\ \eqqcolon\ p_1g_{p_1p_2}~,
        \end{equation}
        and the cocycle condition $g_{p_1p_2}g_{p_2p_3}=g_{p_1p_3}$ is automatically satisfied for all $(p_1,p_2,p_3)\in P^{[3]}$. In the case of the principal $\sfH$-bundle $\sfG\rightarrow\sfG/\sfH$, this yields the cocycle
        \begin{equation}
            \begin{aligned}
                g\,:\,\sfG^{[2]}\ &\rightarrow\ \sfH~,
                \\
                (g_1,g_2)\ &\mapsto\ h_{g_1g_2}\ \coloneqq\ g_1^{-1}g_2~.
            \end{aligned}
        \end{equation}
        
        Let $\{U_i\}_{i\in I}$ be a collection of open subsets of $X$ such that $X=\bigcup_{i\in I}U_i$ and $U_i$, $U_i\cap U_j$, $U_i\cap U_j\cap U_k$, etc for all $i,j,k,\ldots\in I$ are contractible. We shall work with respect to the cover $Y\coloneqq\bigsqcup_{i\in I}U_i$ of $X$. Furthermore, let $s_i:U_i\rightarrow\sfG$ be a local section. Using~\eqref{eq:decompositionMCFormReductive}, we define the 1-forms
        \begin{equation}\label{eq:canonicalConnection}
            A_{s_i}\ \coloneqq\ s_i^*A\ \in\ \Omega^1(U_i,\frh)
            \eand
            m_{s_i}\ \coloneqq\ s_i^*m\ \in\ \Omega^1(U_i,\frm)~,
        \end{equation}
        on $Y_i$, and they satisfy
        \begin{equation}\label{eq:decompositionMCFormReductive2}
            \rmd A_{s_i}+\tfrac12[A_{s_i},A_{s_i}]\ =\ -\tfrac12[m_{s_i},m_{s_i}]_\frh
            \eand
            \nabla_{A_{s_i}}m_{s_i}\ =\ -\tfrac12[m_{s_i},m_{s_i}]_\frm~,
        \end{equation}
        where now $\rmd$ denotes the exterior derivative on $Y_i$. Next, let $\star$ be the Hodge star operator with respect to the metric on $\sfG/\sfH$ that is given by the pullback via ${s_i}$ of the left-invariant metric on $\sfG$ which is induced by the Cartan--Killing form $\inner{-}{-}$. It then follows that
        ~\cite{Trautman:1977im,Harnad:1980ct}
        \begin{equation}\label{eq:canonicalCurvature}
            \nabla_{A_{s_i}}{\star}F_{s_i}\ =\ 0
            \ewith
            F_{s_i}\ \coloneqq\ \rmd A_{s_i}+\tfrac12[A_{s_i},A_{s_i}]~.
        \end{equation}
        To show this, one makes use of the Jacobi identity and the relations~\eqref{eq:reductiveCosetConditions} and~\eqref{eq:decompositionMCFormReductive2}. Hence, $A_{s_i}$ is a Yang--Mills gauge potential on $U_i$ for the gauge group $\sfH$. 
        
        Finally, let $(g_1,g_2)\in\sfG^{[2]}$ and suppose that $s_{1,2}:U_{1,2}\rightarrow\sfG$ are two local sections such that $(s_{1,2}\circ\pi)(g_{1,2})=g_{1,2}$ and $s_2=s_1h_{g_1g_2}$. Therefore, with $A_{s_{1,2}}\coloneqq s_{1,2}^*A\in\Omega^1(U_{1,2},\frh)$, we obtain
        \begin{equation}
            A_{s_2}\ =\ h^{-1}_{g_1g_2}A_{s_1}h_{g_1g_2}+h^{-1}_{g_1g_2}\rmd h_{g_1g_2}~.
        \end{equation}
        
        \subsection{Group structures on principal bundles}\label{app:group_structures}
        
        \paragraph{$\sfG$-structures on $\sfH$-bundles.} 
        Consider a principal $\sfH$-bundle $P$ over a manifold $X$. A \uline{reduction/lift of the structure group from $\sfH$ to $\sfG$} is a monomorphism/epimorphism of Lie groups $\phi:\sfG\rightarrow\sfH$, a principal $\sfG$-bundle $Q$ over $X$, and a bundle isomorphism $\phi_Q:Q\times_\sfG\sfH\xrightarrow{\cong}P$.
        
        If we consider a smooth vector bundle $E\rightarrow X$ of rank $n$, then its frame bundle is a principal $\sfGL(n)$-bundle over $X$. A \uline{$\sfG$-structure} on $E$ is then a reduction of the structure group from $\sfGL(n)$ to $\sfG$ or the reduction from a $\sfG$-structure to another $\sfG$-structure.
        
        Well-known examples of $\sfG$-structures are obtained by considering the tangent bundle $E=TX$. An \uline{orientation} on $X$ is the $\sfG$-structure arising from a reduction of the structure group from $\sfGL(n)$ to $\sfGL^+(n)$ with $n=\dim(X)$. A \uline{Riemannian metric} corresponds to a reduction of the structure group from $\sfGL(n)$ to $\sfO(n)$, which is always possible. A \uline{volume form} on a Riemannian manifold corresponds to a further reduction of the structure group from $\sfO(n)$ to $\sfSO(n)$. Finally, a \uline{spin structure} is a lift of the structure group from $\sfSO(n)$ to $\sfSpin(n)$. The obstructions to a manifold being orientable and carrying a volume form is governed by the first Stiefel--Whitney class, whilst the obstruction to a manifold being spin is governed by the second Stiefel--Whitney class. 
        
        \paragraph{String structures.} 
        If one is interested in the classical dynamics of strings, it is reasonable to consider the free loop space $LX$ of a manifold $X$. It can now be argued that $LX$ should be considered orientable when $X$ is a spin manifold~\cite{Witten:1987cg}; see also the discussions in~\cite{Killingback:1986rd,mclaughlin1992orientation}. As then shown in~\cite{mclaughlin1992orientation}, the loop space $LX$ is spin, if $X$ is \uline{string}, that is, a lift of the structure group of the frame bundle from $\sfSpin(n)$ to $\sfString(n)$, the \uline{string group}. 
        
        Specifically, $\sfString(n)$ is a topological group that fits into the Whitehead tower
        \begin{equation}\label{eq:whitehead_tower}
            \cdots\ \longrightarrow\ \underbrace{\sfString(n)}_{\eqqcolon\,W_3}\ \longrightarrow\ \underbrace{\sfSpin(n)}_{\eqqcolon\,W_2}\ \longrightarrow\ \underbrace{\sfSpin(n)}_{\eqqcolon\,W_1}\ \longrightarrow\ \underbrace{\sfSO(n)}_{\eqqcolon\,W_0}\ \longrightarrow\ \underbrace{\sfO(n)}_{\eqqcolon\,W}~,
        \end{equation}
        where the group homomorphisms are isomorphisms on all but the lowest homotopy groups: $\pi_i(W_j)=0$ for $i\leq j$ and the homotopy groups $\pi_i(W_j)$ for all $i>j$ are all isomorphic. The string group $\sfString(n)$ is thus a 3-connected cover of $\sfSpin(n)$~\cite{Stolz:2004aa}. It should be emphasised that the string group is not uniquely defined via the sequence~\eqref{eq:whitehead_tower}. In fact, it is most conveniently described by an equivalent categorified or 2-group $\sfString(n)$, cf.~\cite{Nikolaus:2011zg} and references therein. A string structure is then simply a topological principal $\sfString(n)$-bundle. In our cases, we shall work with a strict, but infinite-dimensional Lie 2-group $\sfString(n)$; for the complications encountered when working with finite-dimensional models, see e.g.~\cite{Demessie:2016ieh}.
        
        We note that the obstruction for a $\sfSpin(n)$-bundle $P$ to allow for a lift of the structure group to $\sfString(n)$, and thus a string structure, is (for a simply connected manifold and $n\geq 5$) precisely $\tfrac12p_1(P)$ with $p_1(P)$ the first Pontryagin class of $P$~\cite{mclaughlin1992orientation}. We recall from~\cite[Section 5]{Stolz:2004aa} that if a principal bundle $P$ admits a string structure, then the possible choices form a torsor for the cohomology group $H^3(X,\IZ)$. Furthermore,~\cite{Waldorf:2009uf} has shown that isomorphism classes of string structures on a principal bundle $P$ are in bijection with isomorphism classes of gerbes trivialising the Chern--Simons 2-gerbe associated to $P$ by the construction of~\cite{Carey:2004xt}.
        
        \subsection{Central extension and cocycles}\label{app:centralExtensions}
        
        \paragraph{Abstract setting.}
        A \uline{central extension} $\sfE$ of a group $\sfG$ by an Abelian group $\sfA$ is a short exact sequence of groups
        \begin{equation}\label{eq:centralExtension}
            \unit\ \longrightarrow\ \sfA\ \overset{\iota}{\longrightarrow}\ \sfE\ \overset{\pi}{\longrightarrow}\ \sfG\ \longrightarrow\ \unit
        \end{equation}
        such that $\im(\iota)$ is contained in the centre of $\sfE$. Let now $s:\sfG\rightarrow\sfE$ be a section of $\pi:\sfE\rightarrow\sfG$, that is, $\pi\circ s=\id$. Then, for all $g_{1,2}\in\sfG$,
        \begin{equation}
            \pi(s(g_1)s(g_2))\ =\ \pi(s(g_1))\pi(s(g_2))\ =\ \pi(s(g_1g_2))~,
        \end{equation}
        and we parametrise the failure of $s$ to be a group homomorphism by a map $c:\sfG\times\sfG\rightarrow\sfA$ according to
        \begin{equation}\label{eq:cocycle}
            s(g_1g_2)\iota(c(g_1,g_2))\ \coloneqq\ s(g_1)s(g_2)~.
        \end{equation}
        It then follows that $c$ satisfies the condition
        \begin{equation}\label{eq:cocycleConditionapp}
            c(g_1,g_2)c(g_1g_2,g_3)\ =\ c(g_1,g_2g_3)c(g_2,g_3)
        \end{equation}
        for all $g_{1,2,3}\in\sfG$. This condition, when evaluated for $(g_1,g_2,g_3)=(g,\unit,g)$ and $(g_1,g_2,g_3)=(\unit,\unit,g)$ for $g\in\sfG$, yields
        \begin{equation}\label{eq:normalisation1}
            c(\unit,\unit)\ =\ c(\unit,g)\ =\ c(g,\unit)~.
        \end{equation}
        Furthermore, we have a bijection of sets $\phi:\sfG\times\sfA\to\sfE$ defined by $\phi(g,a)\coloneqq s(g)\iota(a)$ for all $(g,a)\in\sfG\times\sfA$. It then straightforwardly follows that $\phi(g_1,a_1)\phi(g_1,a_1)=s(g_1g_2)\iota(a_1a_2c(g_1,g_2))$ and hence, we obtain a product on $\sfG\times\sfA$ given by
        \begin{equation}\label{eq:product}
            (g_1,a_1)(g_2,a_2)\ \coloneqq\ (g_1g_2,a_1a_2c(g_1,g_2))
        \end{equation}
        for all $(g_{1,2},a_{1,2})\in\sfG\times\sfA$. This product is associative if and only if~\eqref{eq:cocycleConditionapp} holds. Evidently, the product~\eqref{eq:product} makes $\phi$ an isomorphism of groups. Moreover, note that $(\unit,(c(\unit,\unit))^{-1})$ is the neutral element with respect to~\eqref{eq:product}, and the inverse of an element $(g,a)\in\sfG\times\sfA$ is $(g,a)^{-1}=(g^{-1},(ac(\unit,\unit)c(g,g^{-1}))^{-1})$.
        
        If now $\tilde s$ is another section of $\sfE$, from $\pi\circ\tilde s=\id=\pi\circ s$ it then immediately follows that $\tilde s(g)=s(g)\iota(d(g))$ for all $g\in\sfG$ and where $d:\sfG\rightarrow\sfA$. In turn, this yields
        \begin{equation}\label{eq:coboundaryCondition}
            \tilde c(g_1,g_2)\ =\ c(g_1,g_2)(d(g_1g_2))^{-1}d(g_1)d(g_2) 
        \end{equation}
        for all $g_{1,2}\in\sfG$. Using~\eqref{eq:coboundaryCondition},~\eqref{eq:normalisation1} then implies that without loss of generality we can always normalise $c$ so that
        \begin{equation}\label{eq:normalisation2}
            c(\unit,\unit)\ =\ c(\unit,g)\ =\ c(g,\unit)\ =\ \unit
        \end{equation}
        for all $g\in\sfG$.
        
        Consider now the nerve of the groupoid $\sfB\sfG$. That is, let $\sfG^p\coloneqq\sfG\times\cdots\times\sfG$ ($p$-copies) with $\sfG^0\coloneqq\unit$ the trivial group. We define (face) maps $m_i:\sfG^{p}\rightarrow\sfG^{p-1}$ for $i=1,\ldots,p+1$ by
        \begin{equation}
            m_i(g_1,\ldots,g_{p})\ \coloneqq\ 
            \begin{cases}
                (g_2,\ldots,g_{p}) &\efor i=1
                \\
                (g_1,\ldots,g_{i-2},g_{i-1}g_i,g_{i+1},\ldots,g_{p}) &\efor i=2,\ldots,p
                \\
                (g_1,\ldots,g_{p-1}) &\efor i=p+1
            \end{cases}
        \end{equation}
        for all $g_{1,\ldots,p}\in\sfG$. This allows us to introduce the map $\delta:\sfMap(\sfG^p,\sfA)\rightarrow\sfMap(\sfG^{p+1},\sfA)$ by
        \begin{equation}
            \delta(c)(g_1,\ldots,g_{p+1}) \coloneqq \prod_{i=0}^{\left\lceil p/2\right\rceil}(c\circ m_{2i+1})(g_1,\ldots,g_{p+1})((c\circ m_{2i+2})(g_1,\ldots,g_{p+1}))^{-1}
        \end{equation}
        for all $g_{1,\ldots,p+1}\in\sfG$, and it is easy to see that this is a differential yielding the complex
        \begin{equation}\label{eq:extensionComplex}
            \underbrace{\sfMap(\sfG^0,\sfA)}_{\cong\,\sfA}\ \overset{\delta}{\longrightarrow}\ \sfMap(\sfG^1,\sfA)\ \overset{\delta}{\longrightarrow}\ \sfMap(\sfG^2,\sfA)\ \overset{\delta}{\longrightarrow}\ \cdots~.
        \end{equation}
        Furthermore, for $c\in\sfMap(\sfG^2,\sfA)$, we have that~\eqref{eq:cocycleConditionapp} is equivalent to $\delta(c)=\unit$ and for $d\in\sfMap(\sfG^1,\sfA)$, the condition~\eqref{eq:coboundaryCondition} can be written as $\tilde c(g_1,g_2)=c(g_1,g_2)\delta(d)(g_1,g_2)$ for all $g_{1,2}\in\sfG$. In summary, central extensions are classified by the second cohomology group $H^2(\sfG,\sfA)$ associated with the complex~\eqref{eq:extensionComplex}. In this context, we shall refer to~\eqref{eq:cocycleConditionapp} as a \uline{group cocycle condition} and to~\eqref{eq:coboundaryCondition} as a \uline{group coboundary condition}, respectively. 
        
        Likewise, the first cohomology group $H^1(\sfG,\sfA)$ can be interpreted in terms of the automorphisms of the central extension. Indeed, we say that two automorphisms of the central extension~\eqref{eq:centralExtension} are equivalent whenever they differ by the conjugation automorphism induced by an element of $\sfA$. Then, it can be shown that $H^1(\sfG,\sfA)$ acts transitively on the set of all automorphisms of the central extension~\eqref{eq:centralExtension} modulo this equivalence. See e.g.~\cite{Azcarraga:2011hqa} and in particular~\cite[Section 4]{Mickelsson:1989hp} for more details.
        
        \begin{remark}\label{rmk:normalisation}
            We shall always assume that the cocycles in $H^2(\sfG,\sfA)$ are normalised such that~\eqref{eq:normalisation2} holds. This means that the neutral element with respect to the product~\eqref{eq:product} is $(\unit,\unit)$ and the inverse of an element $(g,a)\in\sfG\times\sfA$ is given by $(g,a)^{-1}=(g^{-1},(ac(g,g^{-1}))^{-1})$. Note that the cocycle condition~\eqref{eq:cocycleConditionapp} then also implies that 
            \begin{equation}\label{eq:conjugation}
                (g_1,a_1)^{-1}(g_2,a_2)(g_1,a_1)\ =\ (g^{-1}_1g_2g_1,a_2c(g_1^{-1}g_2,g_1)(c(g_1,g_1^{-1}g_2))^{-1})
            \end{equation}
            for all $(g_{1,2},a_{1,2})\in\sfG\times\sfA$.
        \end{remark}
        
        \paragraph{Fr{\'e}chet--Lie groups.}
        So far, we have not specified the type of groups for which we wish to discuss central extensions. As we are mainly concerned with Lie groups, not necessarily finite-dimensional, let now $\sfG$ be a Fr{\'e}chet--Lie group and $\sfA$ an Abelian Fr{\'e}chet--Lie group. We are interested in central extensions~\eqref{eq:centralExtension} for which $\sfE$ carries a Fr{\'e}chet--Lie group structure with respect to the product~\eqref{eq:product} and $\pi:\sfE\rightarrow\sfG$ is a principal $\sfA$-bundle. If, in addition, $\sfG$ is also connected, then one can show that this is the case if and only if the cocycle $c\in H^2(\sfG,\sfA)$ is smooth in a neighbourhood of $(\unit,\unit)\in\sfG\times\sfG$. See~\cite[Prop.~3.11]{Tuynman:1987ij} for the finite-dimensional case and~\cite{neeb2002central} for the infinite-dimensional case.
        
        \paragraph{Principal $\sfU(1)$-bundle.} Let us now specialise to $\sfA=\sfU(1)$. A central extension of Fr{\'e}chet--Lie groups~\eqref{eq:centralExtension} gives rise to a principal circle bundle. We shall develop the relation between the group cocycle and the corresponding \v Cech cocycle $t\in \check H^2(\sfG,\IZ)$ in the following.\footnote{We thank Paul Skerritt for discussions on these constructions.} Let $s:\sfG\rightarrow\sfE$ be a section that is smooth in an open neighbourhood $U$ of $\unit\in\sfG$. Then, the associated cocycle characterising the extension is given by~\eqref{eq:cocycle}, and it   is smooth in the open neighbourhood $U\times U$ of $(\unit,\unit)\in\sfG\times\sfG$. We can construct an open cover $\{U_g\}_{g\in\sfG}$ of $\sfG$ from the patches $U_g\coloneqq L_g(U)$, where $L$ denotes left-multiplication. On each patch $U_g$, we consider the section $\sigma_g:U_g\rightarrow\pi^{-1}(U_g)$ of $\pi:\sfE\rightarrow\sfG$ defined by 
        \begin{equation}\label{eq:locSec}
            \sigma_g\ \coloneqq\ L_{s(g)}\circ s\circ L_{g^{-1}}~.
        \end{equation}
        Explicitly, $\sigma_g(h)=s(g)(s(g^{-1}h))=s(h)\iota(c(g,g^{-1}h))$ for all $h\in U_g$. From the smoothness of $s$ on $U$ it follows that $\sigma_g$ is smooth on $U_g$. On non-empty intersections of patches we introduce smooth maps $t_{g_1g_2}:U_{g_1}\cap U_{g_2}\rightarrow\sfU(1)$ by 
        \begin{equation}
            \sigma_{g_1}(h)(\sigma_{g_2}(h))^{-1}\ =\ \iota(c(g_1,g_1^{-1}h)(c(g_2,g_2^{-1}h))^{-1})\ \eqqcolon\ \iota(t_{g_1g_2}(h))
        \end{equation}
        for all $h\in U_{g_1}\cap U_{g_2}$.\footnote{Note that $U_g\ni h\mapsto c(g,g^{-1}h)\in\sfU(1)$ is not necessarily smooth since $c$ is only smooth on a neighbourhood of $(\unit,\unit)\in\sfG\times\sfG$.} Evidently, for all $g_{1,2,3}\in\sfG$ with $U_{g_1}\cap U_{g_2}\cap U_{g_3}\neq\emptyset$ we obtain
        \begin{equation}
            t_{g_1g_2}(h)t_{g_2g_3}(h)\ =\ t_{g_1g_3}(h)
        \end{equation}
        for all $h\in U_{g_1}\cap U_{g_2}\cap U_{g_3}$. That is, the maps $t_{g_1g_2}:U_{g_1}\cap U_{g_2}\rightarrow\sfU(1)$ satisfy a {\v C}ech cocycle condition. Similarly, we introduce smooth maps $t_g:U_g\rightarrow\sfU(1)$ given by $\iota(t_g(h))\coloneqq d(g)d(g^{-1}h)$ for all $h\in U_g$, and the group coboundary transformation~\eqref{eq:coboundaryCondition} yields
        \begin{equation}
            \tilde t_{g_1g_2}(h)\ =\ t_{g_1}(h)t_{g_1g_2}(h)(t_{g_2}(h))^{-1}~,
        \end{equation}
        that is, {\v C}ech coboundary transformation. Altogether, the central extension is characterised by an element of the {\v C}ech cohomology group $\check H^1(\{U_g\}_{g\in\sfG},\sfU(1))$ and, upon taking the direct limit, of $\check H^1(\sfG,\sfU(1))\cong\check H^2(\sfG,\IZ)$. It is important to realise, however, that the map $H^2(\sfG,\sfU(1))\rightarrow\check H^2(\sfG,\IZ)$ thus constructed is not a bijection but rather a many-to-one map. Indeed, the {\v C}ech cohomology only knows about the bundle structure of $\sfE\rightarrow\sfG$ but not about the group structure. In fact, one can have non-trivial central extensions which admit global smooth sections but which are not morphisms of groups.  
        
        \section{Proofs for \texorpdfstring{\cref{sec:stringBundlesWithConnection}}{Section 4}}\label{app:proofs}
        
        \paragraph{A formula for $(\Ad_g)^* \omega$.}
        Consider the 2-form curvature $\omega$ defined in~\eqref{eq:2FormCurvatureKacMoody} and the left-invariant 1-form $\xi_g$ for fixed $g\in\sfG$ defined in~\eqref{eq:Baez1Form}. Then 
        \begin{equation}
            (\Ad_g)^* \omega\ =\ \omega+\rmd\xi_g~,
        \end{equation}
        cf.~\cite[Equation (4.6.6)]{Pressley:1988qk}. In order to prove this relation, it is sufficient to establish this at the identity $\unit\in L_0\sfG$ due to left-invariance. For any two Lie algebra elements $U,V\in L_0\frg$ we have
        \begin{equation}
            \begin{aligned}
                \omega_\unit(\Ad_gU,\Ad_gV)\ &=\ \frac{\rmi}{2\pi}\int_0^1\rmd r\innerLarge{\Ad_{g(r)}U(r)}{\parder[\Ad_{g(r)}V(r)]{r}}
                \\ 
                &=\ \frac{\rmi}{2\pi}\int_0^1\rmd r\innerLarge{\Ad_{g(r)}U(r)}{\Ad_{g(r)}\left(\parder[V(r)]{r}+\left[g^{-1}(r)\parder[g(r)]{r},V(r)\right]\right)}.
            \end{aligned}
        \end{equation}
        Using the $\Ad$-invariance of the inner product and Cartan's formula $(\rmd\xi_g|_\unit)(U,V)=-\xi_g|_\unit([U,V])$, we see that
        \begin{equation}
            \begin{aligned}
                \omega_\unit(\Ad_gU,\Ad_gV)\ &=\ \frac{\rmi}{2\pi}\int_0^1\rmd r\innerLarge{U(r)}{\parder[V(r)]{r}}-\frac{\rmi}{2\pi}\int_0^1\rmd r\innerLarge{[U(r),V(r)]}{g^{-1}(r)\parder[g(r)]{r}}
                \\
                &=\ \omega_\unit(U,V)+(\rmd\xi_g|_e)(U,V)~.
            \end{aligned}
        \end{equation}
        This verifies the desired formula.
        
        \paragraph{Peiffer identity.}
        Next, we wish to show that $(\sft(\hat h_1)\acton\hat h_2)\hat h_1\hat h_2^{-1}\hat h_1^{-1}\in\sfN$ for all $\hat h_{1,2}\in P_0L_0\sfG\times\sfU(1)$. Note that for all $\hat n_{1,2}\in\sfN$, we have 
        \begin{equation}
            \begin{aligned}
                &(\sft(\hat h_1\hat n_1)\acton(\hat h_2\hat n_2))\hat h_1\hat n_1(\hat h_2\hat n_2)^{-1}(\hat h_1\hat n_1)^{-1}
                \\
                &\kern1cm=\ (\sft(\hat h_1)\acton(\hat h_2\hat n_2))\hat h_1\hat n_1\hat h_2^{-1}\hat n_2^{-1}\hat n_1^{-1}\hat h_1^{-1} 
                \\
                &\kern1cm=\ (\sft(\hat h_1)\acton\hat h_2)(\sft(\hat h_1)\acton\hat n_2)\hat h_1\hat n_1\hat h_2^{-1}\hat n_2^{-1}\hat n_1^{-1}\hat h_1^{-1}
                \\
                &\kern1cm=\ (\sft(\hat h_1)\acton\hat h_2)\hat h_1\hat h_2^{-1}\hat h_1^{-1}
                \\
                &\kern2cm\times\underbrace{\hat h_1\hat h_2\hat h_1^{-1}(\sft(\hat h_1)\acton\hat n_2)\hat h_1\hat h_2^{-1}\hat h_1^{-1}}_{\in\,\sfN}
                \\
                &\kern2cm\times\underbrace{\hat h_1\hat h_2\hat n_1\hat h_2^{-1}\hat n_2^{-1}\hat n_1^{-1}\hat h_1^{-1}}_{\in\,\sfN}~,
            \end{aligned}
        \end{equation}
        where we have used~\eqref{eq:tMomorphismString2Group}, the normality of~$\sfN$, and the fact that the action closes on~$\sfN$. Consequently, $(\sft(\hat h_1)\acton\hat h_2)\hat h_1\hat h_2^{-1}\hat h_1^{-1}\in\sfN$ implies the Peiffer identity $\sft([\hat h_1])\acton[\hat h_2]=[\hat h_1][\hat h_2][\hat h_1]^{-1}$ for all $[\hat h_1],[\hat h_2]\in\big(P_0L_0\sfG\times\sfU(1)\big)/\sfN$ for the morphism~\eqref{eq:tMomorphismString2Group} and the action~\eqref{eq:actionOfP0G}. 
        
        To verify that indeed $(\sft(\hat h_1)\acton\hat h_2)\hat h_1\hat h_2^{-1}\hat h_1^{-1}\in\sfN$, let us set $\hat h_{1,2}=(f_{1,2},z_{1,2})$. Then, using~\eqref{eq:conjugation} together with~\eqref{eq:tMomorphismString2Group} and~\eqref{eq:actionOfP0G}, we find
        \begin{subequations}
            \begin{equation}
                f\ \coloneqq\ \boundary(f_1)f_2\boundary(f_1^{-1})f_1f_2^{-1}f_1^{-1}
            \end{equation}
            for the $P_0L_0\sfG$-component of $(\sft(\hat h_1)\acton\hat h_2)\hat h_1\hat h_2^{-1}\hat h_1^{-1}$. Evidently, $\boundary(f)=\unit\in P_0L_0\sfG$ and so $f\in L_0L_0\sfG$. Likewise, again using~\eqref{eq:conjugation} together with~\eqref{eq:tMomorphismString2Group} and~\eqref{eq:actionOfP0G}, the $\sfU(1)$-component of $(\sft(\hat h_1)\acton\hat h_2)\hat h_1\hat h_2^{-1}\hat h_1^{-1}$ is
            \begin{equation}\label{eq:peiffer_proof}
                \begin{aligned}
                    z\ &\coloneqq\ c(\boundary(f_1)f_2\boundary(f_1^{-1}),f_1f_2^{-1}f_1^{-1})\,c^{-1}(f_1f_2,f_1^{-1})
                    \\ 
                    &\kern1cm\times c(f_1^{-1},f_1f_2)\,c^{-1}(f_1f_2f_1^{-1},f_1f_2^{-1}f_1^{-1})\,\exp\left(\int_0^1\rmd t\,\xi_{\boundary(f_1)}\left(f_2^{-1}\parder[f_2]{t}\right)\right)
                \end{aligned}
            \end{equation}
            with $\xi_g$ defined in~\eqref{eq:Baez1Form}.
        \end{subequations}
        Therefore, it remains to show that
        \begin{equation}
            \sfhol^{-1}(f)\ =\ z~.
        \end{equation}
        To this end, recall that the cocycle is given as an integral over a disk with a particular boundary~\eqref{eq:triangle}. In the case of $c(\boundary(f_1)f_2\boundary(f_1^{-1}),f_1f_2^{-1}f_1^{-1})$, this boundary is, in fact, formed by two loops joined at the identity. The integral over one of the two disks produces exactly $\sfhol^{-1}(f)$. We can combine the second integral with the remaining cocycle terms in~\eqref{eq:peiffer_proof} by matching the boundaries and using~\eqref{eq:holonomy_shift}. Because $\omega$ defined in~\eqref{eq:2FormCurvatureKacMoody} is closed, we now use Stokes' theorem, and we are left with an integral along a square $\Box$ whose sides are given by 
        \begin{equation}
            \left(\boundary(f_1)f_1^{-1},\boundary(f_1)\overline{f_2},\boundary(f_1f_2)\overline{f_1^{-1}},\boundary(f_1)f_2\boundary(f_1^{-1})\right)
        \end{equation}
        with the orientation from right to left. We can then parametrise this square as
        \begin{equation}
            (s,t)\ \mapsto\ (\boundary(f_1))(r)f_2(s,r)f_1^{-1}(1-t,r)
            \eforall
            s,t\ \in\ [0,1]~,
        \end{equation}
        where $r$ is the loop parameter, to obtain
        \begin{equation}
            \int_\Box\omega\ =\ \frac{\rmi}{2\pi}\int_0^1\rmd r\int_0^1\rmd s\innerLarge{f_2^{-1}(s,r)\parder[f_2(s,r)]{s}}{\boundary(f_1^{-1}(s,r))\parder[\,\boundary(f_1(s,r))]{r}}.
        \end{equation}
        Here, we have used the left-invariance of $\omega$ to remove $\boundary(f_1)$ in the square parametrisation and~\eqref{eq:helpful_identity} to perform the $t$ integral. Upon exponentiating this expression, we see that this term precisely cancels the last factor in~\eqref{eq:peiffer_proof}.
        
        \paragraph{Infinitesimal holonomy.}
        Let $\gamma\in L_0L_0\frg$. For $\exp(t\gamma)\in L_0L_0\sfG$ with $t$ sufficiently small, we can write the holonomy as
        \begin{equation}\label{eq:holonomy}
            \sfhol(\exp(t\gamma))\ =\ \exp\left(-\int_{D_{\exp(t\gamma)}}\omega\right)\ =\ \exp\left(-\int_{\tilde D_{t\gamma}}\tilde\omega\right),
        \end{equation}
        where $\tilde D_{t\gamma}$ is a disk in $L_0\frg$ with boundary $t\gamma$ such that $\exp(\tilde D_{t\gamma})=D_{\exp t\gamma}$, and 
        \begin{equation}
            \tilde\omega\ \coloneqq\ \exp^*\omega\ =\ \frac{\rmi}{4\pi}\int_0^1\rmd r\innerLarge{\theta_{V(r)}}{\parder[\theta_{V(r)}]{r}}\ \in\ \Omega^2(L_0\frg)~,
        \end{equation}
        where
        \begin{equation}
            \theta_V\ \coloneqq\ \frac{1-\rme^{-\ad_V}}{\ad_V}\rmd V
        \end{equation}
        and $V$ the identity function $V\mapsto V$ on $L_0\frg$. If we now expand~\eqref{eq:holonomy} in $t$ to the first non-vanishing order, we find that
        \begin{equation}
            \begin{aligned}
                \sfhol(\exp(t\gamma))\ &=\ 1-\frac{\rmi}{4\pi}\int_{\tilde D_{t\gamma}}\int_0^1\rmd r\innerLarge{V(r)}{\parder[V(r)]{r}}+\caO(t^3)
                \\
                &=\ 1+\frac{\rmi}{4\pi}\int_{t\gamma}\int_0^1\rmd r\innerLarge{\rmd V(r)}{\parder[V(r)]{r}}+\caO(t^3)
                \\
                &=\ 1+\frac{\rmi t^2}{4\pi}\int_0^1\rmd s\int_0^1\rmd r\innerLarge{\parder[\gamma(s,r)]{s}}{\parder[\gamma(s,r)]{r}}+\caO(t^3)~.
            \end{aligned} 
        \end{equation}
        In the second step, we have used Stokes' theorem and then parametrised the path in $L_0\frg$ as $s\mapsto t\gamma(s) $. As a corollary, we immediately obtain
        \begin{equation}\label{eq:holder}
            \dder{t}\bigg|_0\sfhol(\exp(t\gamma))\ =\ 0~,
        \end{equation}
        which is~\eqref{eq:derhol}.\footnote{A more general statement of this result can be found as Lemma 7.1 in~\cite{Mackaay:2000ac} or in~\cite{Barrett:1991aj}.}
        
        Finally, utilising the corollary we derive a useful formula. Consider a path $ \ell $ in $ L_0 L_0 \sfG $ parametrised by $ t $ and beginning at $ \ell_0$. Then, we have 
        \begin{equation}\label{eq:holderhol}
            \begin{aligned}
                &\sfhol(\ell_0)\dder{t}\bigg|_0\sfhol^{-1}(\ell(t)) 
                \\
                &\kern.5cm=\ -\frac{\rmi}{2\pi}\int_0^1\rmd s\int_0^1\rmd r\innerLarge{\parder{s}\left(\dder{t}\bigg|_0\ell(t,s,r)\ell_0^{-1}(s,r)\right)}{\parder[\ell_0(s,r)]{r}\ell_0^{-1}(s,r)}
                \\
                &\kern.5cm=\ -\frac{\rmi}{2\pi}\int_0^1\rmd s\int_0^1\rmd r\innerLarge{\ell_0^{-1}(s,r)\parder[\ell_0(s,r)]{s}}{\parder{r}\left(\ell_0^{-1}(s,r)\dder{t}\bigg|_0\ell(t,s,r)\right)}.
            \end{aligned}
        \end{equation}
        This formula follows from
        \begin{equation}
            \dder{t}\bigg|_0\big(\sfhol(\ell_0)\,\sfhol^{-1}(\ell(t))\,c^{-1}(\ell_0,\ell_0^{-1})\,c(\ell_0^{-1},\ell(t))\big)\ =\ \dder{t}\bigg|_0\sfhol^{-1}(\ell_0^{-1}\ell(t))\ =\ 0~,
        \end{equation}
        where the first equality is obtained using the fact $ (\ell_0, \sfhol^{-1}(\ell_0))^{-1}\cdot (\ell(t),\sfhol^{-1}(\ell(t))) \in \sfN $, and the last equality is just~\eqref{eq:holder}.
        
        \paragraph{Connection and curvature on $\widehat{L_0\sfG}$.}
        We start by giving a formula for the connection 1-form on $ P_0L_0\sfG\times\sfU(1)$ for the 2-form curvature $\omega$ as given in~\eqref{eq:2FormCurvatureKacMoody}. As $ \omega \in \Omega^2(L_0\sfG,\fru(1))$ is closed and $P_0L_0\sfG$ is contractible, we can use the contracting homotopy $\Phi_t:P_0L_0\sfG\rightarrow P_0L_0\sfG$ for all $t\in[0,\infty)$ given by $(\Phi_t f)(s,r)\coloneqq f(\rme^{-t}s,r)$ for all $f\in P_0L_0\sfG$. Then, we set
        \begin{equation}
            \hat\mu\ \coloneqq\ \mu'-\int_0^\infty\rmd t\,\Phi^*_t(\iota_\xi\pi_1^*\boundary^*\omega)\ \in\ \Omega^1(P_0L_0\sfG\times\sfU(1),\fru(1))~,
        \end{equation}
        where $\mu'$ is the pull-back to $P_0L_0\sfG\times\sfU(1)$ of the (imaginary) left-invariant Maurer--Cartan form on $\sfU(1)$, $\xi$ is the vector field of the flow $\Phi_t$, $\iota_\xi$ is the contraction by $\xi$, $ \pi_1 $ is the projection onto the first component from $ P_0L_0\sfG \times \sfU(1)$, and $\boundary$ is the endpoint evaluation map $P_0L_0\sfG\rightarrow L_0\sfG$. Explicitly, for all $(f,z)\in P_0L_0\sfG\times\sfU(1)$ we obtain
        \begin{equation}\label{eq:hatmu}
            \hat{\mu}_{(f,z)}\ =\ z^{-1}\rmd z+\frac{\rmi}{2\pi}\int_0^1\rmd s\int_0^1\rmd r\,\innerLarge{f^{-1}(s,r)\parder[f(s,r)]{s}}{\parder{r}(f^{-1}(s,r)\rmd f(s,r))}.
        \end{equation}
        Using the identities
        \begin{subequations}
            \begin{equation}
                \rmd\left(f^{-1}(s,r)\parder[f(s,r)]{s}\right)\ =\ \parder{s}(f^{-1}(s,r)\rmd f(s,r))+\left[f^{-1}(s,r)\parder[f(s,r)]{s},f^{-1}(s,r)\rmd f(s,r)\right],    
            \end{equation}
            \begin{equation}
                \begin{aligned}
                    &\rmd\innerLarge{f^{-1}(s,r)\parder[f(s,r)]{s}}{\parder{r}(f^{-1}(s,r)\rmd f(s,r))}
                    \\
                    &\kern1cm=\ \innerLarge{\parder{s}(f^{-1}(s,r)\rmd f(s,r))}{\parder{r}(f^{-1}(s,r)\rmd f(s,r))},
                \end{aligned}
            \end{equation}
            and
            \begin{equation}
                \begin{aligned}
                    &\int_0^1\rmd s\int_0^1\rmd r\innerLarge{\parder{s}(f^{-1}(s,r)\rmd f(s,r))}{\parder{r}(f^{-1}(s,r)\rmd f(s,r))}
                    \\
                    &\kern2cm=\ \frac{1}{2}\int_0^1\rmd r\innerLarge{f^{-1}(1,r)\rmd f(1,r)}{\parder{r}(f^{-1}(1,r)\rmd f(1,r))},
                \end{aligned}
            \end{equation}
        \end{subequations}
        we then find that 
        \begin{equation}\label{eq:dofmu}
            \rmd\hat\mu\ =\ \pi_1^*\boundary^*\omega~.
        \end{equation}
        Finally, we wish to verify that $\hat\mu$ descends to the quotient space $(P_0L_0\sfG\times\sfU(1))/\sfN\cong\widehat{L_0\sfG}$. This requires checking that it does not depend on elements in the subgroup $\sfN$. Indeed, for all $\hat h\in P_0L_0\sfG\times\sfU(1)$ and for all $\hat n\in\sfN$, it follows that
        \begin{equation}\label{eq:muinvar}
            R^*_{\hat n}\hat\mu_{\hat h\hat n}\ =\ \hat\mu_{\hat h}\ =\ L^*_{\hat n}\hat\mu_{\hat n \hat h}~,
        \end{equation}
        where $R_{\hat n}$ and $ L_{\hat n} $ are the right and left multiplications by $\hat n$, respectively. To see this, one uses formula~\eqref{eq:holderhol}. Alternatively and more explicitly, since the Lie derivative generates pullbacks\footnote{Recall that $\Phi_t^*\alpha=\rme^{t\caL_V}\alpha$ for $\Phi_t$ a flow generated by some vector field $V$ and a form $\alpha$.} it is enough to check that 
        \begin{equation}
            \caL_{V^{R,L}}\hat \mu\ =\ 0 ~,
        \end{equation}
        where $V^R$ and $V^L$ are the right and left fundamental vector fields corresponding to the action by right and left multiplications by elements in $\sfN$, respectively. Using Cartan's formula $\caL_V=\rmd\iota_V+\iota_V\rmd$ and the fact that $\iota_V\rmd\hat\mu =\iota_V\pi_1^*\boundary^*\omega=0$, which follows from the verticality of $V^{R,L}$ with respect to $\boundary$, we only need to check that
        \begin{equation}\label{eq:hatmuX}
            \iota_V\hat\mu\ =\ 0~.
        \end{equation}
        For this, consider a path with $L(t)=(\ell(t),\sfhol^{-1}(\ell(t)))\in\sfN$ which begins at the identity $L(0)=(\unit_{P_0L_0\sfG},\unit_{\sfU(1)})$. This implies that $\dot\ell(0)\coloneqq\dder{t}\big|_0\ell(t)\in L_0L_0\frg$. At a point $F=(f,z)\in P_0L_0\sfG$, the fundamental vector field corresponding to right multiplication by $\sfN$ has the explicit form
        \begin{equation}
            V^R_F\ =\ \dder{t}\bigg|_0F\cdot L(t)\ =\ \left(f\dot\ell(0),z\dder{t}\bigg|_0c(f,\ell(t))\right),
        \end{equation}
        where we have used~\eqref{eq:holder} in the second equality. Contracting the vector field with~\eqref{eq:hatmu}, we find
        \begin{equation}\label{eq:hatmuX_explicit}
            \hat\mu_F (X_F^R)\ =\ \dder{t}\bigg|_0 c(f,\ell(t))+\frac{\rmi}{2\pi}\int_0^1\rmd s\int_0^1\rmd r\,\innerLarge{f^{-1}(s,r)\parder[f(s,r)]{s}}{\parder[\dot \ell(0)]{r}}\ =\ 0~,
        \end{equation} 
        which follows from the explicit expression for the cocycle in~\eqref{eq:cocyclePLGUBaez} and identity~\eqref{eq:helpful_identity}. We can establish~\eqref{eq:hatmuX} for $ V^L $ analogously, or just argue by normality of $ \sfN $. 
        
        Notice that~\eqref{eq:hatmuX_explicit} (and thus~\eqref{eq:hatmuX}), which is the necessary condition for $\hat\mu$ to descend to $\widehat{L_0\sfG}$, relates the expression for the cocycle to that of the connection and through~\eqref{eq:dofmu} to that of the curvature $\omega$. In particular, the correct signs of individual terms in~\eqref{eq:hatmu} are crucial for consistency.
        
        \paragraph{Maurer--Cartan form on $\widehat{L_0\sfG}$.}
        Next, we construct the (left and right) Maurer--Cartan form on $\widehat{L_0\sfG}$. We start with the Maurer--Cartan form on $P_0L_0\sfG\times\sfU(1)$. Consider a path $\hat h(t)=(f(t),z(t))$ in $P_0L_0\sfG\times\sfU(1)$ for all $t\in[0,1]$ such that $\dder{t}\big|_0\hat h(t)\in T_{\hat h(0)}(P_0L_0\sfG\times\sfU(1))$ is its tangent vector at $\hat h(0)$. The left-invariant Maurer--Cartan form $\theta^L$ is then given by its action on $\dder{t}\big|_0\hat h(t)$,
        \begin{equation}
            \hat\theta^L_{\hat h(0)}\left(\dder{t}\bigg|_0\hat h(t)\right)\ \coloneqq\ \hat h^{-1}(0)\dder{t}\bigg|_0\hat h(t)~,
        \end{equation}
        so that
        \begin{equation}\label{eq:mcl}
            \hat\theta^L_{(f,z)}\ =\ \big(\theta^L_f,\hat\mu_{(f,z)}\big)
        \end{equation}
        for all $(f,z)\in P_0L_0\sfG\times\sfU(1)$ and where $\hat\mu$ was given in~\eqref{eq:hatmu}. Likewise, one can compute the right-invariant Maurer--Cartan form 
        \begin{equation}\label{eq:mcr}
            \hat\theta^R_{(f,z)}\ =\ \left(\theta^R_f,\hat\mu_{(f,z)}+\frac{\rmi}{2\pi}\int_0^1\rmd r\innerLarge{f^{-1}(1,r)\parder[\,f(1,r)]{r}}{\theta^L_{f(1,r)}}\right),
        \end{equation}
        where $\theta^R$ is the right-invariant Maurer--Cartan form on $P_0L_0\sfG$. Due to the identity~\eqref{eq:dofmu}, it is easy to verify the Maurer--Cartan equations
        \begin{equation}
            \rmd\hat\theta^{L,R}\ =\ \mp\tfrac12[\hat\theta^{L,R},\hat\theta^{L,R}]~,
        \end{equation}
        where the upper/lower sign is for $ L/R $. As shown previously, $\hat\mu$ descends to $\widehat{L_0\sfG}$. This, in turn, implies that also $\hat\theta^L$ and $\hat\theta^R $ descend to $\widehat{L_0\sfG}$. 
        
        \paragraph{A formula for $\hat h\nabla\hat h^{-1}$.} 
        If we want to explicitly work out the cocycle and coboundary conditions~\eqref{eq:adjustedCocycleConditions},~\eqref{eq:adjustedCoboundaryConditions}, and~\eqref{eq:adjustedHigherGaugeTransformations} for a $ \sfString(\sfG) $-bundle, we have to make sense of the expression
        \begin{equation}
            \hat h\nabla\hat h^{-1}\ =\ \hat h\rmd\hat h^{-1}+ \hat h(A\acton\hat h^{-1})~,
        \end{equation} 
        where $\hat h=[(f,z)]\in (P_0L_0\sfG\times\sfU(1))/\sfN$ and $A$ is a 1-form valued in $P_0\frg$. We have already identified the first part as $\hat h\rmd\hat h^{-1}=-\hat \theta^R_{\hat h}$ in the previous paragraph. The second part is obtained by differentiating at the identity of $P_0\sfG$ the formula
        \begin{equation}
            \begin{aligned}
                &\hat h(g\acton\hat h^{-1})
                \\
                &\kern5pt\ =\  [(f,z)](g\acton[(f,z)]^{-1})
                \\
                &\kern5pt\ =\ \bigg[\bigg(fgf^{-1}g^{-1},c^{-1}(f,f^{-1})c(f,gf^{-1}g^{-1})
                \\
                &\kern3cm\times\exp\left(-\frac{\rmi}{2\pi}\int_0^1\rmd s\int_0^1\rmd r\innerLarge{g^{-1}(r)\parder[g(r)]{r}}{\parder[f(s,r)]{s}f^{-1}(s,r)}\right)\bigg)\bigg],
            \end{aligned}
        \end{equation}
        where $g\in P_0\sfG$, and which is derived using~\eqref{eq:actionOfP0G} and the group multiplication in $P_0L_0\sfG\times \sfU(1)$. Setting $ g=1+A $ in the above formula, using the fact that $\parder{s}A=0$, and dividing by the ideal $\frn$, we find
        \begin{equation}\label{eq:hAhinv}
            \hat h (A\acton\hat h^{-1})\ =\ \left(\boundary(f)[A,\boundary(f^{-1})],\frac{\rmi}{2\pi}\int_0^1\rmd r\innerLarge{\boundary(f^{-1}(s,r))\parder[\,\boundary(f(s,r))]{r}}{A(r)}\right).
        \end{equation}
        
    \end{body}
    

\begin{thebibliography}{100}

\bibitem{Murray:9407015}
M.~K.~Murray,
{\em Bundle gerbes,}
\href{https://dx.doi.org/10.1112/jlms/54.2.403}{J. Lond. Math. Soc. {\bf 54}
  (1996) 403} [{\tt
  \href{https://www.arxiv.org/abs/dg-ga/9407015}{dg-ga/9407015}}].

\bibitem{Murray:2007ps}
M.~K.~Murray,
{\em {An introduction to bundle gerbes},}
in: ``The many facets of geometry: A tribute to Nigel Hitchin,'' eds.\ O.\
  Garc{\'i}a--Prada, J.-P.\ Bourguignon and S.\ Salamon, Oxford University
  Press, Oxford, 2010
[\href{https://dx.doi.org/10.1093/acprof:oso/9780199534920.003.0012}{doi}]
[{\tt \href{https://www.arxiv.org/abs/0712.1651}{0712.1651 [math.DG]}}].

\bibitem{Hitchin:1999fh}
N.~J.~Hitchin,
{\em {Lectures on special Lagrangian submanifolds},}
{\tt \href{https://www.arxiv.org/abs/math.DG/9907034}{math.DG/9907034}}.

\bibitem{Chatterjee:1998}
D.~S.~Chatterjee,
{\em On the construction of abelian gerbs,} PhD thesis, Trinity College
  Cambridge (1998)
available
  \href{https://ethos.bl.uk/OrderDetails.do?did=1&uin=uk.bl.ethos.286657}{online}.

\bibitem{Bergshoeff:1981um}
E.~Bergshoeff, M.~de~Roo, B.~de~Wit, and P.~van~Nieuwenhuizen,
{\em {Ten-dimensional Maxwell--Einstein supergravity, its currents, and the
  issue of its auxiliary fields},}
\href{https://dx.doi.org/10.1016/0550-3213(82)90050-5}{Nucl. Phys. B {\bf 195}
  (1982)~97}.
%%CITATION = NUPHA,B195,97;%%

\bibitem{Chapline:1982ww}
G.~F.~Chapline and N.~S.~Manton,
{\em {Unification of Yang--Mills theory and supergravity in ten dimensions},}
\href{https://dx.doi.org/10.1016/0370-2693(83)90633-0}{Phys. Lett. B {\bf 120}
  (1983) 105}.
%%CITATION = PHLTA,B120,105;%%

\bibitem{Nikolaus:2018qop}
T.~Nikolaus and K.~Waldorf,
{\em Higher geometry for non-geometric T-duals,}
\href{https://dx.doi.org/10.1007/s00220-019-03496-3}{Commun. Math. Phys. {\bf
  374}  (2019) 317} [{\tt
  \href{https://www.arxiv.org/abs/1804.00677}{1804.00677 [math.AT]}}].
%%CITATION = ARXIV:1804.00677;%%

\bibitem{Kim:2022opr}
H.~Kim and C.~Saemann,
{\em Non-geometric T-duality as higher groupoid bundles with connections,}
{\tt \href{https://www.arxiv.org/abs/2204.01783}{2204.01783 [hep-th]}}.

\bibitem{Saemann:2017zpd}
C.~Saemann and L.~Schmidt,
{\em {Towards an M5-brane model I: A 6d superconformal field theory},}
\href{https://dx.doi.org/10.1063/1.5026545}{J. Math. Phys. {\bf 59}  (2018)
  043502} [{\tt \href{https://www.arxiv.org/abs/1712.06623}{1712.06623
  [hep-th]}}].
%%CITATION = ARXIV:1712.06623;%%

\bibitem{Saemann:2019dsl}
C.~Saemann and L.~Schmidt,
{\em {Towards an {M}5-brane model {II}: Metric string structures},}
\href{https://dx.doi.org/10.1002/prop.202000051}{Fortschr. Phys. {\bf 68}
  (2020) 2000051} [{\tt \href{https://www.arxiv.org/abs/1908.08086}{1908.08086
  [hep-th]}}].
%%CITATION = ARXIV:1908.08086;%%

\bibitem{Rist:2020uaa}
D.~Rist, C.~Saemann, and M.~van~der~Worp,
{\em Towards an {M}5-brane model {III}: {S}elf-duality from additional trivial
  fields,}
\href{https://dx.doi.org/10.1007/JHEP06(2021)036}{JHEP {\bf 2106}  (2021) 036}
  [{\tt \href{https://www.arxiv.org/abs/2012.09253}{2012.09253 [hep-th]}}].

\bibitem{Fiorenza:2020hiq}
D.~Fiorenza, H.~Sati, and U.~Schreiber,
{\em {Twisted cohomotopy implies twisted string structure on M5-branes},}
\href{https://dx.doi.org/10.1063/5.0037786}{J. Math. Phys. {\bf 62}  (2021)
  042301} [{\tt \href{https://www.arxiv.org/abs/2002.11093}{2002.11093
  [hep-th]}}].

\bibitem{Jurco:2019woz}
B.~Jur\v{c}o, U.~Schreiber, C.~Saemann, and M.~Wolf,
{\em {Higher structures in M-theory},}
in: ``Higher Structures in M-Theory,'' proceedings of the
  \href{http://www.maths.dur.ac.uk/lms/109/index.html}{LMS/EPSRC Durham
  Symposium}, 12-18 August 2018
[{\tt \href{https://www.arxiv.org/abs/1903.02807}{1903.02807 [hep-th]}}].
%%CITATION = ARXIV:1903.02807;%%

\bibitem{Saemann:2012uq}
C.~Saemann and M.~Wolf,
{\em {Non-abelian tensor multiplet equations from twistor space},}
\href{https://dx.doi.org/10.1007/s00220-014-2022-0}{Commun. Math. Phys. {\bf
  328}  (2014) 527} [{\tt \href{https://www.arxiv.org/abs/1205.3108}{1205.3108
  [hep-th]}}].
%%CITATION = ARXIV:1205.3108;%%

\bibitem{Saemann:2013pca}
C.~Saemann and M.~Wolf,
{\em {Six-dimensional superconformal field theories from principal 3-bundles
  over twistor space},}
\href{https://dx.doi.org/10.1007/s11005-014-0704-3}{Lett. Math. Phys. {\bf 104}
   (2014) 1147} [{\tt \href{https://www.arxiv.org/abs/1305.4870}{1305.4870
  [hep-th]}}].
%%CITATION = ARXIV:1305.4870;%%

\bibitem{Jurco:2014mva}
B.~Jur\v{c}o, C.~Saemann, and M.~Wolf,
{\em {Semistrict higher gauge theory},}
\href{https://dx.doi.org/10.1007/JHEP04(2015)087}{JHEP {\bf 1504}  (2015) 087}
  [{\tt \href{https://www.arxiv.org/abs/1403.7185}{1403.7185 [hep-th]}}].
%%CITATION = ARXIV:1403.7185;%%

\bibitem{Jurco:2016qwv}
B.~Jur\v{c}o, C.~Saemann, and M.~Wolf,
{\em {Higher groupoid bundles, higher spaces, and self-dual tensor field
  equations},}
\href{https://dx.doi.org/10.1002/prop.201600031}{Fortschr. Phys. {\bf 64}
  (2016) 674} [{\tt \href{https://www.arxiv.org/abs/1604.01639}{1604.01639
  [hep-th]}}].
%%CITATION = ARXIV:1604.01639;%%

\bibitem{Samann:2017dah}
C.~Saemann and M.~Wolf,
{\em {Supersymmetric {Y}ang--{M}ills theory as higher {C}hern--{S}imons
  theory},}
\href{https://dx.doi.org/10.1007/JHEP07(2017)111}{JHEP {\bf 1707}  (2017) 111}
  [{\tt \href{https://www.arxiv.org/abs/1702.04160}{1702.04160 [hep-th]}}].
%%CITATION = ARXIV:1702.04160;%%

\bibitem{Nikolaus:2011ag}
T.~Nikolaus and K.~Waldorf,
{\em {Four equivalent versions of non-abelian gerbes},}
\href{https://dx.doi.org/10.2140/pjm.2013.264.355}{Pacific J. Math. {\bf 246}
  (2013) 355} [{\tt \href{https://www.arxiv.org/abs/1103.4815}{1103.4815
  [math.DG]}}].
%%CITATION = ARXIV:1103.4815;%%

\bibitem{Breen:math0106083}
L.~Breen and W.~Messing,
{\em Differential geometry of gerbes,}
\href{https://dx.doi.org/10.1016/j.aim.2005.06.014}{Adv. Math. {\bf 198}
  (2005) 732} [{\tt
  \href{https://www.arxiv.org/abs/math.AG/0106083}{math.AG/0106083}}].

\bibitem{Aschieri:2003mw}
P.~Aschieri, L.~Cantini, and B.~Jur\v{c}o,
{\em Nonabelian bundle gerbes, their differential geometry and gauge theory,}
\href{https://dx.doi.org/10.1007/s00220-004-1220-6}{Commun. Math. Phys. {\bf
  254}  (2005) 367} [{\tt
  \href{https://www.arxiv.org/abs/hep-th/0312154}{hep-th/0312154}}].
%%CITATION = HEP-TH/0312154;%%

\bibitem{Sati:2008eg}
H.~Sati, U.~Schreiber, and J.~Stasheff,
{\em $L_\infty$-algebra connections and applications to String- and
  Chern--Simons $n$-transport,}
in: ``Quantum Field Theory,'' eds. B. Fauser, J. Tolksdorf and E. Zeidler, p.
  303, Birkh{\"a}user 2009
[\href{https://dx.doi.org/10.1007/978-3-7643-8736-5_17}{doi}]
[{\tt \href{https://www.arxiv.org/abs/0801.3480}{0801.3480 [math.DG]}}].

\bibitem{Baez:0511710}
J.~C.~Baez and U.~Schreiber,
{\em Higher gauge theory,}
\href{https://dx.doi.org/10.1090/conm/431/08264}{Contemp. Math. {\bf 431}
  (2007)~7} [{\tt
  \href{https://www.arxiv.org/abs/math.DG/0511710}{math.DG/0511710}}].

\bibitem{Schreiber:0705.0452}
U.~Schreiber and K.~Waldorf,
{\em Parallel transport and functors,}
\href{https://tcms.org.ge/Journals/JHRS/xvolumes/2009/n1a10/v4n1a10hl.pdf}{J.
  Homot. Relat. Struct. {\bf 4}  (2009) 187} [{\tt
  \href{https://www.arxiv.org/abs/0705.0452}{0705.0452 [math.DG]}}].

\bibitem{Schreiber:2008aa}
U.~Schreiber and K.~Waldorf,
{\em Connections on non-abelian gerbes and their holonomy,}
\href{http://www.tac.mta.ca/tac/volumes/28/17/28-17.pdf}{Th. Appl. Cat. {\bf
  28}  (2013) 476} [{\tt \href{https://www.arxiv.org/abs/0808.1923}{0808.1923
  [math.DG]}}].

\bibitem{Gastel:2018joi}
A.~Gastel,
{\em {Canonical gauges in higher gauge theory},}
\href{https://dx.doi.org/10.1007/s00220-019-03530-4}{Commun. Math. Phys. {\bf
  376}  (2019) 1053} [{\tt
  \href{https://www.arxiv.org/abs/1810.06278}{1810.06278 [math-ph]}}].

\bibitem{Sati:2009ic}
H.~Sati, U.~Schreiber, and J.~Stasheff,
{\em {Differential twisted string and fivebrane structures},}
\href{https://dx.doi.org/10.1007/s00220-012-1510-3}{Commun. Math. Phys. {\bf
  315}  (2012) 169} [{\tt \href{https://www.arxiv.org/abs/0910.4001}{0910.4001
  [math.AT]}}].
%%CITATION = ARXIV:0910.4001;%%

\bibitem{Fiorenza:2010mh}
D.~Fiorenza, U.~Schreiber, and J.~Stasheff,
{\em {\v{C}ech cocycles for differential characteristic classes -- An
  infinity-Lie theoretic construction},}
\href{https://dx.doi.org/10.4310/ATMP.2012.v16.n1.a5}{Adv. Theor. Math. Phys.
  {\bf 16}  (2012) 149} [{\tt
  \href{https://www.arxiv.org/abs/1011.4735}{1011.4735 [math.AT]}}].
%%CITATION = ARXIV:1011.4735;%%

\bibitem{Kim:2019owc}
H.~Kim and C.~Saemann,
{\em Adjusted parallel transport for higher gauge theories,}
\href{https://dx.doi.org/10.1088/1751-8121/ab8ef2}{J. Phys. A {\bf 52}  (2020)
  445206} [{\tt \href{https://www.arxiv.org/abs/1911.06390}{1911.06390
  [hep-th]}}].

\bibitem{Borsten:2021ljb}
L.~Borsten, H.~Kim, and C.~Saemann,
{\em $EL_\infty$-algebras, generalized geometry, and tensor hierarchies,}
{\tt \href{https://www.arxiv.org/abs/2106.00108}{2106.00108 [hep-th]}}.

\bibitem{mclaughlin1992orientation}
D.~McLaughlin,
{\em Orientation and string structures on loop space,}
\href{https://msp.org/pjm/1992/155-1/pjm-v155-n1-p08-p.pdf}{Pacific J. Math.
  {\bf 155}  (1992) 143}.

\bibitem{Nikolaus:2011zg}
T.~Nikolaus, C.~Sachse, and C.~Wockel,
{\em {A smooth model for the string group},}
\href{https://dx.doi.org/10.1093/imrn/rns154}{Int. Math. Res. Notices. {\bf
  2013}  (2013) 3678} [{\tt
  \href{https://www.arxiv.org/abs/1104.4288}{1104.4288 [math.AT]}}].
%%CITATION = ARXIV:1104.4288;%%

\bibitem{Baez:2005sn}
J.~C.~Baez, D.~Stevenson, A.~S.~Crans, and U.~Schreiber,
{\em {From loop groups to 2-groups},}
\href{http://projecteuclid.org/euclid.hha/1201127333}{Homol. Homot. Appl. {\bf
  9}  (2007) 101} [{\tt
  \href{https://www.arxiv.org/abs/math.QA/0504123}{math.QA/0504123}}].
%%CITATION = MATH/0504123;%%

\bibitem{Schommer-Pries:0911.2483}
C.~Schommer{--}Pries,
{\em Central extensions of smooth 2-groups and a finite-dimensional string
  2-group,}
\href{https://dx.doi.org/10.2140/gt.2011.15.609}{Geom. Top. {\bf 15}  (2011)
  609} [{\tt \href{https://www.arxiv.org/abs/0911.2483}{0911.2483 [math.AT]}}].

\bibitem{Demessie:2016ieh}
G.~A.~Demessie and C.~Saemann,
{\em {Higher gauge theory with string 2-groups},}
\href{https://dx.doi.org/10.4310/ATMP.2017.v21.n8.a2}{Adv. Theor. Math. Phys.
  {\bf 21}  (2017) 1895} [{\tt
  \href{https://www.arxiv.org/abs/1602.03441}{1602.03441 [math-ph]}}].
%%CITATION = ARXIV:1602.03441;%%

\bibitem{Bagger:2007jr}
J.~Bagger and N.~D.~Lambert,
{\em Gauge symmetry and supersymmetry of multiple M2-branes,}
\href{https://dx.doi.org/10.1103/PhysRevD.77.065008}{Phys. Rev. D {\bf 77}
  (2008) 065008} [{\tt \href{https://www.arxiv.org/abs/0711.0955}{0711.0955
  [hep-th]}}].
%%CITATION = 0711.0955;%%

\bibitem{Gustavsson:2007vu}
A.~Gustavsson,
{\em Algebraic structures on parallel M2-branes,}
\href{https://dx.doi.org/10.1016/j.nuclphysb.2008.11.014}{Nucl. Phys. B {\bf
  811}  (2009)~66} [{\tt \href{https://www.arxiv.org/abs/0709.1260}{0709.1260
  [hep-th]}}].
%%CITATION = 0709.1260;%%

\bibitem{Saemann:2017rjm}
C.~Saemann and L.~Schmidt,
{\em {The non-abelian self-dual string and the (2,0)-theory},}
\href{https://dx.doi.org/10.1007/s11005-019-01250-3}{Lett. Math. Phys. {\bf
  110}  (2020) 1001} [{\tt
  \href{https://www.arxiv.org/abs/1705.02353}{1705.02353 [hep-th]}}].
%%CITATION = ARXIV:1705.02353;%%

\bibitem{Baez:0307200}
J.~C.~Baez and A.~D.~Lauda,
{\em Higher-dimensional algebra V: 2-groups,}
\href{http://www.kurims.kyoto-u.ac.jp/EMIS/journals/TAC/volumes/12/14/12-14.pdf}{Th.
  App. Cat. {\bf 12}  (2004) 423} [{\tt
  \href{https://www.arxiv.org/abs/math.QA/0307200}{math.QA/0307200}}].

\bibitem{Aldrovandi:0808.3627}
E.~Aldrovandi and B.~Noohi,
{\em Butterflies I: Morphisms of 2-group stacks,}
\href{https://dx.doi.org/10.1016/j.aim.2008.12.014}{Adv. Math. {\bf 221}
  (2009) 687} [{\tt \href{https://www.arxiv.org/abs/0808.3627}{0808.3627
  [math.CT]}}].

\bibitem{Noohi:0506313}
B.~Noohi,
{\em On weak maps between 2-groups,}
{\tt \href{https://www.arxiv.org/abs/math.CT/0506313}{math.CT/0506313}}.

\bibitem{Noohi:0910.1818}
B.~Noohi,
{\em Integrating morphisms of Lie 2-algebras,}
\href{https://dx.doi.org/10.1112/S0010437X1200067X}{Compositio Math. {\bf 149}
  (2013) 264} [{\tt \href{https://www.arxiv.org/abs/0910.1818}{0910.1818
  [math.QA]}}].

\bibitem{Jurco:2018sby}
B.~Jur\v{c}o, L.~Raspollini, C.~Saemann, and M.~Wolf,
{\em {$L_\infty$-algebras of classical field theories and the
  {B}atalin--{V}ilkovisky formalism},}
\href{https://dx.doi.org/10.1002/prop.201900025}{Fortsch. Phys. {\bf 67}
  (2019) 1900025} [{\tt \href{https://www.arxiv.org/abs/1809.09899}{1809.09899
  [hep-th]}}].
%%CITATION = ARXIV:1809.09899;%%

\bibitem{Jurco:2019bvp}
B.~Jur\v{c}o, T.~Macrelli, L.~Raspollini, C.~Saemann, and M.~Wolf,
{\em {$L_\infty$-algebras, the BV formalism, and classical fields},}
in: ``Higher Structures in M-Theory,'' proceedings of the
  \href{http://www.maths.dur.ac.uk/lms/109/index.html}{LMS/EPSRC Durham
  Symposium}, 12--18 August 2018
[{\tt \href{https://www.arxiv.org/abs/1903.02887}{1903.02887 [hep-th]}}].
%%CITATION = ARXIV:1903.02887;%%

\bibitem{Demessie:2014ewa}
G.~A.~Demessie and C.~Saemann,
{\em {Higher Poincar\'e lemma and integrability},}
\href{https://dx.doi.org/10.1063/1.4929537}{J. Math. Phys. {\bf 56}  (2015)
  082902} [{\tt \href{https://www.arxiv.org/abs/1406.5342}{1406.5342
  [hep-th]}}].
%%CITATION = ARXIV:1406.5342;%%

\bibitem{Ho:2012nt}
P.-M.~Ho and Y.~Matsuo,
{\em {Note on non-abelian two-form gauge fields},}
\href{https://dx.doi.org/10.1007/JHEP09(2012)075}{JHEP {\bf 1209}  (2012) 075}
  [{\tt \href{https://www.arxiv.org/abs/1206.5643}{1206.5643 [hep-th]}}].
%%CITATION = ARXIV:1206.5643;%%

\bibitem{Munkres:2014aa}
J.~Munkres,
{\em Topology,} Springer, 2014.

\bibitem{Dabrowski:1986en}
L.~Dabrowski and A.~Trautman,
{\em {Spinor structures in spheres and projective spaces},}
\href{https://dx.doi.org/10.1063/1.527021}{J. Math. Phys. {\bf 27}  (1986)
  2022}.

\bibitem{Fiorenza:2019usl}
D.~Fiorenza, H.~Sati, and U.~Schreiber,
{\em Twisted cohomotopy implies M-theory anomaly cancellation on 8-manifolds,}
\href{https://dx.doi.org/10.1007/s00220-020-03707-2}{Commun. Math. Phys. {\bf
  377}  (2020) 1961} [{\tt
  \href{https://www.arxiv.org/abs/1904.10207}{1904.10207 [hep-th]}}].

\bibitem{Porteous:1995eh}
I.~R.~Porteous,
{\em {Clifford algebras and the classical groups},} Cambridge University Press,
  1995.

\bibitem{Atiyah:1979iu}
M.~F.~Atiyah,
{\em Geometry of Yang--Mills fields,} Lezioni Fermiane, Pisa, 1979.

\bibitem{Popov:2015wsa}
A.~D.~Popov,
{\em Loop groups in {Y}ang--{M}ills theory,}
\href{https://dx.doi.org/10.1016/j.physletb.2015.07.041}{Phys. Lett. B {\bf
  748}  (2015) 439} [{\tt
  \href{https://www.arxiv.org/abs/1505.06634}{1505.06634 [hep-th]}}].

\bibitem{Waldorf:2009uf}
K.~Waldorf,
{\em String connections and Chern--Simons theory,}
\href{https://dx.doi.org/10.1090/S0002-9947-2013-05816-3}{Trans. Amer. Math.
  Soc. {\bf 365}  (2013) 4393} [{\tt
  \href{https://www.arxiv.org/abs/0906.0117}{0906.0117 [math.DG]}}].

\bibitem{Tellez-Dominguez:2023wwr}
R.~Tellez-Dominguez,
{\em {C}hern correspondence for higher principal bundles,}
{\tt \href{https://www.arxiv.org/abs/2310.12738}{2310.12738 [math.DG]}}.

\bibitem{Carey:1989ck}
A.~L.~Carey and M.~K.~Murray,
{\em String structures and the path fibration of a group,}
\href{https://dx.doi.org/10.1007/BF02102809}{Commun. Math. Phys. {\bf 141}
  (1991) 441}.

\bibitem{Murray:2001xq}
M.~K.~Murray and D.~Stevenson,
{\em {Higgs fields, bundle gerbes and string structures},}
\href{https://dx.doi.org/10.1007/s00220-003-0984-4}{Commun. Math. Phys. {\bf
  243}  (2003) 541} [{\tt
  \href{https://www.arxiv.org/abs/math.DG/0106179}{math.DG/0106179}}].
%%CITATION = MATH/0106179;%%

\bibitem{Pressley:1988qk}
A.~Pressley and G.~Segal,
{\em Loop groups,} Oxford Mathematical Monographs, Oxford, 1988.

\bibitem{Mickelsson:1989hp}
J.~Mickelsson,
{\em Current algebras and groups,} Plenum Press, New York, 1989
[\href{https://dx.doi.org/10.1007/978-1-4757-0295-8}{doi}].

\bibitem{Murray:1987ua}
M.~K.~Murray,
{\em Another construction of the central extension of the loop group,}
\href{https://dx.doi.org/10.1007/BF01239026}{Commun. Math. Phys. {\bf 116}
  (1988)~73}.

\bibitem{Mickelsson:1987:173-183}
J.~Mickelsson,
{\em Kac--Moody groups, topology of the Dirac determinant bundle, and
  fermionization,}
\href{https://dx.doi.org/10.1007/BF01207361}{Commun. Math. Phys. {\bf 110}
  (1987) 173}.

\bibitem{Murray:2001eu}
M.~K.~Murray and D.~Stevenson,
{\em {Yet another construction of the central extension of the loop group},}
{National Research Symposium on Geometric Analysis and Applications Canberra,
  Australia, June 26-30, 2000}
[{\tt \href{https://www.arxiv.org/abs/math.DG/0101258}{math.DG/0101258}}].
%%CITATION = MATH/0101258;%%

\bibitem{Roberts:2022wwl}
D.~M.~Roberts,
{\em Explicit string bundles,}
notes for talk at the ``Workshop on Higher Gauge Theory and Higher
  Quantization,'' Heriot--Watt University, Edinburgh, June 26/27 2014
[{\tt \href{https://www.arxiv.org/abs/2203.04544}{2203.04544 [math.DG]}}].

\bibitem{Brown:1976:296-302}
R.~Brown and C.~B.~Spencer,
{\em $\scG$-groupoids, crossed modules and the fundamental groupoid of a
  topological group,}
\href{https://dx.doi.org/10.1016/1385-7258(76)90068-8}{Indag. Math. (Proc.)
  {\bf 79}  (1976) 296}.

\bibitem{Porst:0812.1464}
S.-S.~Porst,
{\em Strict 2-groups are crossed modules,}
{\tt \href{https://www.arxiv.org/abs/0812.1464}{0812.1464 [math.CT]}}.

\bibitem{Jurco:2005qj}
B.~Jur\v{c}o,
{\em Crossed module bundle gerbes; classification, string group and
  differential geometry,}
\href{https://dx.doi.org/10.1142/S0219887811005555}{Int. J. Geom. Meth. Mod.
  Phys. 08. {\bf 2011}  (2011) 1079} [{\tt
  \href{https://www.arxiv.org/abs/math.DG/0510078}{math.DG/0510078}}].
%%CITATION = MATH/0510078;%%

\bibitem{Howe:1997ue}
P.~S.~Howe, N.~D.~Lambert, and P.~C.~West,
{\em The self-dual string soliton,}
\href{https://dx.doi.org/10.1016/S0550-3213(97)00750-5}{Nucl. Phys. B {\bf 515}
   (1998) 203} [{\tt
  \href{https://www.arxiv.org/abs/hep-th/9709014}{hep-th/9709014}}].
%%CITATION = HEP-TH/9709014;%%

\bibitem{Akyol:2012cq}
M.~Akyol and G.~Papadopoulos,
{\em {(1,0) superconformal theories in six dimensions and Killing spinor
  equations},}
\href{https://dx.doi.org/10.1007/JHEP07(2012)070}{JHEP {\bf 1207}  (2012) 070}
  [{\tt \href{https://www.arxiv.org/abs/1204.2167}{1204.2167 [hep-th]}}].
%%CITATION = ARXIV:1204.2167;%%

\bibitem{Saemann:2011nb}
C.~Saemann and M.~Wolf,
{\em {On twistors and conformal field theories from six dimensions},}
\href{https://dx.doi.org/10.1063/1.4769410}{J. Math. Phys. {\bf 54}  (2013)
  013507} [{\tt \href{https://www.arxiv.org/abs/1111.2539}{1111.2539
  [hep-th]}}].
%%CITATION = ARXIV:1111.2539;%%

\bibitem{Eastwood:1985aa}
M.~G.~Eastwood,
{\em The generalized Penrose--Ward transform,}
\href{https://dx.doi.org/10.1017/S030500410006271X}{Math. Proc. Camb. Phil.
  Soc. {\bf 97}  (1985) 165}.

\bibitem{Ward:1977ta}
R.~S.~Ward,
{\em {On self-dual gauge fields},}
\href{https://dx.doi.org/10.1016/0375-9601(77)90842-8}{Phys. Lett. A {\bf 61}
  (1977)~81}.
%%CITATION = PHLTA,A61,81;%%

\bibitem{Hitchin:1982gh}
N.~J.~Hitchin,
{\em {Monopoles and geodesics},}
\href{https://dx.doi.org/10.1007/BF01208717}{Commun. Math. Phys. {\bf 83}
  (1982) 579}.

\bibitem{Witten:2003nn}
E.~Witten,
{\em {Perturbative gauge theory as a string theory in twistor space},}
\href{https://dx.doi.org/10.1007/s00220-004-1187-3}{Commun. Math. Phys. {\bf
  252}  (2004) 189} [{\tt
  \href{https://www.arxiv.org/abs/hep-th/0312171}{hep-th/0312171}}].

\bibitem{Popov:2004rb}
A.~D.~Popov and C.~Saemann,
{\em On supertwistors, the {P}enrose--{W}ard transform and {$\caN = 4$} super
  {Y}ang--{M}ills theory,}
\href{https://dx.doi.org/10.4310/ATMP.2005.v9.n6.a2}{Adv. Theor. Math. Phys.
  {\bf 9}  (2005) 931} [{\tt
  \href{https://www.arxiv.org/abs/hep-th/0405123}{hep-th/0405123}}].
%%CITATION = HEP-TH/0405123;%%

\bibitem{Popov:2005uv}
A.~D.~Popov, C.~Saemann, and M.~Wolf,
{\em {The topological {B}-model on a mini-supertwistor space and supersymmetric
  {B}ogomolny monopole equations},}
\href{https://dx.doi.org/10.1088/1126-6708/2005/10/058}{JHEP {\bf 0510}  (2005)
  058} [{\tt \href{https://www.arxiv.org/abs/hep-th/0505161}{hep-th/0505161}}].
%%CITATION = HEP-TH/0505161;%%

\bibitem{Isenberg:1978kk}
J.~Isenberg, P.~B.~Yasskin, and P.~S.~Green,
{\em Non-self-dual gauge fields,}
\href{https://dx.doi.org/10.1016/0370-2693(78)90486-0}{Phys. Lett. B {\bf 78}
  (1978) 462}.
%%CITATION = PHLTA,B78,462;%%

\bibitem{Witten:1978xx}
E.~Witten,
{\em An interpretation of classical Yang--Mills theory,}
\href{https://dx.doi.org/10.1016/0370-2693(78)90585-3}{Phys. Lett. B {\bf 77}
  (1978) 394}.
%%CITATION = PHLTA,B77,394;%%

\bibitem{Khenkin:1980ff}
G.~Khenkin and Y.~Manin,
{\em {Twistor description of classical Yang--Mills--Dirac fields},}
\href{https://dx.doi.org/10.1016/0370-2693(80)90178-1}{Phys. Lett. B {\bf 95}
  (1980) 405}.

\bibitem{Pool:1981aa}
R.~Pool,
{\em Some applications of complex geometry to field theory,}
PhD Thesis, Rice University Texas, 1981, available
  \href{https://www.proquest.com/docview/303131865}{online}.

\bibitem{Buchdahl:1985aa}
N.~P.~Buchdahl,
{\em Analysis on analytic spaces and non-self-dual Yang--Mills fields,}
\href{https://dx.doi.org/10.1090/S0002-9947-1985-0776387-3}{Trans. Amer. Math.
  Soc. {\bf 288}  (1985) 431}.

\bibitem{Popov:2021mfl}
A.~D.~Popov,
{\em A twistor space action for {Y}ang--{M}ills theory,}
\href{https://dx.doi.org/10.1103/PhysRevD.104.026015}{Phys. Rev. D {\bf 104}
  (2021) 026015} [{\tt \href{https://www.arxiv.org/abs/2103.11840}{2103.11840
  [hep-th]}}].

\bibitem{Saemann:2005ji}
C.~Saemann,
{\em On the mini-superambitwistor space and {$\caN=8$} super {Y}ang--{M}ills
  theory,}
\href{https://dx.doi.org/10.1155/2009/784215}{Adv. Math. Phys. {\bf 2009}
  (2009) 784215} [{\tt
  \href{https://www.arxiv.org/abs/hep-th/0508137}{hep-th/0508137}}].
%%CITATION = HEP-TH/0508137;%%

\bibitem{Saemann:2012rr}
C.~Saemann, R.~Wimmer, and M.~Wolf,
{\em {A twistor description of six-dimensional $\caN=(1,1)$ super
  {Y}ang--{M}ills theory},}
\href{https://dx.doi.org/10.1007/JHEP05(2012)020}{JHEP {\bf 1205}  (2012)~20}
  [{\tt \href{https://www.arxiv.org/abs/1201.6285}{1201.6285 [hep-th]}}].
%%CITATION = ARXIV:1201.6285;%%

\bibitem{Penrose:1976js}
R.~Penrose,
{\em {Nonlinear gravitons and curved twistor theory},}
\href{https://dx.doi.org/10.1007/BF00762011}{Gen. Rel. Grav. {\bf 7}
  (1976)~31}.

\bibitem{Atiyah:1978wi}
M.~Atiyah, N.~J.~Hitchin, and I.~Singer,
{\em {Self-duality in four-dimensional Riemannian geometry},}
\href{https://dx.doi.org/10.1098/rspa.1978.0143}{Proc. Roy. Soc. Lond. A {\bf
  362}  (1978) 425}.

\bibitem{Merkulov:1992}
S.~A.~Merkulov,
{\em Supersymmetric nonlinear graviton,}
\href{https://dx.doi.org/10.1007/BF01077086}{Funct. Anal. Appl. {\bf 26}
  (1992)~72}.

\bibitem{Wolf:2007tx}
M.~Wolf,
{\em Self-dual supergravity and twistor theory,}
\href{https://dx.doi.org/10.1088/0264-9381/24/24/010}{Class. Quant. Grav. {\bf
  24}  (2007) 6287} [{\tt \href{https://www.arxiv.org/abs/0705.1422}{0705.1422
  [hep-th]}}].
%%CITATION = 0705.1422;%%

\bibitem{Mason:2007ct}
L.~J.~Mason and M.~Wolf,
{\em Twistor actions for self-dual supergravities,}
\href{https://dx.doi.org/10.1007/s00220-009-0732-5}{Commun. Math. Phys. {\bf
  288}  (2009)~97} [{\tt \href{https://www.arxiv.org/abs/0706.1941}{0706.1941
  [hep-th]}}].
%%CITATION = 0706.1941;%%

\bibitem{0264-9381-2-4-020}
C.~R.~LeBrun,
{\em Ambi-twistors and Einstein's equations,}
\href{https://dx.doi.org/10.1088/0264-9381/2/4/020}{Class. Quant. Grav. {\bf 2}
   (1985) 555}.

\bibitem{Merkulov:1992qa}
S.~A.~Merkulov,
{\em Simple supergravity, supersymmetric nonlinear gravitons and supertwistor
  theory,}
\href{https://dx.doi.org/10.1088/0264-9381/9/11/006}{Class. Quant. Grav. {\bf
  9}  (1992) 2369}.
%%CITATION = CQGRD,9,2369;%%

\bibitem{Mason:2011nw}
L.~J.~Mason, R.~A.~Reid-Edwards, and A.~Taghavi-Chabert,
{\em {Conformal field theories in six-dimensional twistor space},}
\href{https://dx.doi.org/10.1016/j.geomphys.2012.08.001}{J. Geom. Phys. {\bf
  62}  (2012) 2353} [{\tt \href{https://www.arxiv.org/abs/1111.2585}{1111.2585
  [hep-th]}}].
%%CITATION = ARXIV:1111.2585;%%

\bibitem{Wolf:2010av}
M.~Wolf,
{\em {A first course on twistors, integrability and gluon scattering
  amplitudes},}
\href{https://dx.doi.org/10.1088/1751-8113/43/39/393001}{J. Phys. A {\bf 43}
  (2010) 393001} [{\tt \href{https://www.arxiv.org/abs/1001.3871}{1001.3871
  [hep-th]}}].

\bibitem{Kobayashi:1963}
S.~Kobayashi and K.~Nomizu,
{\em Foundations of differential geometry. Volume {I},} Interscience, New York,
  1963.

\bibitem{Trautman:1977im}
A.~Trautman,
{\em {Solutions of the Maxwell and {Y}ang--{M}ills equations associated with
  Hopf fibrings},}
\href{https://dx.doi.org/10.1007/BF01811088}{Int. J. Theor. Phys. {\bf 16}
  (1977) 561}.

\bibitem{Harnad:1980ct}
J.~P.~Harnad, J.~Tafel, and S.~Shnider,
{\em {Canonical connections on Riemannian symmetric spaces and solutions to the
  {E}instein--{Y}ang--{M}ills equations},}
\href{https://dx.doi.org/10.1063/1.524658}{J. Math. Phys. {\bf 21}  (1980)
  2236}.

\bibitem{Bais:1985ns}
F.~A.~Bais and P.~Batenburg,
{\em {Y}ang--{M}ills duality in higher dimensions,}
\href{https://dx.doi.org/10.1016/0550-3213(86)90228-2}{Nucl. Phys. B {\bf 269}
  (1986) 363}.

\bibitem{Ivanova:2009yi}
T.~A.~Ivanova, O.~Lechtenfeld, A.~D.~Popov, and T.~Rahn,
{\em {Instantons and Yang--Mills flows on coset spaces},}
\href{https://dx.doi.org/10.1007/s11005-009-0336-1}{Lett. Math. Phys. {\bf 89}
  (2009) 231} [{\tt \href{https://www.arxiv.org/abs/0904.0654}{0904.0654
  [hep-th]}}].
%%CITATION = 0904.0654;%%

\bibitem{Witten:1987cg}
E.~Witten,
{\em {The index of the Dirac operator in loop space},}
in: ``Elliptic curves and modular forms in algebraic topology,'' Lectures Notes
  in Math.~vol 1326, Springer, 1988.

\bibitem{Killingback:1986rd}
T.~P.~Killingback,
{\em {World sheet anomalies and loop geometry},}
\href{https://dx.doi.org/10.1016/0550-3213(87)90229-X}{Nucl. Phys. B {\bf 288}
  (1987) 578}.
%%CITATION = NUPHA,B288,578;%%

\bibitem{Stolz:2004aa}
S.~Stolz and P.~Teichner,
{\em What is an elliptic object?,}
in ``Topology, geometry and quantum field theory,'' volume~308 of LMS Lecture
  Note Series, p.~247--343. Cambridge University Press, 2004, available
  \href{https://math.berkeley.edu/~teichner/Papers/Oxford.pdf}{online}.

\bibitem{Carey:2004xt}
A.~L.~Carey, S.~Johnson, M.~K.~Murray, D.~Stevenson, and B.-L.~Wang,
{\em Bundle gerbes for Chern--Simons and Wess--Zumino--Witten theories,}
\href{https://dx.doi.org/10.1007/s00220-005-1376-8}{Commun. Math. Phys. {\bf
  259}  (2005) 577} [{\tt
  \href{https://www.arxiv.org/abs/math.DG/0410013}{math.DG/0410013}}].
%%CITATION = MATH/0410013;%%

\bibitem{Azcarraga:2011hqa}
J.~A.~d.~Azc\'arraga and J.~M.~Izquierdo,
{\em {Lie groups, Lie algebras, cohomology and some applications in physics},}
  Cambridge University Press, 1995.

\bibitem{Tuynman:1987ij}
G.~M.~Tuynman and W.~A. J.~J.~Wiegerinck,
{\em {Central extensions and physics},}
\href{https://dx.doi.org/10.1016/0393-0440(87)90027-1}{J. Geom. Phys. {\bf 4}
  (1987) 207}.

\bibitem{neeb2002central}
K.-H.~Neeb,
{\em Central extensions of infinite-dimensional Lie groups,}
\href{http://www.numdam.org/article/AIF_2002__52_5_1365_0.pdf}{Ann. Inst.
  Fourier {\bf 52}  (2002) 1365}.

\bibitem{Mackaay:2000ac}
M.~Mackaay and R.~Picken,
{\em {The holonomy of gerbes with connections},}
\href{https://dx.doi.org/10.1006/aima.2002.2085}{Adv. Math. {\bf 170}  (2002)
  287} [{\tt
  \href{https://www.arxiv.org/abs/math.DG/0007053}{math.DG/0007053}}].
%%CITATION = MATH/0007053;%%

\bibitem{Barrett:1991aj}
J.~W.~Barrett,
{\em {Holonomy and path structures in general relativity and Yang--Mills
  theory},}
\href{https://dx.doi.org/10.1007/BF00671007}{Int. J. Theor. Phys. {\bf 30}
  (1991) 1171}.

\end{thebibliography}
\end{document}